\documentclass[11pt, hyp]{nyjm}
\usepackage{latexsym}
\usepackage{amssymb}
\usepackage{amsmath}
\usepackage{color}
\usepackage{bbm}

\usepackage{graphicx}
\usepackage{tcolorbox}
\usepackage{tabularx}

\usepackage{tikz}
\usepackage{xparse}

\usepackage{hyperref}
\hypersetup{nesting=true,debug=true,naturalnames=true}
\usepackage{graphicx,amssymb,upref}

\let\<\langle
\let\>\rangle
\newcommand{\doi}[1]{doi:\,\href{http://dx.doi.org/#1}{#1}}

\let\uml\"

\newtheorem{theorem}{Theorem}[section]
\newtheorem{lemma}[theorem]{Lemma}
\newtheorem{proposition}[theorem]{Proposition}
\newtheorem{corollary}[theorem]{Corollary}
\newtheorem{definition}[theorem]{Definition}

\newtheorem{remark}[theorem]{Remark}

\newcommand{\cl}[1]{\mathcal{#1}}
\newcommand{\bb}[1]{\mathbb{#1}}

\begin{document}

\title[Transfer of quantum game strategies]
{Transfer of quantum game strategies}

\author[G. Hoefer]{Gage Hoefer}
\address{Department of Mathematics\\ Dartmouth College\\ Hanover NH \\ 03755 \\ USA}
\email{gage.hoefer@dartmouth.edu}

\subjclass[2020]{47L07, 47L25, 81P16, 81R50}

\keywords{Non-local game, quantum no-signalling correlations, operator systems, game algebra.}



\begin{abstract} 
We develop a method for the transfer of perfect strategies between various classes of two-player, one round cooperative non-local games with quantum inputs and outputs using the simulation paradigm in quantum information theory. We show that such a transfer is possible when canonically associated operator spaces for each game are quantum homomorphic or isomorphic, as defined in \cite{gt_two}. We examine a new class of QNS correlations, needed for the transfer of strategies between games, and characterize them in terms of states on tensor products of canonical operator systems. We define jointly tracial correlations and show they correspond to traces acting on tensor products of canonical ${\rm C}^{*}$-algebras associated with individual game parties. We then make an inquiry into the initial application of such results to the study of concurrent quantum games. 
\end{abstract}

\maketitle

\tableofcontents


\section{Introduction}\label{s_intro}

In the past few decades, non-local games have been studied under a variety of names (such as Bell inequalities) and across the disciplines of physics, mathematics, and computer science; they have significant connections to areas as diverse as noncommutative geometry, quantum complexity theory, entanglement theory, and operator algebras. The latter provides a particularly fruitful framework for approaching questions of non-locality in quantum systems, as the input-output behavior of measurements on bipartite quantum systems can be encoded through noncommutative operator algebras and their state spaces. The study of such bipartite systems can therefore be translated into the study of associated operator algebras, and so numerous tools from functional analysis can be applied to help obtain answers for more physically motivated questions. 

A non-local game is (formally speaking) a tuple $\mathbb{G} = (X, Y, A, B, \lambda)$ of finite sets $X, Y, A, B$ and a function $\lambda: X\times Y\times A\times B\rightarrow \{0, 1\}$. The game is played cooperatively by two players, against a referee. Our two players--- call them Alice and Bob--- are separated spatially, and are not allowed to communicate during the game. For our purposes, the game takes place over a single round; during a round, the referee samples question pairs $(x, y) \in X\times Y$, and sends question $x$ to Alice and question $y$ to Bob. Alice and Bob must respond with answers $a \in A$ and $b \in B$, respectively. The two win the round if and only if $\lambda$ evaluates to $1$ on their question-answer pairs; that is, they win if and only if $\lambda(x, y, a, b) = 1$. 

While our two players are not allowed to communicate during the game, to improve their chances of success they can coordinate their answers according to a predetermined strategy. Our players may have access to a shared quantum entangled state, and measurements on this state by each player can improve their chance of winning by coordinating their answers to the referee in such a way that would not be possible classically (see \cite{bell}: the CHSH inequality and its use in the proof of the celebrated “Bell’s Theorem” during the 1960’s, or \cite{bct}). Specific classes of strategies--- known as correlations--- based on the use (or non-use) of shared classical and quantum resources between the two players are of particular interest in describing their behavior; such correlations are formalized using classical information channels, when considered as conditional probability distributions. Different mathematical models corresponding to each strategy type describe the outcomes of these experiments: local (corresponding to the use of classical resources), quantum (corresponding to finite-dimensional entangled resources), quantum approximate (corresponding to liminal entangled resources), quantum commuting (which arises from the commuting model of quantum mechanics), and general no-signalling (which does not necessarily rely on the use of a shared resource, but as a probabilistic strategy must still satisfy the basic constraints of the game), are the main correlation classes of interest. These are denoted as $\cl{C}_{\rm loc}, \cl{C}_{\rm q}, \cl{C}_{\rm qa}, \cl{C}_{\rm qc}$ and $\cl{C}_{\rm ns}$, respectively. 

One particular connection between the study of quantum information theory, non-locality, and operator algebras driving much of the recent development in these areas is the equivalence of Tsirelson's problem in quantum physics, and Connes' embedding problem (or CEP) in von Neumann algebra theory; this equivalence was established in \cite{fritz, jnppsw, ozawa}. The subsequent investigation of this equivalence led to the resolution of many other important questions including a refutation to the strong Tsirelson problem in \cite{slofstra_one} (see also \cite{dpp}), and a negative answer to the CEP in \cite{mipre}. Non-local games lay at the base of all of these approaches, and questions involving non-local games are the main motivation for the current work.

While classical non-local games are fruitful objects of study, as combinatorial objects with finite sets of inputs and outputs they ultimately have inherent restrictions; in an attempt to surpass some of these limitations, more attention has been given to \textit{quantum games}. These are non-local games where the inputs and outputs are allowed to be quantum states, or sometimes mixtures of classical and quantum states. In this setting, question and answer sets $X, Y, A$ and $B$ are replaced by spaces $\bb{C}^{|X|}, \bb{C}^{|Y|}, \bb{C}^{|A|}$ and $\bb{C}^{|B|}$, and strategies are implemented using quantum channels $\Phi: M_{X}\otimes M_{Y}\rightarrow M_{A}\otimes M_{B}$ instead of classical channels $\mathcal{N}: X\times Y\rightarrow A\times B$. The lack of communication between players is enforced by strictly requiring the use of \textit{quantum no-signalling (QNS)} channels, as introduced in \cite{dw}. Furthermore, the hierarchy of classical correlations is replaced by their quantum analogues, first introduced in \cite{tt} (see also \cite{bks}). The rules of the game can be generalized to the quantum context by replacing $\lambda$ with a $0$-preserving and join-preserving map between the projection lattices $\cl{P}_{XY}$ and $\cl{P}_{AB}$ of $M_{X}\otimes M_{Y}$ and $M_{A}\otimes M_{B}$, respectively. Such games have been increasingly studied over the past few years (see \cite{synch_bhtt, bhtt, cjpp, rv, tt} for a non-exhaustive list). As legitimate generalizations of classical non-local games, the hope is that by ``enlarging" the space of possible inputs and outputs (when compared to finite sets), a wealth of new examples are provided that might help shed more light on some of the previously mentioned questions in all areas. 

The purpose of the present paper is to generalize the results of \cite{gt_one} to the context of quantum non-local games; in that work, generalized homomorphisms and isomorphisms between classical non-local games were introduced. The existence of a homomorphism or isomorphism of type ${\rm t}$ from game $\bb{G}_{1}$ to game $\bb{G}_{2}$ lead to a relation between optimal game values, when playing with strategies of type ${\rm t}$; specifically, an inequality of values in the former case, and equality in the latter case. We wanted to obtain similar results for quantum games, and identify necessary conditions for when two quantum games are similar in this sense. In order to identify when two quantum games are homomorphic (respectively, isomorphic) of some type ${\rm t}$, we looked at several classes of quantum games which have canonically associated \textit{quantum hypergraphs}. These are subspaces of linear operators acting between finite-dimensional spaces which in some sense encode the properties or rules of the game. Using the notion of ${\rm t}$-homomorphism and ${\rm t}$-isomorphism of quantum hypergraphs introduced in \cite{gt_two}, we thus had a way to characterize when our quantum games were similar. 

Quantum game homomorphisms of a given type ${\rm t}$ were defined using QNS correlations of the same type, subject to additional constraints. Such conditions were added in order to allow the transport of perfect strategies of type ${\rm t}$ from the first game to perfect strategies of type ${\rm t}$ for the second. As in \cite{gt_one, gt_two}, ${\rm t}$-isomorphisms require the use of QNS bicorrelations, which were first defined in \cite{bhtt}. This process of strategy transfer employs the simulation paradigm for quantum channels (see \cite{dw}). According to the paradigm, if we start with a quantum channel $\mathcal{E}: M_{X_{1}}\rightarrow M_{A_{1}}$ from alphabet $X_{1}$ to alphabet $A_{1}$, using assistance from no-signalling resources over $(X_{2}, A_{1}, X_{1}, A_{2})$ (i.e., a QNS channel $\Gamma: M_{X_{2}}\otimes M_{A_{1}}\rightarrow M_{X_{1}}\otimes M_{A_{2}}$) we can construct a new quantum channel $\Gamma[\cl{E}]: M_{X_{2}}\rightarrow M_{A_{2}}$ dependent on $\cl{E}$ and $\Gamma$ from alphabet $X_{2}$ to $A_{2}$. A similar approach for classical channels was used in \cite{gt_one}, and one of the main focuses of this work is to show how it extends to the quantum case. 

We now describe the organization of the paper in more detail. In Section \ref{s_prelim}, we set notation and recall the definition of the main no-signalling correlation types for both classical and quantum channels, and introduce the simulation setup for quantum channels. Section \ref{s_ssom} contains the definitions of the different types of a class of stochastic (and bistochastic) operator matrices, and how these operator matrices will be used to define necessary subclasses of QNS correlations (which we call strongly quantum no-signalling) over the quadruple $(X_{2}\times Y_{2})\times (A_{1}\times B_{1})\times (X_{1}\times Y_{1})\times (A_{2}\times B_{2})$. In Section \ref{strat_transport}, the strongly quantum no-signalling correlations are used in the simulation paradigm for quantum channels to establish strategy transport for quantum games; this is achieved in Theorem \ref{strat_transp} and Theorem \ref{bi_strat_transp} for the QNS correlation and bicorrelation cases, respectively. We note that the definitions of Sections \ref{s_ssom} and \ref{sqns_correlations_section} generalize those from \cite[Section 4, 5]{gt_one}; thus, we recover many of the results (in the classical case) from \cite{gt_one} in Sections \ref{s_ssom}-\ref{strat_transport}.

Section \ref{s_perfstrats} contains some of the main results of the paper, where we restrict our attention to the transfer of perfect strategies for various types of quantum games. This application employs SQNS correlations as strategies  for the ${\rm t}$-homomorphism and isomorphism games between quantum non-local games, when utilizing the framework of generalized homomorphisms of quantum hypergraphs as introduced in \cite[Section 3]{gt_two}. A characterization of when perfect strategy transfer is possible between games when the first leg is classical while the second is quantum, along with when both are implication games are included (see \cite{tt} for relevant introductions).

The last section features the other main focus of the paper, wherein we investigate our new QNS correlations leading up to an operator system characterization for each subclass (similar to those obtained in \cite{lmprsstw, tt} and \cite[Section 7]{gt_one}). Using these results, we also focus on the transfer of strategies between \textit{concurrent} quantum games--- a class of quantum game first introduced in \cite{synch_bhtt} and further developed in \cite{bhtt, tt}. As a proposed quantization of the class of synchronous games, and analogous to the adaptation of the synchronicity condition for classical game homomorphisms which necessitated the definition of \textit{jointly synchronous} correlations in \cite[Section 8]{gt_one}, we define \textit{jointly tracial} SQNS correlations in this section. A tracial characterization of jointly tracial correlations is contained in Theorem \ref{joint_tracial_strategies_thm} (in the same vein as the characterizations established in \cite{synch_bhtt}, \cite{bhtt} and \cite{hmps}), and Theorem \ref{simulated_tracial} shows that these are the right subclass of SQNS correlations to use for transferring tracial correlations. Thus, the transfer of strategies for quantum games developed within specializes to the concurrent case. Finally, initial connections between the transfer of concurrent/tracial strategies and traces on canonically associated $*$-algebras and ${\rm C}^{*}$-algebras for concurrent games are discussed.

\subsection*{Acknowledgements}
I would like to thank my advisor Ivan G. Todorov for the many fruitful discussions on earlier drafts of the paper. I would also like to thank Lyudmila Turowska for helpful comments on some results in the last section, and Alexandros Chatzinikolaou for identifying the gap in the proof of \cite[Lemma 5.8]{gt_one}, necessitating the change in definition. Finally, I would like to thank the referee(s) for their detailed reading and remarks.


\section{Preliminaries}\label{s_prelim}
In this section we set notation, and include the necessary preliminaries on quantum no-signalling correlations to be used throughout the rest of the paper. For a finite set $X$, we let $\bb{C}^{X} = \oplus_{x \in X}\bb{C}$ and write $(e_{x})_{x \in X}$ for the canonical orthonormal basis of $\bb{C}^{X}$. Similarly, if $H$ is a Hilbert space we set $H^{X} = \oplus_{x \in X}H$. If $X$ is countable (not necessarily finite), let $\ell_{2}^{X}$ denote the Hilbert space of square summable sequences over $\bb{C}$; in the specific case that $X$ is finite, we will sometimes write $\ell_{2}^{X} = \bb{C}^{X}$. For finite $X, Y$, we will often abbreviate the Cartesian product $X\times Y$ as $XY$ for ease of use in notation. We denote by $M_{X}$ the algebra of all complex matrices of size $X\times X$, and by $\cl{D}_{X}$ its subalgebra of all diagonal matrices. We write $\epsilon_{xx'}, x, x' \in X$ for the canonical matrix units in $M_{X}$, denote by ${\rm Tr}$ the trace functional on $M_{X}$, and set $\langle S, T\rangle = {\rm Tr}(ST^{*})$, where the adjoint is with respect to the canonical orthonormal basis. We let $\Delta_{X}: M_{X}\rightarrow \cl{D}_{X}$ denote the canonical conditional expectation on the full matrix algebra. Set
\begin{gather*}
	J_{X} := \sum\limits_{x, x'}\epsilon_{xx'}\otimes \epsilon_{xx'}, \;\;\;\; J_{X}^{\rm cl} := \sum\limits_{x \in X}\epsilon_{xx}\otimes \epsilon_{xx},
\end{gather*}
\noindent and $\tilde{J}_{X} := \frac{1}{|X|}J_{X}$. If $\mathfrak{m}_{X} =\frac{1}{\sqrt{|X|}}\sum\limits_{x \in X}e_{x}\otimes e_{x}$ is the maximally entangled unit vector in $\bb{C}^{X}\otimes \bb{C}^{X}$, then $\tilde{J}_{X} = \mathfrak{m}_{X}\mathfrak{m}_{X}^{*}$ is the corresponding rank-one projection. Note that $J_{X}^{\rm cl} = \Delta_{XX}(J_{X})$, and thus we may think of $J_{X}^{\rm cl}$ as the ``classical" part of the (unnormalized) state $J_{X}$.

For a Hilbert space $H$, let $\cl{B}(H)$ be the ${\rm C}^{*}$-algebra of all bounded linear operators on $H$, and denote by $I_{H}$ the identity operator on $H$. An ${\it operator \; system}$ in $\cl{B}(H)$ is a selfadjoint linear subspace $\cl{S} \subseteq \cl{B}(H)$ such that $I_{H} \in \cl{S}$. If $\cl{A}$ is a ${\rm C}^{*}$-algebra, we denote by $\cl{A}^{\rm op}$ its \textit{opposite} ${\rm C}^{*}$-algebra. As a set, $\cl{A}^{\rm op}$ can be identified with $\cl{A}$, with $\cl{A}^{\rm op} = \{a^{\rm op}: \; a \in \cl{A}\}$; they both have the same additive, norm, and involutive structure--- their only difference is in their multiplicative structure, as we set $a^{\rm op}b^{\rm op} = (ba)^{\rm op}$, for $a^{\rm op}, b^{\rm op} \in \cl{A}^{\rm op}$ in the opposite algebra. 

Let $X$ and $A$ be finite sets. A ${\it classical \; information \; channel}$ from $X$ to $A$ is a positive trace preserving linear map $\cl{N}: \cl{D}_{X}\rightarrow \cl{D}_{A}$. If $\cl{N}$ is an information channel, setting $p(\cdot|x) = \cl{N}(\epsilon_{xx})$ for each $x \in X$ it is easy to see that $\cl{N}$ is completely determined by its corresponding family of conditional probability distributions $\{(p(a|x))_{a \in A}: \; x \in X\}$.

A ${\it quantum \; information \; channel}$ from $X$ to $A$ is a completely positive trace preserving linear map $\Phi: M_{X}\rightarrow M_{A}$. A quantum channel will be called $(X, A)$-${\it classical}$ if $\Phi = \Delta_{A}\circ \Phi \circ \Delta_{X}$. Any classical channel $\cl{N}: \cl{D}_{X}\rightarrow \cl{D}_{A}$, has a corresponding $(X, A)$-classical (quantum) channel $\Phi_{\cl{N}}: M_{X}\rightarrow M_{A}$ given by $\Phi_{\cl{N}} = \cl{N}\circ \Delta_{X}$. Conversely, any quantum channel $\Phi: M_{X}\rightarrow M_{A}$ induces a classical channel $\cl{N}_{\Phi}: \cl{D}_{X}\rightarrow \cl{D}_{A}$ given by $\Delta_{A} \circ \Phi|_{\cl{D}_{X}}$. Finally, if $\cl{E}: \cl{D}_{X}\rightarrow M_{A}$ is a (classical-to-quantum) channel, set $\Gamma_{\cl{E}} = \cl{E}\circ \Delta_{X}$, so $\Gamma_{\cl{E}}$ is a quantum channel from $M_{X}$ to $M_{A}$. 

In the remainder of this section, we recall the basic types of quantum and classical no-signalling 
correlations that will be used throughout the paper, along with establishing the simulation paradigm arising from quantum information theory.
Let $X, Y, A$ and $B$ be finite sets. 
A \textit{quantum no-signalling (QNS) correlation} \cite{dw} is a quantum channel 
$\Gamma: M_{XY}\rightarrow M_{AB}$ such that
\begin{gather}
	{\rm Tr}_{A}\Gamma(\rho_{X}\otimes \rho_{Y}) = 0 \text{ whenever } \rho_{X} \in M_{X} \text{ and } {\rm Tr}(\rho_{X}) = 0,
\end{gather}
\noindent and
\begin{gather}
	{\rm Tr}_{B}\Gamma(\rho_{X}\otimes \rho_{Y}) = 0 \text{ whenever } \rho_{Y} \in M_{Y} \text{ and } {\rm Tr}(\rho_{Y}) = 0.
\end{gather}
We set 
$$\Gamma(aa',bb'|xx',yy') = \langle \Gamma(\epsilon_{x,x'} \otimes \epsilon_{y,y'}),\epsilon_{a,a'} \otimes \epsilon_{b,b'}
\rangle;$$
thus, 
$(\Gamma(aa',bb'|xx',yy'))_{x,x',a,a'}^{y,y',b,b'}$ is the Choi matrix of $\Gamma$ (see e.g. \cite{pa}).

A \emph{stochastic operator matrix} acting on a Hilbert space $H$ 
is a positive block operator matrix $E = (E_{x, x', a, a'})_{x, x', a, a'} \in M_{XA}(\cl{B}(H))$ such that 
${\rm Tr}_{A}(E) = I_{X} \otimes I_{H}$. Stochastic operator matrix $E$ is \textit{bistochastic} \cite{bhtt} if $X = A$ and we have ${\rm Tr}_{X}(E) = I_{A}\otimes I_{H}$. For a stochastic operator matrix $E$ over $(X, A)$, set
\begin{gather*}
	E_{a, a'} = (E_{x, x', a, a'})_{x, x' \in X} \in M_{X}\otimes \cl{B}(H).
\end{gather*}
\noindent As discussed in \cite[Section 3]{tt}, stochastic operator matrices $E$ are the Choi matrices of unital completely positive maps $\Phi_{E}: M_{A}\rightarrow M_{X}\otimes \cl{B}(H)$ given by
\begin{gather}\label{stochastic_ucp}
	\Phi_{E}(\epsilon_{a, a'}) = E_{a, a'}, \;\;\;\; a, a' \in A.
\end{gather}
\noindent If we let $\Phi = \Phi_{E}$ as in (\ref{stochastic_ucp}), for any state $\sigma \in \cl{T}(H)$ we let $\Gamma_{E, \sigma}: M_{X}\rightarrow M_{A}$ be the quantum channel defined via
\begin{gather}\label{stochastic_qns}
	\Gamma_{E, \sigma}(\rho_{X}) = \Phi_{*}(\rho_{X}\otimes \sigma), \;\;\;\; \rho_{X} \in M_{X}.
\end{gather}
\noindent Note here that $\Phi_{*}: M_{X}\otimes \cl{T}(H)\rightarrow M_{A}$ is the predual of the unital completely positive map $\Phi$. If $E = (E_{x, x', a, a'})_{x, x', a, a'}$ and $F = (F_{y, y', b, b'})_{y, y', b, b'}$ are stochastic operator matrices in $M_{XA}\otimes \cl{B}(H)$ and $M_{YB}\otimes \cl{B}(H)$, respectively, such that 
$$E_{x, x', a, a'}F_{y, y', b, b'} = F_{y, y', b, b'}E_{x, x', a, a'}$$
for all $x, x' \in X, y, y' \in Y, a, a' \in A, b, b' \in B$, we let $E \cdot F$ be the (unique) stochastic operator matrix over $(XY, AB)$ (see \cite[Proposition 4.1]{tt}) defined by
$$ (E_{x, x', a, a'}F_{y, y', b, b'})_{x, x', y, y'}^{a, a', b, b'} \in M_{XYAB}\otimes \cl{B}(H).$$ 
If $\xi$ is a unit vector in Hilbert space $H$, we let $\Gamma_{E, F, \xi} = \Gamma_{E\cdot F, \xi\xi^{*}}$ where the latter is the quantum channel from $M_{XY}$ to $M_{AB}$ defined as in (\ref{stochastic_qns}). If $E \in M_{XA}\otimes \cl{B}(H_{A})$ and $F \in M_{YB}\otimes \cl{B}(H_{B})$ are stochastic operator matrices, we let $E\odot F$ denote the stochastic operator matrix $E\otimes F$, considered as an element of $M_{XY}\otimes M_{AB}\otimes \cl{B}(H_{A}\otimes H_{B})$. 

A QNS correlation $\Gamma: M_{XY}\rightarrow M_{AB}$ is called 
\emph{quantum commuting} if there exists a Hilbert space $H$, a unit vector $\xi \in H$ and stochastic operator matrices $E = (E_{x, x', a, a'})_{x, x', a, a'}$ and $F = (F_{y, y', b, b'})_{y, y', b, b'}$ on $H$ such that $E$ and $F$ are mutually commuting, with $\Gamma = \Gamma_{E, F, \xi}$. 
\emph{Quantum} QNS correlations are those for which there exist finite dimensional Hilbert spaces $H_{A}, H_{B}$, stochastic operator matrices $E \in M_{XA}\otimes \cl{B}(H_{A})$ and $F \in M_{YB}\otimes \cl{B}(H_{B})$, and a pure state $\sigma \in \cl{T}(H_{A}\otimes H_{B})$ such that $\Gamma = \Gamma_{E\odot F, \sigma}$. 
\emph{Approximately quantum} QNS correlations are the 
limits of quantum QNS correlations, while \emph{local} QNS correlations are 
the convex combinations of the form 
$\Gamma = \sum_{i=1}^k \lambda_i \Phi_i \otimes \Psi_i$, 
where $\Phi_i : M_X\to M_A$ and $\Psi_i : M_Y\to M_B$ are quantum channels, $i = 1,\dots,k$.
We write $\cl Q_{\rm qc}$ (resp. $\cl Q_{\rm qa}$, $\cl Q_{\rm q}$, $\cl Q_{\rm loc}$) for the (convex) set of all quantum commuting (resp. approximately quantum, quantum, local) QNS correlations, and note the (strict, see \cite{tt})
inclusions 
\begin{equation}\label{eq_Qinc}
\cl Q_{\rm loc}\subseteq \cl Q_{\rm q}\subseteq \cl Q_{\rm qa}\subseteq \cl Q_{\rm qc}\subseteq \cl Q_{\rm ns}.
\end{equation}
Let
$$	\cl{L}_{X, A} = \bigg\{(\lambda_{x, x', a, a'}) \in M_{XA} :  \exists \;c \in \bb{C} \;\text{s.t.} \;
	\sum_{a\in A}\lambda_{x, x', a, a} = \delta_{x, x'}c, \; x, x' \in X\bigg\},$$
and consider it as an operator subsystem of $M_{XA}$. Similarly, 
\begin{gather*}
	\cl{L}_{X} = \bigg\{(\lambda_{x, x', a, a'}) \in M_{XX} : \exists \; c \in \bb{C} \; \text{s.t.} \;
	\sum_{a \in X}\lambda_{x, x', a, a} = \delta_{x, x'}c \\ \text{ and } \sum_{x \in X}\lambda_{x, x, a, a'} = \delta_{a, a'}c, \; x, x', a, a' \in X\bigg\}
\end{gather*}
\noindent may be considered as an operator subsystem of $M_{XX}$. By \cite[Proposition 5.5, Theorem 6.2]{tt}, the elements $\Gamma$ of $\cl Q_{\rm ns}$ 
correspond canonically to elements of the tensor product  $\cl{L}_{X, A}\otimes_{\min} \cl L_{Y,B}$
(viewed as an operator subsystem of $M_{XA}\otimes M_{YB}$), and elements in $\cl{Q}_{\rm ns}^{\rm bi}$ corresponding to elements of $\cl{L}_{X}\otimes_{\rm min} \cl{L}_{Y}$ \cite[Proposition 3.6, Theorem 5.4]{bhtt}.

A \textit{classical correlation} over $(X, Y, A, B)$ is a collection
\begin{gather*}
	p = \bigg\{(p(a, b|x, y))_{a \in A, b \in B}: \; (x, y) \in X\times Y\bigg\},
\end{gather*}
\noindent where $(p(a, b|x, y))_{a \in A, b \in B}$ is a probability distribution for each $(x, y) \in X\times Y$. Given a classical correlation $p$, let $\cl{N}_{p}: \cl{D}_{XY}\rightarrow \cl{D}_{AB}$ be the classical channel given by
\begin{gather}
	\cl{N}_{p}(\rho) = \sum\limits_{x \in X, y \in Y}\sum\limits_{a \in A, b \in B}p(a, b|x, y)\langle \rho(e_{x}\otimes e_{y}), e_{x}\otimes e_{y}\rangle\epsilon_{aa}\otimes \epsilon_{bb}.
\end{gather}
If ${\rm t} \in \{\rm loc, q, qa, qc, ns\}$, let $\cl{C}_{\rm t}$ denote the collection of all classical correlations $p$ for which $\Gamma_{\cl{N}_{p}} \in \cl{Q}_{\rm t}$.

Suppose $X_{i}, Y_{i}, i = 1, 2$ are finite sets. Let $\Gamma$ be a QNS correlation over $(X_{2}, Y_{1}, X_{1}, Y_{2})$ and $\cl{E}: M_{X_{1}}\rightarrow M_{Y_{1}}$ be a quantum channel. If we write $\Gamma = \sum_{i=1}^{k}\Phi_{i}\otimes \Psi_{i}$, where $\Phi_{i}: M_{X_{2}}\rightarrow M_{X_{1}}$ and $\Psi_{i}: M_{Y_{1}}\rightarrow M_{Y_{2}}$ are linear maps for $i = 1, \hdots, k$, we let $\Gamma[\cl{E}]: M_{X_{2}}\rightarrow M_{Y_{2}}$ be the linear map defined by setting
\begin{gather}\label{qns_simulation}
	\Gamma[\cl{E}] = \sum\limits_{i=1}^{k}\Psi_{i}\circ \cl{E}\circ \Phi_{i}.
\end{gather}
\noindent It was shown in \cite{dw} that $\Gamma[\cl{E}]$ is a quantum channel, called therein the \textit{simulated channel} from $\cl{E}$ assisted by \textit{simulator} $\Gamma$. We note that (see \cite{gt_one}, \cite{gt_two}) the Choi matrix of $\Gamma[\cl{E}]$ coincides with
\begin{gather}\label{choi_simulated_eqn}
	\bigg(\sum\limits_{x_{1}, x_{1}'}\sum\limits_{y_{1}, y_{1}'}\Gamma(x_{1}x_{1}', y_{2}y_{2}'|x_{2}x_{2}', y_{1}y_{1}')\cl{E}(y_{1}y_{1}'|x_{1}x_{1}')\bigg)_{x_{2}, x_{2}', y_{2}, y_{2}'},
\end{gather}
\noindent where the internal sums range over $X_{1}$ and $Y_{1}$, for all $x_{2}, x_{2}' \in X_{2}, y_{2}, y_{2}' \in Y_{2}$. Thus, channel simulation is a process for constructing new quantum channels from given ones, with the assistance of correlations; morever, it is natural to want the simulated channel to depend on a shared resource between two parties. If Alice and Bob have access to some shared resource (for instance: shared randomness, or entanglement, etc.), their interaction with the resource yields the no-signalling correlation $\Gamma$, and simulated channel $\Gamma[\cl{E}]$ is dependent on their local operations (see \cite[Section II]{clmw}, and \cite{dw}). 
 

\section{Strongly stochastic operator matrices}\label{s_ssom}

Let $X, Y, A$ and $B$ be finite sets, and $H$ be a Hilbert space. In the sequel, to simplify notation we will abbreviate an ordered pair $(x, y) \in X\times Y$ to $xy$. A stochastic operator matrix $P = (P_{xx', yy'}^{aa', bb'})_{xx', yy', aa', bb'} \in M_{XYAB}\otimes \cl{B}(H)$ over $(XY, AB)$ will be called a ${\it strongly \; stochastic \; operator \; matrix}$ over $(X, Y, A, B)$ if ${\rm Tr}_{B}(L_{\sigma_{Y}}(P))$ (resp. ${\rm Tr}_{A}(L_{\sigma_{X}}(P))$) is a well-defined stochastic operator matrix over $(X, A)$ (resp. $(Y, B)$) and ${\rm Tr}_{B}(L_{\sigma_{Y}}(P)) = {\rm Tr}_{B}(L_{\sigma_{Y}'}(P))$ (resp. ${\rm Tr}_{A}(L_{\sigma_{X}}(P)) = {\rm Tr}_{A}(L_{\sigma_{X}'}(P))$) for each pure state $\sigma_{X}, \sigma_{X}' \in M_{X}$ and $\sigma_{Y}, \sigma_{Y}' \in M_{Y}$.

\begin{remark}
\rm In fact, by convexity and linearity of the slice map we may assume $\sigma_{X}, \sigma_{X'} \in M_{X}$ and $\sigma_{Y}, \sigma_{Y'} \in M_{Y}$ are arbitrary states.
\end{remark}

\begin{remark}\label{ns_generalization}
\rm A positive operator $P = (P_{xy, ab})_{x, y, a, b} \in \cl{D}_{XYAB}\otimes \cl{B}(H)$ is a \textit{no-signalling (NS) operator matrix} \cite{gt_one} if marginal operators
\begin{gather*}
	P_{x, a} := \sum\limits_{b \in B}P_{xy, ab}, \;\;\;\; P_{y, b} := \sum\limits_{a \in A}P_{xy, ab}
\end{gather*}
\noindent are well-defined, and $(P_{x, a})_{a \in A}, (P_{y, b})_{b \in B}$ are POVM's for every $x \in X$ and $y \in Y$. If we start with a NS operator matrix $P$, for $x, x' \in X, y, y' \in Y, a, a' \in A, b, b' \in B$ set
\begin{gather*}
	P_{xx', yy'}^{aa', bb'} = \delta_{xx'}\delta_{yy'}\delta_{aa'}\delta_{bb'}P_{xy, ab}.
\end{gather*}
\noindent Setting $\tilde{P} = (P_{xx', yy'}^{aa', bb'})_{xx', yy', aa', bb'} \in M_{XYAB}\otimes \cl{B}(H) $, one may easily check that $\tilde{P}$ is strongly stochastic over $(XY, AB)$ with ${\rm Tr}_{B}(L_{\sigma_{Y}}(\tilde{P}))$ and ${\rm Tr}_{A}(L_{\sigma_{X}}(\tilde{P}))$ \textit{classical stochastic operator matrices} (as introduced in \cite[Section 3]{tt}) for each (pure) state $\sigma_{X} \in M_{X}$ and $\sigma_{Y} \in M_{Y}$. Thus, we may think of strongly stochastic operator matrices over $(XY, AB)$ as generalizations of NS operator matrices, whose ``marginal" stochastic operators are no longer necessarily classical.
\end{remark}
A strongly stochastic operator matrix $P = (P_{xx', yy'}^{aa', bb'})_{xx', yy', aa', bb'}$ is called ${\it dilatable}$ if there exists a Hilbert space $K$, an isometry $V: H\rightarrow K$ and stochastic operator matrices $(E_{xx', aa'})_{x, x' \in X, a, a' \in A}$ (resp. $(F_{yy', bb'})_{y, y' \in Y, b, b' \in B}$) in $M_{XA}\otimes \cl{B}(K)$ (resp. $M_{YB}\otimes \cl{B}(K)$) such that $E_{xx', aa'}F_{yy', bb'} = F_{yy', bb'}E_{xx', aa'}$ and
\begin{gather}\label{dilate_eqn}
	P_{xx', yy'}^{aa', bb'} = V^{*}E_{xx', aa'}F_{yy', bb'}V, \;\;\;\; x, x' \in X, y, y' \in Y, a, a' \in A, b, b' \in B.
\end{gather} 

We will call a strongly stochastic operator matrix $P = (P_{xx', yy'}^{aa', bb'})_{xx', yy', aa', bb'}$ acting on a Hilbert space $H$ ${\it locally \; dilatable}$ if there exists a dilation of the form (\ref{dilate_eqn}), with the additional stipulation that the family $\{E_{xx', aa'}, F_{yy', bb'}: \; x, x' \in X, y, y' \in Y, a, a' \in A, b, b' \in B\}$ is commutative. We will call a strongly stochastic operator matrix $P = (P_{xx', yy'}^{aa', bb'})_{xx', yy', aa', bb'}$ acting on a Hilbert space $H$ ${\it quantum \; dilatable}$ if there exist stochastic operator matrices $(E_{xx', aa'})_{x, x' \in X, a, a' \in A}$ and $(F_{yy', bb'})_{y, y' \in Y, b, b' \in B}$ acting on finite dimensional Hilbert spaces $H_{A}$ and $H_{B}$ respectively, and an isometry $V: H\rightarrow H_{A}\otimes H_{B}$, such that
\begin{gather}\label{qdilate_eqn}
	P_{xx', yy'}^{aa', bb'} = V^{*}(E_{xx', aa'}\otimes F_{yy', bb'})V, 
\end{gather}
\noindent for $x, x' \in X, y, y' \in Y, a, a' \in A, b, b' \in B$.

\begin{proposition}\label{qloc_prop}
A QNS correlation $\Gamma$ over $(X, Y, A, B)$ belongs to $\cl{Q}_{\rm loc}$ if and only if there exists a Hilbert space $H$, a locally dilatable strongly stochastic operator matrix $P = (P_{xx', yy'}^{aa', bb'})_{xx', yy', aa', bb'}$ acting on $H$ and a unit vector $\xi \in H$ such that
\begin{gather}\label{qloc_value}
	\Gamma(aa', bb'|xx', yy') = \langle P_{xx', yy'}^{aa', bb'}\xi, \xi\rangle, 
\end{gather}
\noindent for  $x, x' \in X, y, y' \in Y, a, a' \in A, b, b' \in B$.
\end{proposition}
\begin{proof}
First, assume that $\Gamma \in \cl{Q}_{\rm loc}$ is a convex combination $\Gamma = \sum_{j=1}^{k}\lambda_{j}\Phi^{(j)}\otimes \Psi^{(j)}$, where $\Phi^{(j)}: M_{X}\rightarrow M_{A}$ and $\Psi^{(j)}: M_{Y}\rightarrow M_{B}$ are quantum channels, $j = 1, \hdots, k$. By the comment before \cite[Remark 3.2]{tt}, $\Phi^{(j)}$ (resp. $\Psi^{(j)}$) is of the form $\Gamma_{E^{(j)}, 1}$ (resp. $\Gamma_{F^{(j)}, 1}$) for some stochastic operator matrix $E^{(j)} \in M_{X}\otimes M_{A}\otimes \cl{B}(\bb{C})$ (resp. $F^{(j)} \in M_{Y}\otimes M_{B}\otimes \cl{B}(\bb{C})$) for $j = 1, \hdots, k$. For each $x, x' \in X, a, a' \in A$ define the matrix $E_{x, x'}^{a, a'} = (E_{x, x', a, a'}^{(j)})_{j=1}^{k} \in \cl{D}_{k}$. Define the corresponding matrices $F_{y, y'}^{b, b'}$ in $\cl{D}_{k}$ for each $y, y' \in Y, b, b' \in B$. We note that the family
\begin{gather*}
	\{E_{x, x'}^{a, a'}, F_{y, y'}^{b, b'}: \; x, x' \in X, a, a' \in A, y, y' \in Y, b, b' \in B\}, 
\end{gather*}
\noindent is commutative. Then, let
\begin{gather*}
	E:= (E_{x, x'}^{a, a'})_{x, x', a, a'}, \;\;\;\; F := (F_{y, y'}^{b, b'})_{y, y', b, b'}.
\end{gather*}
\noindent Note that $E$ (resp. $F$) is stochastic over $(X, A)$ (resp. $(Y, B)$). Indeed: since $E^{(j)}$ is a stochastic operator matrix in $M_{X}\otimes M_{A}\otimes \cl{B}(\bb{C})$ for each $j = 1, \hdots, k$, we know that $\sum\limits_{a}E_{x, x', a, a}^{(j)} = \delta_{x, x'}$ for each $x, x' \in X$ and $j = 1, \hdots, k$. Thus,
\begin{eqnarray*}
	{\rm Tr}_{A}(E) 
& = &
\sum\limits_{a}\sum\limits_{x, x'}\sum\limits_{j=1}^{k}E_{x, x', a, a}^{(j)}\epsilon_{x, x'}\otimes \epsilon_{jj} \\
& = &
\sum\limits_{j=1}^{k}\sum\limits_{x, x'}\sum\limits_{a}E_{x, x', a, a}^{(j)}\epsilon_{x, x'}\otimes \epsilon_{jj} \\
& = &
\sum\limits_{j=1}^{k}\sum\limits_{x, x'}\delta_{x, x'}\epsilon_{x, x'}\otimes \epsilon_{jj} \\
& = &
I_{X}\otimes I_{k}.
\end{eqnarray*}
\noindent Set $\xi = \sum\limits_{j=1}^{k}\sqrt{\lambda_{j}}e_{j}$, so $\xi$ is a unit vector in $\bb{C}^{k}$. For $x, x' \in X, y, y' \in Y, a, a' \in A, b, b' \in B$ we see
\begin{eqnarray*}
\langle E_{x, x'}^{a, a'}F_{y, y'}^{b, b'}\xi, \xi\rangle 
& = &
\sum\limits_{j=1}^{k}\lambda_{j}E_{x, x', a, a'}^{(j)}F_{y, y', b, b'}^{(j)} \\
& = &
\sum\limits_{j=1}^{k}\lambda_{j}\langle E_{x, x', a, a'}^{(j)}\otimes F_{y, y', b, b'}^{(j)}, I\rangle \\
& = &
\langle \Gamma(\epsilon_{x, x'}\otimes \epsilon_{y, y'}), \epsilon_{a, a'}\otimes \epsilon_{b, b'}\rangle \\
& = &
\Gamma(aa', bb'|xx', yy').
\end{eqnarray*}
\noindent Setting $P_{xx', yy'}^{aa', bb'} = E_{x, x'}^{a, a'}F_{y, y'}^{b, b'}$ for each $x, x' \in X, y, y' \in Y, a, a' \in A, b, b' \in B$ and letting $P = (P_{xx', yy'}^{aa', bb'})$ (where the latter matrix entries range over $X, Y, A, B$) gives us our locally dilatable strongly stochastic operator matrix satisfying (\ref{qloc_value}).

Now, assume $P = (P_{xx', yy'}^{aa', bb'})_{xx', yy', aa', bb'}$ is a locally dilatable strongly stochastic operator matrix acting on $H$ with unit vector $\xi \in H$ satisfying (\ref{qloc_value}). If we replace $H$ with the Hilbert space $K$ arising from the dilation (\ref{dilate_eqn}) of $P$ and the vector $\xi$ with $V\xi$, we may without loss of generality directly assume $P_{xx', yy'}^{aa', bb'} = E_{x, x', a, a'}F_{y, y', b, b'}, x, x' \in X, y, y' \in Y, a, a' \in A, b, b' \in B$, where the family
\begin{gather*}
	\{E_{x, x', a, a'}, F_{y, y', b, b'}: \; x, x' \in X, a, a' \in A, y, y' \in Y, b, b' \in B\}
\end{gather*}
\noindent is commutative. 

Let $\cl{A}_{X}^{A}$ (resp. $\cl{A}_{Y}^{B}$) be the abelian ${\rm C}^{*}$-algebra generated by $\{E_{x, x', a, a'}: \; x, x' \in X, a, a' \in A\}$ (resp. $\{F_{y, y', b, b'}: \; y, y' \in Y, b, b' \in B\}$). Let $s: \cl{A}_{X}^{A}\otimes_{\rm max} \cl{A}_{Y}^{B} \rightarrow \bb{C}$ be the state given by $s(S\otimes T) = \langle ST\xi, \xi\rangle$. By nuclearity of abelian ${\rm C}^{*}$-algebras, we may consider $s$ as a state on $\cl{A}_{X}^{A}\otimes_{\rm min}\cl{A}_{Y}^{B}$. Using the identification $\cl{A}_{X}^{A} = C(\Omega_{1}), \cl{A}_{Y}^{B} = C(\Omega_{2})$ for compact Hausdorff spaces $\Omega_{1}, \Omega_{2}$, we view $s$ as a Borel probability measure $\mu$ on the product topological space $\Omega_{1}\times \Omega_{2}$. For each $\omega_{1} \in \Omega_{1}$ (resp. $\omega_{2} \in \Omega_{2}$), let $\Phi(\omega_{1}): M_{X}\rightarrow M_{A}$ (resp. $\Psi(\omega_{2}): M_{Y}\rightarrow M_{B}$) be the quantum channel given by $\Phi(\omega_{1})(\epsilon_{x, x'}) = (E_{x, x', a, a'}(\omega_{1}))_{a, a'}$ (resp. $\Psi(\omega_{2})(\epsilon_{y, y'}) = (F_{y, y', b, b'}(\omega_{2}))_{b, b'}$). We then have
\begin{eqnarray*}
& &
	\langle\Gamma(\epsilon_{x, x'}\otimes \epsilon_{y, y'}), \epsilon_{a,a'}\otimes \epsilon_{b,b'}\rangle \\
& = &
\langle E_{x, x', a, a'}F_{y, y', b, b'}\xi, \xi\rangle \\
& = &
\int\limits_{\Omega_{1}\times \Omega_{2}}E_{x, x', a, a'}(\omega_{1})F_{y, y', b, b'}(\omega_{2})d\mu(\omega_{1}, \omega_{2}) \\
& = &
\bigg\langle\int\limits_{\Omega_{1}\times \Omega_{2}}\Phi(\omega_{1})(\epsilon_{x, x'})\otimes \Psi(\omega_{2})(\epsilon_{y, y'})d\mu(\omega_{1}, \omega_{2}), \epsilon_{a,a'}\otimes \epsilon_{b, b'}\bigg\rangle.
\end{eqnarray*}
\noindent If we approximate measure $\mu$ using convex combinations of product measures $\mu_{1}\times \mu_{2}$ (where $\mu_{1} \in {\rm M}(\Omega_{1}), \mu_{2} \in {\rm M}(\Omega_{2})$) we see we can approximate $\Gamma$ using channels which may be written as convex combinations $\tilde{\Gamma} = \sum_{j=1}^{k}\lambda_{j}\Phi_{j}\otimes \Psi_{j}$, where $\Phi_{j}: M_{X}\rightarrow M_{A}$ and $\Psi_{j}: M_{Y}\rightarrow M_{B}$ are quantum channels, $j = 1, \hdots, k$. By the Carathéodory Theorem and compactness, use a similar argument as in \cite[Remark 4.10]{tt} to conclude that $\Gamma$ itself is of that form. Thus, $\Gamma \in \cl{Q}_{\rm loc}$. 
\end{proof}

For certain classes of quantum games (for instance, those concerned with the behavior of quantum symmetries of quantum objects as in \cite{bhtt, daws}) we will draw our questions and answer states from the same space; thus, a particular notion of strongly stochastic operator matrix will be required in this context. We assume that $X = A$ and $Y = B$. For notational ease, we will continue to refer to $X$ and $A$ (resp. $Y$ and $B$) as distinct entities, even though they are copies of the same set. 

A positive operator $P = (P_{xx', yy'}^{aa', bb'})_{xx', yy', aa', bb'} \in M_{XYAB}\otimes \cl{B}(H)$ will be called a \textit{strongly bistochastic operator matrix} if it is a strongly stochastic operator matrix over $(XY, AB)$, with ${\rm Tr}_{B}(L_{\sigma_{Y}}(P))$ and ${\rm Tr}_{A}(L_{\sigma_{X}}(P))$ bistochastic operator matrices for each state $\sigma_{X} \in M_{X}, \sigma_{Y} \in M_{Y}$. A strongly bistochastic operator matrix $P = (P_{xx', yy'}^{aa', bb'})_{xx', yy', aa', bb'}$ is called \textit{dilatable} if there exists a Hilbert space $K$, an isometry $V: H\rightarrow K$, and bistochastic operator matrices $(E_{xx', aa'})_{x, x' \in X, a, a' \in A}$ (resp. $(F_{yy', bb'})_{y, y' \in Y, b, b' \in B}$) in $M_{XA}\otimes \cl{B}(K)$ (resp. $M_{YB}\otimes \cl{B}(K)$) such that $E_{xx', aa'}F_{yy', bb'} = F_{yy', bb'}E_{xx', aa'}$ and which satisfy relations (\ref{dilate_eqn}) for all $x, x' \in X, a, a' \in A, y, y' \in Y, b, b' \in B$. \textit{Locally dilatable} and \textit{quantum dilatable} strongly bistochastic operator matrices are defined exactly as with strongly stochastic operator matrices, but using bistochastic operator matrices in place of stochastic operator matrices.


\subsection{SQNS correlations}\label{sqns_correlations_section}

One goal for this work is to develop a method for transferring perfect strategies of one quantum input-output game to another. The framework we use for this purpose is by embedding both quantum games into a single ``game-of-games", whose winning QNS strategies encode the winning information for both simultaneously. The winning QNS strategies are used as simulators in the simulation paradigm, allowing us to take a perfect strategy for the first game and--- in conjunction with a winning strategy for the ``game-of-games"--- construct the desired strategy for the second. In order to adhere to the no-signalling condition for both games, we have to modify the classes of QNS correlations we use as simulators; this leads us to the introduction of a particular sub-class of QNS correlations, which is the focus of this subsection.

Let $X_{i}, Y_{i}, A_{i}, B_{i}, i = 1, 2$ be finite sets. A quantum channel 
\begin{gather*}
	\Gamma: M_{X_{2}Y_{2}\times A_{1}B_{1}}\rightarrow M_{X_{1}Y_{1}\times A_{2}B_{2}}
\end{gather*}
\noindent will be called a \textit{strongly quantum no-signalling (SQNS) correlation} if
\begin{gather}\label{pt_one}
{\rm Tr}_{X_{1}}\Gamma(\rho_{X_{2}Y_{2}}\otimes \rho_{A_{1}B_{1}}) = 0 \text{ if } \rho_{X_{2}Y_{2}} \in M_{X_{2}Y_{2}} \text{ and } {\rm Tr}_{X_{2}}(\rho_{X_{2}Y_{2}}) = 0,
\end{gather}
\begin{gather}\label{pt_two}
{\rm Tr}_{Y_{1}}\Gamma(\rho_{X_{2}Y_{2}}\otimes \rho_{A_{1}B_{1}}) = 0 \text{ if } \rho_{X_{2}Y_{2}} \in M_{X_{2}Y_{2}} \text{ and } {\rm Tr}_{Y_{2}}(\rho_{X_{2}Y_{2}}) = 0,
\end{gather}
\begin{gather}\label{pt_three}
{\rm Tr}_{A_{2}}\Gamma(\rho_{X_{2}Y_{2}}\otimes \rho_{A_{1}B_{1}}) = 0 \text{ if } \rho_{A_{1}B_{1}} \in M_{A_{1}B_{1}} \text{ and } {\rm Tr}_{A_{1}}(\rho_{A_{1}B_{1}}) = 0,
\end{gather}
\noindent and
\begin{gather}\label{pt_four}
{\rm Tr}_{B_{2}}\Gamma(\rho_{X_{2}Y_{2}}\otimes \rho_{A_{1}B_{1}}) = 0 \text{ if } \rho_{A_{1}B_{1}} \in M_{A_{1}B_{1}} \text{ and } {\rm Tr}_{B_{1}}(\rho_{A_{1}B_{1}}) = 0.
\end{gather}
We denote by $\cl{Q}_{\rm sns}$ the (convex) set of all SQNS correlations; it is clear that $\cl{Q}_{\rm sns} \subseteq \cl{Q}_{\rm ns}$. 

\medskip

A ${\it classical \; strongly \; no-signalling \; (SNS) \; correlation}$ \cite{gt_one} is a correlation 
\begin{gather*}
	p = (p(x_{1}y_{1}, a_{2}b_{2}|x_{2}y_{2}, a_{1}b_{1}))_{x_{1}y_{1}, a_{1}b_{1}, x_{2}y_{2}, a_{2}b_{2}}
\end{gather*}
\noindent that satisfies the conditions
$$\label{sns_one}\sum_{x_1\in X_1} \hspace{-0.1cm}
p(x_1y_1, a_2b_2 | x_2y_2, a_1b_1) \hspace{-0.05cm} = \hspace{-0.2cm}
\sum_{x_1\in X_1} \hspace{-0.1cm}
p(x_1y_1, a_2b_2 | x_2'y_2, a_1b_1), \ \ x_2,x_2'\in X_2,$$
$$\label{sns_two}\sum_{y_1\in Y_1} \hspace{-0.1cm}
p(x_1y_1, a_2b_2 | x_2y_2, a_1b_1) \hspace{-0.05cm} = \hspace{-0.2cm}
\sum_{y_1\in Y_1} \hspace{-0.1cm}
p(x_1y_1, a_2b_2 | x_2y_2', a_1b_1), \ \ y_2,y_2'\in Y_2,$$
$$\label{sns_three}\sum_{a_2\in A_2} \hspace{-0.1cm}
p(x_1y_1, a_2b_2 | x_2y_2, a_1b_1) \hspace{-0.05cm} = \hspace{-0.2cm}
\sum_{a_2\in A_2} \hspace{-0.1cm}
p(x_1y_1, a_2b_2 | x_2y_2, a_1'b_1), \ \ a_1,a_1'\in A_1,$$
and
$$\label{sns_four}\sum_{b_2\in B_2} \hspace{-0.1cm}
p(x_1y_1, a_2b_2 | x_2y_2, a_1b_1) \hspace{-0.05cm} = \hspace{-0.2cm}
\sum_{b_2\in B_2} \hspace{-0.1cm}
p(x_1y_1, a_2b_2 | x_2y_2, a_1b_1'), \ \ b_1,b_1'\in B_1.$$

\noindent The (convex) collection of all SNS correlations is denoted by $\cl{C}_{\rm sns}$. If $p \in \cl{C}_{\rm sns}$, then marginal conditional probability distributions 
\begin{gather*}
	p(x_{1}y_{1}, a_{2}|x_{2}y_{2}, a_{1}), \; p(x_{1}y_{1}, b_{2}|x_{2}y_{2}, b_{1}), 
\\ 	p(x_{1}, a_{2}b_{2}|x_{2}, a_{1}b_{1}) \text{ and } p(y_{1}, a_{2}b_{2}|y_{2}, a_{1}b_{1})
\end{gather*}
\noindent are all well-defined. 

\begin{remark}\label{sns_vs_sqns}
\rm If $p$ is a classical correlation over $(X_{2}Y_{2}, A_{1}B_{1}, X_{1}Y_{1}, A_{2}B_{2})$, then $p$ is an SNS correlation precisely when $\Gamma_{p}$ is a SQNS correlation. To see why, first assume that $\rho_{X_{2}Y_{2}} \in M_{X_{2}Y_{2}}$ with ${\rm Tr}_{X_{2}}(\rho_{X_{2}Y_{2}}) = 0$. Writing
\begin{gather*}
	\rho_{X_{2}Y_{2}} = \sum \tilde{\rho}_{x_{2}, x_{2}'}^{y_{2}, y_{2}'}\epsilon_{x_{2}x_{2}'}\otimes \epsilon_{y_{2}y_{2}'},
\end{gather*}
\noindent where the sum is over all $x_{2}, x_{2}' \in X_{2}$ and $y_{2}, y_{2}' \in Y_{2}$. The partial trace condition on $\rho_{X_{2}Y_{2}}$ implies
\begin{gather}
	\sum\limits_{x_{2} \in X_{2}}\sum\limits_{y_{2}, y_{2}' \in Y_{2}}\tilde{\rho}_{x_{2}. x_{2}}^{y_{2}, y_{2}'}\epsilon_{y_{2}y_{2}'} = 0.
\end{gather}
Specifically, $\sum\limits_{x_{2}}\tilde{\rho}_{x_{2}, x_{2}}^{y_{2}, y_{2}} = 0$ for any $y_{2} \in Y_{2}$. If we then take $\rho_{A_{1}B_{1}} \in M_{A_{1}B_{1}}$, we see
\begin{eqnarray*}
& &
{\rm Tr}_{X_{1}}\Gamma_{p}(\rho_{X_{2}Y_{2}}\otimes \rho_{A_{1}B_{1}}) \\
& = &
\sum\limits_{x_{2}, y_{2}}\sum\limits_{a_{1}, b_{1}}\sum\limits_{x_{1}, y_{1}}\sum\limits_{a_{2}, b_{2}}\tilde{\rho}_{x_{2}, x_{2}}^{y_{2}, y_{2}}p(x_{1}y_{1}, a_{2}b_{2}|x_{2}y_{2}, a_{1}b_{1})\langle \rho_{A_{1}B_{1}}(e_{a_{1}}\otimes e_{b_{1}}), e_{a_{1}}\otimes e_{b_{1}}\rangle \\
& &
\;\;\;\;\;\;\;\; \times\epsilon_{y_{1}y_{1}}\otimes \epsilon_{a_{2}a_{2}}\otimes \epsilon_{b_{2}b_{2}} \\
& = &
\sum\limits_{a_{1}, b_{1}}\sum\limits_{a_{2}, b_{2}}\sum\limits_{y_{1}}\bigg(\sum\limits_{y_{2}}\sum\limits_{x_{2}}\tilde{\rho}_{x_{2}, x_{2}}^{y_{2}, y_{2}}p(y_{1}, a_{2}b_{2}|y_{2}, a_{1}b_{1})\bigg)\langle \rho_{A_{1}B_{1}}(e_{a_{1}}\otimes e_{b_{1}}), e_{a_{1}}\otimes e_{b_{1}}\rangle \\
& &
\;\;\;\;\;\;\;\; \times \epsilon_{y_{1}y_{1}}\otimes \epsilon_{a_{2}a_{2}}\otimes \epsilon_{b_{2}b_{2}} \\
& = &
0.
\end{eqnarray*}
Using a similar argument, we can check that conditions (\ref{pt_two}), (\ref{pt_three}), and (\ref{pt_four}) are satisfied.  Conversely, assuming $\Gamma_{p}$ satisfies (\ref{pt_one})-(\ref{pt_four}), if we let 
\begin{gather*}
	\rho = \epsilon_{x_{2}x_{2}}\otimes \epsilon_{y_{2}y_{2}}\otimes \epsilon_{a_{1}a_{1}}\otimes\epsilon_{b_{1}b_{1}}-\epsilon_{x_{2}'x_{2}'}\otimes \epsilon_{y_{2}y_{2}}\otimes \epsilon_{a_{1}a_{1}}\otimes\epsilon_{b_{1}b_{1}},
\end{gather*}
for $x_{2}, x_{2}' \in X_{2}, y_{2} \in Y_{2}, a_{1} \in A_{1}$ and $b_{1} \in B_{1}$ we have that ${\rm Tr}_{X_{2}}(\rho) = 0$ with
\begin{gather*}
	\Gamma_{p}(\rho)  = \sum\limits_{x_{1}, y_{1}}\sum\limits_{a_{2}, b_{2}}(p(x_{1}y_{1}, a_{2}b_{2}|x_{2}y_{2}, a_{1}b_{1})-p(x_{1}y_{1}, a_{2}b_{2}|x_{2}'y_{2}, a_{1}b_{1}))\\\;\;\;\;\;\;\;\;\;\;\;\;\times\epsilon_{x_{1}x_{1}}\otimes \epsilon_{y_{1}y_{1}}\otimes \epsilon_{a_{2}a_{2}}\otimes \epsilon_{b_{2}b_{2}}.
\end{gather*}
\noindent By (\ref{pt_one}), we see
\begin{gather*}
	\sum\limits_{x_{1}}p(x_{1}y_{1}, a_{2}b_{2}|x_{2}y_{2}, a_{1}b_{1}) = \sum\limits_{x_{1}}p(x_{1}y_{1}, a_{2}b_{2}|x_{2}'y_{2}, a_{1}b_{1}), \;\;\;\; x_{2}, x_{2}' \in X_{2}.
\end{gather*}
\noindent The other SNS conditions can be verified by replacing $\rho$ with
\begin{gather*}
	\epsilon_{x_{2}x_{2}}\otimes \epsilon_{y_{2}y_{2}}\otimes \epsilon_{a_{1}a_{1}}\otimes\epsilon_{b_{1}b_{1}}-\epsilon_{x_{2}x_{2}}\otimes \epsilon_{y_{2}'y_{2}'}\otimes \epsilon_{a_{1}a_{1}}\otimes\epsilon_{b_{1}b_{1}}, \\
	\epsilon_{x_{2}x_{2}}\otimes \epsilon_{y_{2}y_{2}}\otimes \epsilon_{a_{1}a_{1}}\otimes\epsilon_{b_{1}b_{1}}-\epsilon_{x_{2}x_{2}}\otimes \epsilon_{y_{2}y_{2}}\otimes \epsilon_{a_{1}'a_{1}'}\otimes\epsilon_{b_{1}b_{1}},
\end{gather*}
\noindent and
\begin{gather*}
	\epsilon_{x_{2}x_{2}}\otimes \epsilon_{y_{2}y_{2}}\otimes \epsilon_{a_{1}a_{1}}\otimes\epsilon_{b_{1}b_{1}}-\epsilon_{x_{2}x_{2}}\otimes \epsilon_{y_{2}y_{2}}\otimes \epsilon_{a_{1}a_{1}}\otimes\epsilon_{b_{1}'b_{1}'},
\end{gather*}
respectively. 

Thus, we see that if $p$ is an SNS correlation over $(X_{2}Y_{2}, A_{1}B_{1}, X_{1}Y_{1}, A_{2}B_{2})$, then $\Gamma_{p}$ is SQNS. Also, if $\Gamma$ is a $(X_{2}Y_{2}\times A_{1}B_{1}, X_{1}Y_{1}\times A_{2}B_{2})$ classical SQNS correlation then $\Gamma = \Gamma_{p}$ for some SNS correlation $p$. 
\end{remark}
\begin{remark}
\rm An easily verified alternative characterization of QNS correlations can be stated as follows: $\Phi: M_{XY}\rightarrow M_{AB}$ is QNS if and only if the unital completely positive map $\Phi^{*}: M_{AB}\rightarrow M_{XY}$ preserves subalgebras, in the sense that
\begin{gather*}
	\Phi^{*}(M_{A}\otimes 1) \subseteq M_{X}\otimes 1, \;\;\;\; \Phi^{*}(1\otimes M_{B}) \subseteq 1\otimes M_{Y}.
\end{gather*}
The strengthening of partial trace conditions in the definition of an SQNS correlation leads to the following result.
\end{remark}
\begin{proposition}
A QNS correlation $\Phi: M_{X_{2}Y_{2}, A_{1}B_{1}}\rightarrow M_{X_{1}Y_{1}, A_{2}B_{2}}$ is SQNS if and only if the unital completely positive map $\Phi^{*}: M_{X_{1}Y_{1}, A_{2}B_{2}}\rightarrow M_{X_{2}Y_{2}, A_{1}B_{1}}$ preserves subalgebras in the sense that
\begin{gather}\label{sub_alg_one}
	\Phi^{*}(1\otimes M_{Y_{1}}\otimes M_{A_{2}}\otimes M_{B_{2}}) \subseteq 1\otimes M_{Y_{2}}\otimes M_{A_{1}}\otimes M_{B_{1}},
\end{gather}
\begin{gather}\label{sub_alg_two}
	 \Phi^{*}(M_{X_{1}}\otimes 1\otimes M_{A_{2}} \otimes M_{B_{2}}) \subseteq M_{X_{2}}\otimes 1\otimes M_{A_{1}}\otimes M_{B_{1}},
\end{gather}
\begin{gather}\label{sub_alg_three}
	\Phi^{*}(M_{X_{1}}\otimes M_{Y_{1}}\otimes 1\otimes M_{B_{2}}) \subseteq M_{X_{2}}\otimes M_{Y_{2}}\otimes 1\otimes M_{B_{1}},
\end{gather}
\noindent and
\begin{gather}\label{sub_alg_four}
	\Phi^{*}(M_{X_{1}}\otimes M_{Y_{1}}\otimes M_{A_{2}}\otimes 1) \subseteq M_{X_{2}}\otimes M_{Y_{2}}\otimes M_{A_{1}}\otimes 1.
\end{gather}
\end{proposition}
\begin{proof}
First, let $\Phi$ be an SQNS correlation and take an arbitrary $\rho \in 1\otimes M_{Y_{1}}\otimes M_{A_{2}}\otimes M_{B_{2}}$. Write
\begin{gather*}
	\rho = \sum\limits_{y_{1}, y_{1}'}\sum\limits_{a_{2}, a_{2}'}\sum\limits_{b_{2}, b_{2}'}\lambda_{y_{1}y_{1}'}^{a_{2}a_{2}', b_{2}b_{2}'}1\otimes \epsilon_{y_{1}y_{1}'}\otimes\epsilon_{a_{2}a_{2}'}\otimes \epsilon_{b_{2}b_{2}'}.
\end{gather*}
\noindent If we now fix $x_{2}, x_{2}' \in X_{2}$ such that $x_{2} \neq x_{2}'$, and for any $y_{2}, y_{2}' \in Y_{2}, a_{1}, a_{1}' \in A_{1}, b_{1}, b_{1}' \in B_{1}$ we see
\begin{eqnarray*}
& &
	\langle \Phi^{*}(\rho), \epsilon_{x_{2}x_{2}'}\otimes \epsilon_{y_{2}y_{2}'}\otimes \epsilon_{a_{1}a_{1}'}\otimes \epsilon_{b_{1}b_{1}'}\rangle \\
& = &
\sum\lambda_{y_{1}y_{1}'}^{a_{2}a_{2}', b_{2}b_{2}'}\langle\Phi^{*}(1\otimes \epsilon_{y_{1}y_{1}'}\otimes \epsilon_{a_{2}a_{2}'}\otimes \epsilon_{b_{2}b_{2}'}), \epsilon_{x_{2}x_{2}'}\otimes \epsilon_{y_{2}y_{2}'}\otimes \epsilon_{a_{1}a_{1}'}\otimes \epsilon_{b_{1}b_{1}'}\rangle \\
& = &
\sum\lambda_{y_{1}y_{1}'}^{a_{2}a_{2}', b_{2}b_{2}'}\langle 1\otimes \epsilon_{y_{1}y_{1}'}\otimes\epsilon_{a_{2}a_{2}'}\otimes \epsilon_{b_{2}b_{2}'}, \Phi(\epsilon_{x_{2}x_{2}'}\otimes \epsilon_{y_{2}y_{2}'}\otimes \epsilon_{a_{1}a_{1}'}\otimes \epsilon_{b_{1}b_{1}'})\rangle \\
& = &
\sum\lambda_{y_{1}y_{1}'}^{a_{2}a_{2}', b_{2}b_{2}'}\langle \epsilon_{y_{1}y_{1}'}\otimes \epsilon_{a_{2}a_{2}'}\otimes \epsilon_{b_{2}b_{2}'}, {\rm Tr}_{X_{1}}\Phi(\epsilon_{x_{2}x_{2}'}\otimes \epsilon_{y_{2}y_{2}'}\otimes \epsilon_{a_{1}a_{1}'}\otimes \epsilon_{b_{1}b_{1}'})\rangle \\
& = &
0,
\end{eqnarray*}
\noindent where the latter sums are over all $y_{1}, y_{1}' \in Y_{1}, a_{2}, a_{2}' \in A_{2}, b_{2}, b_{2}' \in B_{2}$ as ${\rm Tr}_{X_{2}}(\epsilon_{x_{2}x_{2}'}\otimes \epsilon_{y_{2}y_{2}'}) = 0$ by our choice of $x_{2}, x_{2}' \in X_{2}$. As this holds for all $y_{2}, y_{2}' \in Y_{2}, a_{1}, a_{1}' \in A_{1}$ and $b_{1}, b_{1}' \in B_{1}$, and our choice of $\rho \in 1\otimes M_{Y_{1}}\otimes M_{A_{2}}\otimes M_{B_{2}}$ was arbitrary, this implies 
\begin{gather*}
	\Phi^{*}(1\otimes M_{Y_{1}}\otimes M_{A_{2}}\otimes M_{B_{2}}) \subseteq 1\otimes M_{Y_{2}}\otimes M_{A_{1}}\otimes M_{B_{1}}.
\end{gather*}
\noindent Using (\ref{pt_two})-(\ref{pt_four}), we may then show that conditions (\ref{sub_alg_two})-(\ref{sub_alg_four}) also hold. 

Conversely, assume that $\Phi^{*}$ preserves subalgebras as in (\ref{sub_alg_one})-(\ref{sub_alg_four}). We wish to show that $\Phi$ is an SQNS correlation; to that end, let $\rho_{X_{2}Y_{2}} \in M_{X_{2}Y_{2}}$ and $\rho_{A_{1}B_{1}} \in M_{A_{1}B_{1}}$ such that ${\rm Tr}_{X_{2}}(\rho_{X_{2}Y_{2}}) = 0$. Write
\begin{gather*}
	\rho_{X_{2}Y_{2}} = \sum\limits_{x_{2}, x_{2}'}\sum\limits_{y_{2}, y_{2}'}\rho_{x_{2}x_{2}'}^{y_{2}y_{2}'}\epsilon_{x_{2}x_{2}'}\otimes \epsilon_{y_{2}y_{2}'}, \;\;\;\; \rho_{A_{1}B_{1}} = \sum\limits_{a_{1}, a_{1}'}\sum\limits_{b_{1}, b_{1}'}\rho_{a_{1}a_{1}'}^{b_{1}b_{1}'}\epsilon_{a_{1}a_{1}'}\otimes \epsilon_{b_{1}b_{1}'}.
\end{gather*}
\noindent For $y_{1}, y_{1}' \in Y_{1}, a_{2}, a_{2}' \in A_{2}, b_{2}, b_{2}' \in B_{2}$ we see
\begin{eqnarray*}
& &
\langle {\rm Tr}_{X_{1}}\Phi(\rho_{X_{2}Y_{2}}\otimes \rho_{A_{1}B_{1}}), \epsilon_{y_{1}y_{1}'}\otimes \epsilon_{a_{2}a_{2}'}\otimes \epsilon_{b_{2}b_{2}'}\rangle \\
& = &
\sum\limits_{x_{2}, x_{2}'}\sum\rho_{x_{2}x_{2}'}^{y_{2}y_{2}'}\rho_{a_{1}a_{1}'}^{b_{1}b_{1}'}\langle {\rm Tr}_{X_{1}}\Phi(\epsilon_{x_{2}x_{2}'}\otimes \epsilon_{y_{2}y_{2}'}\otimes \epsilon_{a_{1}a_{1}'}\otimes \epsilon_{b_{1}b_{1}'}), \epsilon_{y_{1}y_{1}'}\otimes \epsilon_{a_{2}a_{2}'}\otimes \epsilon_{b_{2}b_{2}'}\rangle \\
& = &
\sum\limits_{x_{2}, x_{2}'}\sum\rho_{x_{2}x_{2}'}^{y_{2}y_{2}'}\rho_{a_{1}a_{1}'}^{b_{1}b_{1}'}\langle \epsilon_{x_{2}x_{2}'}\otimes \epsilon_{y_{2}y_{2}'}\otimes \epsilon_{a_{1}a_{1}'}\otimes \epsilon_{b_{1}b_{1}'}, \Phi^{*}(1\otimes \epsilon_{y_{1}y_{1}'}\otimes \epsilon_{a_{2}a_{2}'}\otimes \epsilon_{b_{2}b_{2}'})\rangle \\
& =  &
\sum\limits_{x_{2}}\sum\rho_{x_{2}x_{2}}^{y_{2}y_{2}'}\rho_{a_{1}a_{1}'}^{b_{1}b_{1}'}(\Phi^{*})_{y_{2}'y_{2}}^{a_{1}'a_{1}, b_{1}'b_{1}} \\
& = &
0,
\end{eqnarray*}
\noindent where the second summand ranges over all $y_{2}, y_{2}' \in Y_{2}, a_{1}, a_{1}' \in A_{1}, b_{1}, b_{1}' \in B_{1}$ as $\sum\limits_{x_{2}}\rho_{x_{2}x_{2}}^{y_{2}y_{2}'} = 0$ for any $y_{2}, y_{2}' \in Y_{2}$. Note here that we are using (\ref{sub_alg_one}) and writing
\begin{gather*}
	\Phi^{*}(1\otimes \epsilon_{y_{1}y_{1}'}\otimes \epsilon_{a_{2}a_{2}'}\otimes \epsilon_{b_{2}b_{2}'}) = \sum (\Phi^{*})_{y_{2}y_{2}'}^{a_{1}a_{1}', b_{1}b_{1}'}(1\otimes \epsilon_{y_{2}y_{2}'}\otimes \epsilon_{a_{1}a_{1}'}\otimes \epsilon_{b_{1}b_{1}'}),  
\end{gather*}                                                                    
\noindent where the latter sum is over all of $Y_{2}, A_{1}$ and $B_{1}$. As our choice of $y_{1}, y_{1}' \in Y_{1}, a_{2}, a_{2}' \in A_{2}$ and $b_{2}, b_{2}' \in B_{2}$ were arbitrary, we conclude that (\ref{pt_one}) holds. Conditions (\ref{pt_two})-(\ref{pt_four}) are verified similarly.                                        
\end{proof}
\begin{definition}\label{sqns_subclass_defn}
An SQNS correlation $\Gamma$ over $(X_{2}Y_{2}, A_{1}B_{1}, X_{1}Y_{1}, A_{2}B_{2})$ is called
\begin{itemize}
	\item[(i)] {\rm quantum commuting} if there exists a Hilbert space $H$, stochastic operator matrices
	\begin{gather*}\label{sqc_matrices}
		E_{X} = (E_{x_{2}x_{2}', x_{1}x_{1}'})_{x_{2}x_{2}', x_{1}x_{1}'}, \;\;\;\; E_{Y} = (E^{y_{2}y_{2}', y_{1}y_{1}'})_{y_{2}y_{2}', y_{1}y_{1}'},
\\	           F_{A} = (F_{a_{1}a_{1}', a_{2}a_{2}'})_{a_{1}a_{1}', a_{2}a_{2}'}, \;\;\;\; F_{B} = (F^{b_{1}b_{1}', b_{2}b_{2}'})_{b_{1}b_{1}', b_{2}b_{2}'}
	\end{gather*} 
	\noindent acting on $H$ with mutually commuting entries, and a unit vector $\xi \in H$ such that $\Gamma = \Gamma_{E, F, \xi}$, where $E = E_{X}\cdot E_{Y}, F = F_{A}\cdot F_{B}$.
	\item[(ii)] {\rm quantum} if there exists finite dimensional Hilbert spaces $H$ and $K$, quantum dilatable strongly stochastic operator matrices 
	\begin{gather*}\label{sq_matrices}
		M = (M_{x_{2}x_{2}', a_{1}a_{1}'}^{x_{1}x_{1}', a_{2}a_{2}'})_{x_{2}x_{2}', a_{1}a_{1}', x_{1}x_{1}', a_{2}a_{2}'} \;\;\;\; {\rm on \;} H, 
\\		N = (N_{y_{2}y_{2}', b_{1}b_{1}'}^{y_{1}y_{1}', b_{2}b_{2}'})_{y_{2}y_{2}', b_{1}b_{1}', y_{1}y_{1}', b_{2}b_{2}'} \;\;\;\; {\rm on \;} K,
	\end{gather*}
	\noindent and a unit vector $\xi \in H\otimes K$, such that $\Gamma = \Gamma_{M \odot N, \xi}$. 
	\item[(iii)] {\rm approximately quantum} if it is the limit of quantum SQNS correlations.
	\item[(iv)] {\rm local} if it is quantum, and the matrices $M$ and $N$ from (ii) can be chosen to be locally dilatable.
\end{itemize}
\end{definition}
We denote by $\cl{Q}_{\rm sqc}$ (resp. $\cl{Q}_{\rm sqa}, \cl{Q}_{\rm sq}, \cl{Q}_{\rm sloc}$) the classes of quantum commuting (resp. approximately quantum, quantum, local) SQNS correlations. We note that, by definition, $\cl{Q}_{\rm st} \subseteq \cl{Q}_{\rm t}$, for ${\rm t} \in \{\rm loc, q, qa, qc, ns\}$. In the sequel, for ${\rm t} \in \{\rm loc, q, qa, qc, ns\}$ let $\cl{C}_{\rm st}$ denote the class of all SNS correlations of type ${\rm t}$, as defined in \cite[Section 5]{gt_one}. 

\begin{remark}\label{change_in_defn}
\rm We pause here for a correction to a previous work, and some necessary clarifying remarks. In \cite{gt_one}, we defined a classical correlation $\Gamma$ to belong to the class $\cl{C}_{\rm sqc}$  if there existed a Hilbert space $H$, dilatable NS operator matrices $P = (P_{x_{2}y_{2}, x_{1}y_{1}})_{x_{2}y_{2}, x_{1}y_{1}}$ and $Q = (Q_{a_{1}b_{1}, a_{2}b_{2}})_{a_{1}b_{1}, a_{2}b_{2}}$ on $H$ with mutually commuting entries, and a unit vector $\xi \in H$ such that $\Gamma = \Gamma_{P, Q, \xi}$. In \cite[Lemma 5.8]{gt_one} we claimed that this was a necessary and sufficient condition for the existence of a Hilbert space $K$, PVM's $(P_{x_{2}, x_{1}})_{x_{1} \in X_{1}}, (P^{y_{2}, y_{1}})_{y_{1} \in Y_{1}}, (Q_{a_{1}, a_{2}})_{a_{2} \in A_{2}}$ and $(Q^{b_{1}, b_{2}})_{b_{2} \in B_{2}}$ on $K$ with mutually commuting entries, and a unit vector $\eta \in K$ such that
\begin{gather*}
	\Gamma(x_{1}y_{1}, a_{2}b_{2}|x_{2}y_{2}, a_{1}b_{1}) = \langle P_{x_{2}, x_{1}}P^{y_{2}, y_{1}}Q_{a_{1}, a_{2}}Q^{b_{1}, b_{2}}\eta, \eta\rangle
\end{gather*}
\noindent for all $x_{i} \in X_{i}, y_{i} \in Y_{i}, a_{i} \in A_{i}, b_{i} \in B_{i}, i = 1, 2$. However, we have since discovered a gap in the proof of said lemma resulting from the incorrect use of associativity of the commuting operator system tensor product (see Section \ref{s_char_sqns}).

\medskip

In generalizing to the case of what we call SQNS correlations, as the natural generalization of NS operator matrices are strongly stochastic operator matrices if we were to follow the path laid out in \cite[Definition 5.6 (i)]{gt_one} one would expect that we say $\Gamma \in \cl{Q}_{\rm sqc}$ if there exists a Hilbert space $H$, dilatable strongly stochastic operator matrices $P = (P_{x_{2}x_{2}', y_{2}y_{2}'}^{x_{1}x_{1}', y_{1}y_{1}'})_{x_{2}x_{2}', y_{2}y_{2}', x_{1}x_{1}', y_{1}y_{1}'}$ and $Q = (Q_{a_{1}a_{1}', b_{1}b_{1}'}^{a_{2}a_{2}', b_{2}b_{2}'})_{a_{1}a_{1}', b_{1}b_{1}', a_{2}a_{2}', b_{2}b_{2}'}$ with mutually commuting entries acting on $H$, and a unit vector $\xi \in H$ such that $\Gamma = \Gamma_{P, Q, \xi}$. However, in order to get the decomposition of strongly stochastic operator matrices that we will need for the results in Section \ref{s_char_sqns} due to the failure of \cite[Lemma 5.8]{gt_one} (and its extension) we must instead require it in the definition (as seen in Definition \ref{sqns_subclass_defn} (i)). Furthermore, we will need to amend the definition for $\Gamma \in \cl{C}_{\rm sqc}$, which we do below. We point out that working with Definition \ref{sqc_amended_definition} (see below) instead of \cite[Definition 5.6(i)]{gt_one} recovers all subsequent results in \cite{gt_one} which were instead obtained using \cite[Lemma 5.8]{gt_one}. 
\begin{definition}\label{sqc_amended_definition}
An SNS correlation $\Gamma \in \cl{C}_{\rm sns}$ is now quantum commuting if there exists a Hilbert space $K$, PVM's $(P_{x_{2}, x_{1}})_{x_{1} \in X_{1}}, (P^{y_{2}, y_{1}})_{y_{1} \in Y_{1}}, (Q_{a_{1}, a_{2}})_{a_{2} \in A_{2}}$ and $(Q^{b_{1}, b_{2}})_{b_{2} \in B_{2}}$ on $K$ with mutually commuting entries, and a unit vector $\eta \in K$ such that
\begin{gather*}
	\Gamma(x_{1}y_{1}, a_{2}b_{2}|x_{2}y_{2}, a_{1}b_{1}) = \langle P_{x_{2}, x_{1}}P^{y_{2}, y_{1}}Q_{a_{1}, a_{2}}Q^{b_{1}, b_{2}}\eta, \eta\rangle
\end{gather*}
\noindent for all $x_{i} \in X_{i}, y_{i} \in Y_{i}, a_{i} \in A_{i}, b_{i} \in B_{i}, i = 1, 2$. 

\end{definition}
\end{remark}

\begin{proposition}\label{st_vs_qst}
Let ${\rm t} \in \{\rm loc, q, qa, qc, ns\}$ and $p$ an SNS correlation. Then $p \in \cl{C}_{\rm st}$ if and only if $\Gamma_{p} \in \cl{Q}_{\rm st}$. 
\end{proposition}
\begin{proof}
In the case when ${\rm t} = {\rm ns}$, this follows directly by Remark \ref{sns_vs_sqns}. For ${\rm t} = {\rm loc, q, qa}$ or ${\rm qc}$, this follows from Remark \ref{ns_generalization} and \cite[Lemma 7.2]{tt}.
\end{proof}

\begin{remark}\label{sloc_convex_remark}
\rm Let $\Gamma$ be a local SQNS correlation over $(X_{2}Y_{2}$, $A_{1}B_{1}$, $X_{1}Y_{1}$, $A_{2}B_{2})$. If we choose dilations of the matrices $M$ and $N$ from (\ref{sq_matrices}) with mutually commuting entries, we may write the values of $\Gamma$ in the form
\begin{gather*}
	\Gamma(x_{1}x_{1}', y_{1}y_{1}', a_{2}a_{2}', b_{2}b_{2}'|x_{2}x_{2}', y_{2}y_{2}', a_{1}a_{1}', b_{1}b_{1}') = \langle E_{x_{2}, x_{2}'}^{x_{1}, x_{1}'}E_{a_{1}, a_{1}'}^{a_{2}, a_{2}'}F_{y_{2}, y_{2}'}^{y_{1}, y_{1}'}F_{b_{1}, b_{1}'}^{b_{2}, b_{2}'}\xi, \xi\rangle
\end{gather*}
\noindent where the stochastic operator matrices $(E_{x_{2}, x_{2}'}^{x_{1}, x_{1}'})_{x_{2}, x_{2}', x_{1}, x_{1}'}$, $(E_{a_{1}, a_{1}'}^{a_{2}, a_{2}'})_{a_{1}, a_{1}', a_{2}, a_{2}'}$, $(F_{y_{2}, y_{2}'}^{y_{1}, y_{1}'})_{y_{2}, y_{2}', y_{1}, y_{1}'}$, $(F_{b_{1}, b_{1}'}^{b_{2}, b_{2}'})_{b_{1}, b_{1}', b_{2}, b_{2}'}$ have mutually commuting entries. Using the fact that the tensor product of convex combinations of channels remains a convex combination, along with the arguments from Proposition \ref{qloc_prop} we may conclude that $\Gamma$ is a convex combination of the form
\begin{gather}\label{sloc_convex_comb}
	\Gamma = \sum\limits_{j=1}^{k}\lambda_{j}\Phi_{X}^{(j)}\otimes \Phi_{Y}^{(j)}\otimes \Phi_{A}^{(j)}\otimes \Phi_{B}^{(j)}, 
\end{gather}
\noindent where $\Phi_{X}^{(j)}: M_{X_{2}}\rightarrow M_{X_{1}}, \Phi_{Y}^{(j)}: M_{Y_{2}}\rightarrow M_{Y_{1}}, \Phi_{A}^{(j)}: M_{A_{1}}\rightarrow M_{A_{2}}$, and $\Phi_{B}^{(j)}: M_{B_{1}}\rightarrow M_{B_{2}}$ are quantum channels, $j = 1, \hdots, k$. It is also easy to verify that any SQNS correlation of the form (\ref{sloc_convex_comb}) is in $\cl{Q}_{\rm sloc}$. 
\end{remark}

\begin{remark}
\rm One may easily see that the operator matrices $E, F$ arising from Definition \ref{sqns_subclass_defn} (i) of any $\Gamma \in \cl{Q}_{\rm sqc}$ will be dilatable; as discussed in Remark \ref{change_in_defn} it is unknown if all mutually commuting dilatable strongly stochastic operator matrices can be made jointly dilatable. 
\end{remark}

In the remainder of this section, assume that $X_{1} = X_{2} := X, Y_{1} = Y_{2} := Y, A_{1} = A_{2} := A$, and $B_{1} = B_{2} := B$. Furthermore, assume that $X = A$ and $Y = B$. An SQNS correlation $\Gamma$ will be called an \textit{SQNS bicorrelation} if $\Gamma$ is unital and $\Gamma^{*}$ is an SQNS correlation. The collection of all SQNS correlations will be denoted $\cl{Q}_{\rm sns}^{\rm bi}$. An SQNS bicorrelation $\Gamma$ over $(XY, XY, XY, XY)$ is called \textit{quantum commuting} if there exists a Hilbert space $H$, strongly bistochastic operator matrices 
\begin{gather*}
	E_{X} = (E_{x_{2}x_{2}', x_{1}x_{1}'})_{x_{2}x_{2}', x_{1}x_{1}'}, \;\;\;\; E_{Y} = (E_{y_{2}y_{2}', y_{1}y_{1}'})_{y_{2}y_{2}', y_{1}y_{1}'}, 
\\	F_{A} = (F_{a_{1}a_{1}', a_{2}a_{2}'})_{a_{1}a_{1}', a_{2}a_{2}'}, \;\;\;\; F_{B} = (F_{b_{1}b_{1}', b_{2}b_{2}'})_{b_{1}b_{1}', b_{2}b_{2}'},
\end{gather*}
\noindent with mutually commuting entries acting on $H$, and a unit vector $\xi \in H$ such that $\Gamma = \Gamma_{E, F, \xi}$ where $E = E_{X}\cdot E_{Y}, F = F_{A}\cdot F_{B}$; we denote this class by $\cl{Q}_{\rm sqc}^{\rm bi}$. The classes of quantum SQNS bicorrelations  (denoted $\cl{Q}_{\rm sq}^{\rm bi}$), approximately quantum SQNS bicorrelations (denoted $\cl{Q}_{\rm sqa}^{\rm bi}$) and local SQNS bicorrelations (denoted $\cl{Q}_{\rm sloc}^{\rm bi}$) are defined similarly to their SQNS correlation counterparts, with dilatable strongly bistochastic operator matrices of the appropriate type in place of the strongly stochastic operator matrices of said type. In the sequel, we let $\cl{C}_{\rm st}^{\rm bi}$ denote the class of all SQNS bicorrelations of type ${\rm t}$, as defined in \cite[Section 6]{gt_one} and Remark \ref{change_in_defn}.

\begin{remark}\label{bisloc_convex_remark}
\rm Arguing as in Remark \ref{sloc_convex_remark}, we may identify $\cl{Q}_{\rm sloc}^{\rm bi}$ with SQNS correlations $\Gamma$ of the form
\begin{gather*}
	\Gamma = \sum\limits_{j=1}^{k}\lambda_{j}\Phi_{X}^{(j)}\otimes\Phi_{Y}^{(j)}\otimes\Phi_{A}^{(j)}\otimes \Phi_{B}^{(j)},
\end{gather*}
\noindent where $\Phi_{X}^{(j)}: M_{X}\rightarrow M_{X}, \Phi_{Y}^{(j)}: M_{Y}\rightarrow M_{Y}, \Phi_{A}^{(j)}: M_{A}\rightarrow M_{A}, \Phi_{B}^{(j)}: M_{B}\rightarrow M_{B}$ are unital quantum channels, $\lambda_{j} \geq 0$, $j = 1, \hdots, k$ with $\sum_{j=1}^{k}\lambda_{j} = 1$.
\end{remark}

\begin{remark}
\rm For a correlation type ${\rm t} \in \{\rm loc, q, qa, qc, ns\}$, it is clear that $\cl{Q}_{\rm st}^{\rm bi} \subseteq \cl{Q}_{\rm st}$. Additionally, if $\Gamma \in \cl{Q}_{\rm st}^{\rm bi}$, then $\Gamma^{*} \in \cl{Q}_{\rm st}^{\rm bi}$. Indeed: it is part of the definition when ${\rm t} = {\rm ns}$, and easily verified (using Remark \ref{bisloc_convex_remark}) when ${\rm t} = {\rm loc}$. The case when ${\rm t} = {\rm qc, q}$ or ${\rm qa}$ follow using a modification of the argument given in \cite[Remark 5.2]{bhtt}.
\end{remark}

\begin{proposition}
Let ${\rm t} \in \{\rm loc, q, qc, ns\}$ and $p$ an SNS bicorrelation. Then $p \in \cl{C}_{\rm st}^{\rm bi}$ if and only if $\Gamma_{p} \in \cl{Q}_{\rm st}^{\rm bi}$.
\end{proposition}
\begin{proof}
This holds essentially by using the same arguments as in Proposition \ref{st_vs_qst} and \cite[Proposition 5.9]{bhtt}.
\end{proof}


\section{Strategy transport}\label{strat_transport}

Let $X_{i}, Y_{i}, A_{i}, B_{i}$ be finite sets, for $i = 1, 2$. Recall that if $\mathcal{E}$ is a QNS correlation over $(X_{1}, Y_{1}, A_{1}, B_{1})$, it acts as a quantum channel $\cl{E}: M_{X_{1}Y_{1}}\rightarrow M_{A_{1}B_{1}}$; thus, if $\Gamma$ is SQNS over $(X_{2}Y_{2}, A_{1}B_{1}, X_{1}Y_{1}, A_{2}B_{2})$, then $\Gamma[\cl{E}]: M_{X_{2}Y_{2}}\rightarrow M_{A_{2}B_{2}}$. The two theorems in this section are crucial to our goal of transferring strategies between quantum games; briefly, they say that if $\cl{E}$ is viewed as a strategy of some type for one game, under mild conditions on $\Gamma$ the simulated channel $\Gamma[\cl{E}]$ can be viewed as a strategy \textit{of the same type} for another quantum game.
\begin{theorem}\label{strat_transp}
Let $\Gamma$ be an SQNS correlation over $(X_{2}Y_{2}, A_{1}B_{1}, X_{1}Y_{1}, A_{2}B_{2})$ and $\cl{E}$ be a QNS correlation over $(X_{1}, Y_{1}, A_{1}, B_{1})$. The following hold:
\begin{itemize}
	\item[(i)] $\Gamma[\cl{E}] \in \cl{Q}_{\rm ns}$;
	\item[(ii)] if $\Gamma \in \cl{Q}_{\rm sqc}$ and $\cl{E} \in \cl{Q}_{\rm qc}$, then $\Gamma[\cl{E}] \in \cl{Q}_{\rm qc}$;
	\item[(iii)] if $\Gamma \in \cl{Q}_{\rm sqa}$ and $\cl{E} \in \cl{Q}_{\rm qa}$, then $\Gamma[\cl{E}] \in \cl{Q}_{\rm qa}$;
	\item[(iv)] if $\Gamma \in \cl{Q}_{\rm sq}$ and $\cl{E} \in \cl{Q}_{\rm q}$, then $\Gamma[\cl{E}] \in \cl{Q}_{\rm q}$.
	\item[(v)] if $\Gamma \in \cl{Q}_{\rm sloc}$ and $\cl{E} \in \cl{Q}_{\rm loc}$, then $\Gamma[\cl{E}] \in \cl{Q}_{\rm loc}$.
\end{itemize}
\end{theorem}
\begin{proof}
(i) Let $C$ denote the Choi matrix of quantum channel $\Gamma[\cl{E}]$. If $(\Gamma(x_{1}x_{1}', y_{1}y_{1}', a_{2}a_{2}', b_{2}b_{2}'|x_{2}x_{2}', y_{2}y_{2}', a_{1}a_{1}', b_{1}b_{1}'))$ is the Choi matrix of $\Gamma$ and $(\cl{E}(a_{1}a_{1}', b_{1}b_{1}'|x_{1}x_{1}', y_{1}y_{1}'))$ is the Choi matrix of $\cl{E}$ (where the former matrix ranges over $X_{1}, X_{2}, Y_{1}, Y_{2}, A_{1}, A_{2}, B_{1}, B_{2}$ and the latter ranges over $X_{1}, Y_{1}, A_{1}$ and $B_{1}$), as $\Gamma \in \cl{Q}_{\rm sns}$ and $\cl{E} \in \cl{Q}_{\rm ns}$ there exists $c_{x_{1}x_{1}', y_{1}y_{1}', x_{2}x_{2}', y_{2}y_{2}'}^{b_{2}b_{2}', b_{1}b_{1}'}$, $d_{y_{1}, y_{1}'}^{b_{1}, b_{1}'} \in \bb{C}$ such that 
\begin{gather*}
	\sum\limits_{a_{2} \in A_{2}}\Gamma(x_{1}x_{1}', y_{1}y_{1}', a_{2}a_{2}, b_{2}b_{2}'|x_{2}x_{2}', y_{2}y_{2}', a_{1}a_{1}', b_{1}b_{1}') = \delta_{a_{1}, a_{1}'}c_{x_{1}x_{1}', y_{1}y_{1}', x_{2}x_{2}', y_{2}y_{2}'}^{b_{2}b_{2}', b_{1}b_{1}'},
\\	\sum\limits_{a_{1} \in A_{1}}\cl{E}(a_{1}a_{1}, b_{1}b_{1}'|x_{1}x_{1}', y_{1}y_{1}') = \delta_{x_{1}, x_{1}'}d_{y_{1}, y_{1}'}^{b_{1}, b_{1}'},
\end{gather*}
\noindent for $x_{i}, x_{i}' \in X_{i}, y_{i}, y_{i}' \in Y_{i}, b_{i}, b_{i}' \in B_{i}, i = 1, 2$ (see e.g. \cite{dw}). Additionally, we may use (\ref{pt_one}) to guarantee the existence of $c_{y_{1}y_{1}', y_{2}y_{2}'}^{b_{2}b_{2}', b_{1}b_{1}'} \in \bb{C}$ such that
\begin{gather*}
	\sum\limits_{x_{1} \in X_{1}}c_{x_{1}x_{1}, y_{1}y_{1}', x_{2}x_{2}', y_{2}y_{2}'}^{b_{2}b_{2}', b_{1}b_{1}'} = \delta_{x_{2}, x_{2}'}c_{y_{1}y_{1}', y_{2}y_{2}'}^{b_{2}b_{2}', b_{1}b_{1}'},
\end{gather*}
\noindent for $y_{i}, y_{i}' \in Y_{i}, b_{i}, b_{i}' \in B_{i}, i = 1, 2$. For $y_{2}, y_{2}' \in Y_{2}, b_{2}, b_{2}' \in B_{2}$, set
\begin{gather*}
	\tilde{c}_{y_{2}, y_{2}'}^{b_{2}, b_{2}'} := \sum\limits_{y_{1}, y_{1}' \in Y_{1}}\sum\limits_{b_{1}, b_{1}' \in B_{1}}c_{y_{1}y_{1}', y_{2}y_{2}'}^{b_{2}b_{2}', b_{1}b_{1}'}d_{y_{1}, y_{1}'}^{b_{1}, b_{1}'}.
\end{gather*}
\noindent If we fix $y_{2}, y_{2}' \in Y_{2}, b_{2}, b_{2}' \in B_{2}$, then
\begin{eqnarray*}
& &
	\sum\limits_{a_{2}}\sum\limits_{\substack{x_{1}x_{1}', y_{1}y_{1}', \\ a_{1}a_{1}', b_{1}b_{1}'}}\Gamma(x_{1}x_{1}', y_{1}y_{1}', a_{2}a_{2}, b_{2}b_{2}'|x_{2}x_{2}', y_{2}y_{2}', a_{1}a_{1}', b_{1}b_{1}')\cl{E}(a_{1}a_{1}', b_{1}b_{1}'|x_{1}x_{1}', y_{1}y_{1}') \\
& = &
	\sum\limits_{\substack{x_{1}x_{1}', y_{1}y_{1}' \\ a_{1}a_{1}', b_{1}b_{1}'}}\delta_{a_{1}, a_{1}'}c_{x_{1}x_{1}', y_{1}y_{1}', x_{2}x_{2}', y_{2}y_{2}'}^{b_{2}b_{2}', b_{1}b_{1}'}\cl{E}(a_{1}a_{1}', b_{1}b_{1}'|x_{1}x_{1}', y_{1}y_{1}') \\
& = &
	\sum\limits_{a_{1}}\sum\limits_{\substack{x_{1}x_{1}', y_{1}y_{1}' \\ b_{1}b_{1}'}}c_{x_{1}x_{1}', y_{1}y_{1}', x_{2}x_{2}', y_{2}y_{2}'}^{b_{2}b_{2}', b_{1}b_{1}'}\cl{E}(a_{1}a_{1}, b_{1}b_{1}'|x_{1}x_{1}', y_{1}y_{1}') \\
& = &
	\sum\limits_{\substack{x_{1}x_{1}', y_{1}y_{1}', \\ b_{1}b_{1}'}}\delta_{x_{1}, x_{1}'}c_{x_{1}x_{1}', y_{1}y_{1}', x_{2}x_{2}', y_{2}y_{2}'}^{b_{2}b_{2}', b_{1}b_{1}'}d_{y_{1}, y_{1}'}^{b_{1}, b_{1}'} \\
& = &
	\sum\limits_{x_{1}}\sum\limits_{\substack{y_{1}y_{1}', \\ b_{1}b_{1}'}}c_{x_{1}x_{1}, y_{1}y_{1}', x_{2}x_{2}', y_{2}y_{2}'}^{b_{2}b_{2}', b_{1}b_{1}'}d_{y_{1}, y_{1}'}^{b_{1}, b_{1}'} \\
& = &
	\delta_{x_{2}, x_{2}'}\bigg(\sum\limits_{y_{1}y_{1}', b_{1}b_{1}'}c_{y_{1}y_{1}',  y_{2}y_{2}'}^{b_{2}b_{2}', b_{1}b_{1}'}d_{y_{1}, y_{1}'}^{b_{1}, b_{1}'}\bigg) \\
& = &
	\delta_{x_{2}, x_{2}'}\tilde{c}_{y_{2}, y_{2}'}^{b_{2}, b_{2}'}.
\end{eqnarray*}
\noindent We may use a similar argument (relying on (\ref{pt_two}) and (\ref{pt_four})) showing the existence of $\tilde{c}_{x_{2}, x_{2}'}^{a_{2}, a_{2}'} \in \bb{C}$ such that 
\begin{eqnarray*}
& &
	\sum\limits_{b_{2}}\sum\limits_{\substack{ x_{1}x_{1}', y_{1}y_{1}' \\ a_{1}a_{1}', b_{1}b_{1}'}}\Gamma(x_{1}x_{1}', y_{1}y_{1}', a_{2}a_{2}', b_{2}b_{2}|x_{2}x_{2}', y_{2}y_{2}', a_{1}a_{1}', b_{1}b_{1}')\cl{E}(a_{1}a_{1}', b_{1}b_{1}'|x_{1}x_{1}', y_{1}y_{1}') \\
& = &
	\delta_{y_{2}, y_{2}'}\tilde{c}_{x_{2}, x_{2}'}^{a_{2}, a_{2}'}
\end{eqnarray*}
\noindent for $x_{2}, x_{2}' \in X_{2}, a_{2}, a_{2}' \in A_{2}$. By (\ref{choi_simulated_eqn}), the former equality implies $L_{\omega_{XA}}(C) \in \cl{L}_{Y_{2}B_{2}}$ for every $\omega_{XA} \in M_{X_{2}A_{2}}$, while the latter equality implies $L_{\omega_{YB}}(C) \in \cl{L}_{X_{2}A_{2}}$ for every $\omega_{YB} \in M_{Y_{2}B_{2}}$. Thus, $C \in (\cl{L}_{X_{2}A_{2}}\otimes \cl{L}_{Y_{2}B_{2}}) \cap M_{X_{2}Y_{2}A_{2}B_{2}}^{+}$; by the injectivity of the minimal operator system tensor product, $C \in (\cl{L}_{X_{2}A_{2}}\otimes_{\rm min} \cl{L}_{Y_{2}B_{2}})^{+}$. Argue now as \cite[Theorem 6.2]{tt} to finish the proof.

(ii) Use Definition \ref{sqns_subclass_defn} (i) to obtain a Hilbert space $H$, stochastic operator matrices $P_{X} = (P_{x_{2}x_{2}', x_{1}x_{1}'})_{x_{2}x_{2}', x_{1}x_{1}'}$, $P_{Y} = (P^{y_{2}y_{2}', y_{1}y_{1}'})_{y_{2}y_{2}', y_{1}y_{1}'}$, $Q_{A} = (Q_{a_{1}a_{1}', a_{2}a_{2}'})_{a_{1}a_{1}', a_{2}a_{2}'}$ and $Q_{B} = (Q^{b_{1}b_{1}', b_{2}b_{2}'})_{b_{1}b_{1}', b_{2}b_{2}'}$ acting on $K$ and a unit vector $\xi \in H$ such that $\Gamma = \Gamma_{P, Q, \xi}$. Let $K$ be a Hilbert space, $E = (E_{x_{1}x_{1}', a_{1}a_{1}'})_{x_{1}x_{1}', a_{1}a_{1}'}$ and $F = (F^{y_{1}y_{1}', b_{1}b_{1}'})_{y_{1}y_{1}', b_{1}b_{1}'}$ stochastic operator matrices over $(X_{1}, A_{1})$ and $(Y_{1}, B_{1})$ with mutually commuting entries, and $\eta \in K$ a unit vector such that $\cl{E} = \cl{E}_{E, F, \eta}$. For $x_{2}, x_{2}' \in X_{2}, a_{2}, a_{2}' \in A_{2}, y_{2}, y_{2}' \in Y_{2}$, and $b_{2}, b_{2}' \in B_{2}$ set
\begin{gather}\label{qc_simulate_one}
	\tilde{E}_{x_{2}x_{2}', a_{2}a_{2}'} = \sum\limits_{x_{1}, x_{1}' \in X_{1}}\sum\limits_{a_{1}, a_{1}' \in A_{1}}P_{x_{2}x_{2}', x_{1}x_{1}'}Q_{a_{1}a_{1}', a_{2}a_{2}'}\otimes E_{x_{1}x_{1}', a_{1}a_{1}'}
\end{gather}
\noindent and
\begin{gather}\label{qc_simulate_two}
	\tilde{F}_{y_{2}y_{2}', b_{2}b_{2}'} = \sum\limits_{y_{1}, y_{1}' \in Y_{1}}\sum\limits_{b_{1}, b_{1}' \in B_{1}}P^{y_{2}y_{2}', y_{1}y_{1}'}Q^{b_{1}b_{1}', b_{2}b_{2}'}\otimes F^{y_{1}y_{1}', b_{1}b_{1}'}.
\end{gather}
\noindent We note that $\tilde{E} = (\tilde{E}_{x_{2}x_{2}', a_{2}a_{2}'})_{x_{2}x_{2}', a_{2}a_{2}'}$ and $\tilde{F} = (\tilde{F}_{y_{2}y_{2}', b_{2}b_{2}'})_{y_{2}y_{2}', b_{2}b_{2}'}$ are positive (as the entries of $P_{X}, P_{Y}, Q_{A}, Q_{B}$ all commute with one another), and claim that $\tilde{E}$ (resp. $\tilde{F}$) is a stochastic operator matrix over $(X_{2}, A_{2})$ (resp. over $(Y_{2}, B_{2})$). Indeed: as $P_{X}, Q_{A}$ and $E$ are all stochastic operator matrices, we have
\begin{eqnarray*}
	{\rm Tr}_{A_{2}}(\tilde{E})
& = &
\sum\limits_{a_{2}}\sum\limits_{x_{2}, x_{2}'}\sum\limits_{x_{1}, x_{1}'}\sum\limits_{a_{1}, a_{1}'}P_{x_{2}x_{2}', x_{1}x_{1}'}Q_{a_{1}a_{1}', a_{2}a_{2}}\otimes E_{x_{1}x_{1}', a_{1}a_{1}'} \\
& = &
\sum\limits_{x_{2}, x_{2}'}\sum\limits_{x_{1}, x_{1}'}\sum\limits_{a_{1}, a_{1}'}P_{x_{2}x_{2}', x_{1}x_{1}'}\big(\delta_{a_{1}, a_{1}'}I_{H}\big)\otimes E_{x_{1}x_{1}', a_{1}a_{1}'} \\
& = &
\sum\limits_{x_{2}, x_{2}'}\sum\limits_{x_{1}, x_{1}'}\sum\limits_{a_{1}}P_{x_{2}x_{2}', x_{1}x_{1}'}\otimes E_{x_{1}x_{1}', a_{1}a_{1}} \\
& = &
\sum\limits_{x_{2}, x_{2}'}\sum\limits_{x_{1}, x_{1}'}P_{x_{2}x_{2}', x_{1}x_{1}'}\otimes \big(\delta_{x_{1}, x_{1}'}I_{K}) \\
& = &
\sum\limits_{x_{2}, x_{2}'}\sum\limits_{x_{1}}P_{x_{2}x_{2}', x_{1}x_{1}}\otimes I_{K} \\
& = &
\sum\limits_{x_{2}, x_{2}'}\big(\delta_{x_{2}, x_{2}'}I_{H}\big)\otimes I_{K} \\
& = &
I_{X_{2}}\otimes I_{H}\otimes I_{K}.
\end{eqnarray*}
\noindent Thus, $\tilde{E} \in M_{X_{2}A_{2}}\otimes \cl{B}(H\otimes K)$ is stochastic, as claimed. By symmetry, $\tilde{F} \in M_{Y_{2}B_{2}}\otimes \cl{B}(H\otimes K)$ is stochastic. We also note that
\begin{gather*}
	\tilde{E}_{x_{2}x_{2}', a_{2}a_{2}'}\tilde{F}_{y_{2}y_{2}', b_{2}b_{2}'} = \tilde{F}_{y_{2}y_{2}', b_{2}b_{2}'}\tilde{E}_{x_{2}x_{2}', a_{2}a_{2}'},
\end{gather*}
\noindent for all $x_{2}, x_{2}' \in X_{2}, y_{2}, y_{2}' \in Y_{2}, a_{2}, a_{2}' \in A_{2}, b_{2}, b_{2}' \in B_{2}$. Using (\ref{choi_simulated_eqn}), we see
\begin{eqnarray*}
& &
	\Gamma[\cl{E}](a_{2}a_{2}', b_{2}b_{2}'|x_{2}x_{2}', y_{2}y_{2}') \\
& = &
	\sum\Gamma(x_{1}x_{1}', y_{1}y_{1}', a_{2}a_{2}', b_{2}b_{2}'|x_{2}x_{2}', y_{2}y_{2}', a_{1}a_{1}', b_{1}b_{1}')\cl{E}(a_{1}a_{1}', b_{1}b_{1}'|x_{1}x_{1}', y_{1}y_{1}') \\
& = &
	\sum\langle P_{x_{2}x_{2}', x_{1}x_{1}'}P^{y_{2}y_{2}', y_{1}y_{1}'}Q_{a_{1}a_{1}', a_{2}a_{2}'}Q^{b_{1}b_{1}', b_{2}b_{2}'}\xi, \xi\rangle\langle E_{x_{1}x_{1}', a_{1}a_{1}'}F^{y_{1}y_{1}', b_{1}b_{1}'}\eta, \eta\rangle \\
& = &
	\big\langle \tilde{E}_{x_{2}x_{2}', a_{2}a_{2}'}\tilde{F}_{y_{2}y_{2}', b_{2}b_{2}'}(\xi\otimes \eta), \xi\otimes \eta\big\rangle,
\end{eqnarray*}
\noindent for all $x_{2}, x_{2}' \in X_{2}, y_{2}, y_{2}' \in Y_{2}, a_{2}, a_{2}' \in A_{2}, b_{2}, b_{2}' \in B_{2}$, where the latter sums range over all of $X_{1}, Y_{1}, A_{1}$, and $B_{1}$. Thus, $\Gamma[\cl{E}] \in \cl{Q}_{\rm qc}$.

(iii) This is a direct consequence of (iv).

(iv) Let $M = (M_{x_{2}x_{2}', a_{1}a_{1}'}^{x_{1}x_{1}', a_{2}a_{2}'})_{x_{2}x_{2}', a_{1}a_{1}', x_{1}x_{1}', a_{2}a_{2}'}$ and $N = (N_{y_{2}y_{2}', b_{1}b_{1}'}^{y_{1}y_{1}', b_{2}b_{2}'})_{y_{2}y_{2}', b_{1}b_{1}', y_{1}y_{1}', b_{2}b_{2}'}$ be quantum dilatable strongly stochastic operator matrices acting on $H$ and $K$, respectively, and $\xi \in H\otimes K$ be a unit vector for which $\Gamma = \Gamma_{M\odot N, \xi}$. Let $E = (E_{x_{1}x_{1}', a_{1}a_{1}'})_{x_{1}x_{1}', a_{1}a_{1}'}$ and $F = (F_{y_{1}y_{1}', b_{1}b_{1}'})_{y_{1}y_{1}', b_{1}b_{1}'}$ be finite-dimensionally acting stochastic operator matrices with $\eta$ a unit vector such that $\cl{E} = \cl{E}_{E\odot F, \eta}$. For $x_{2}, x_{2}' \in X_{2}, a_{2}, a_{2}' \in A_{2}, y_{2}, y_{2}' \in Y_{2},$ and $b_{2}, b_{2}' \in B_{2}$ set
\begin{gather*}
	\tilde{E}_{x_{2}x_{2}', a_{2}a_{2}'} = \sum\limits_{x_{1}, x_{1}' \in X_{1}}\sum\limits_{a_{1}, a_{1}' \in A_{1}}M_{x_{2}x_{2}', a_{1}a_{1}'}^{x_{1}x_{1}', a_{2}a_{2}'}\otimes E_{x_{1}x_{1}', a_{1}a_{1}'}, 
\end{gather*}
\noindent and
\begin{gather*}
	\tilde{F}_{y_{2}y_{2}', b_{2}b_{2}'} = \sum\limits_{y_{1}, y_{1}' \in Y_{1}}\sum\limits_{b_{1}, b_{1}' \in B_{1}}N_{y_{2}y_{2}', b_{1}b_{1}'}^{y_{1}y_{1}', b_{2}b_{2}'}\otimes F_{y_{1}y_{1}', b_{1}b_{1}'};
\end{gather*}
\noindent using an almost identical argument as in the proof of (ii), it is easy to verify that $\tilde{E} = (\tilde{E}_{x_{2}x_{2}', a_{2}a_{2}'})_{x_{2}x_{2}', a_{2}a_{2}'}$ and $\tilde{F} = (\tilde{F}_{y_{2}y_{2}', b_{2}b_{2}'})_{y_{2}y_{2}', b_{2}b_{2}'}$ are finite-dimensionally acting stochastic operator matrices such that 
\begin{gather*}
	\Gamma[\cl{E}](a_{2}a_{2}', b_{2}b_{2}'|x_{2}x_{2}', y_{2}y_{2}') = \big\langle (\tilde{E}_{x_{2}x_{2}', a_{2}a_{2}'}\otimes \tilde{F}_{y_{2}y_{2}', b_{2}b_{2}'})(\xi\otimes \eta), \xi\otimes \eta\big\rangle. 
\end{gather*}
\noindent Thus, $\Gamma[\cl{E}] \in \cl{Q}_{\rm q}$. 

(v) If $\Phi_{X}: M_{X_{2}}\rightarrow M_{X_{1}}, \Phi_{Y}: M_{Y_{2}}\rightarrow M_{Y_{1}}, \Phi_{A}: M_{A_{1}}\rightarrow M_{A_{2}}$ and $\Phi_{B}: M_{B_{1}}\rightarrow M_{B_{2}}$ are quantum channels, it is easily verified (see \cite{gt_one}) that
\begin{gather*}
	(\Phi_{X}\otimes \Phi_{Y}\otimes \Phi_{A}\otimes \Phi_{B})[\cl{E}\otimes \cl{F}] = (\Phi_{X}\otimes \Phi_{A})[\cl{E}]\otimes (\Phi_{Y}\otimes \Phi_{B})[\cl{F}]
\end{gather*}
\noindent for all quantum channels $\cl{E}: M_{X_{1}}\rightarrow M_{A_{1}}$ and $\cl{F}: M_{Y_{1}}\rightarrow M_{B_{1}}$. Using Remark \ref{sloc_convex_remark}, we may conclude that $\Gamma[\cl{E}] \in \cl{Q}_{\rm loc}$ as claimed. 
\end{proof}

Fix finite sets $X$ and $Y$, and let $A = X, B = Y$.
\begin{theorem}\label{bi_strat_transp}
Let $\Gamma$ be an SQNS bicorrelation over $(XY, XY, XY, XY)$ and $\cl{E}$ be a QNS bicorrelation over $(X, Y, X, Y)$. The following hold:
\begin{itemize}
	\item[(i)] $\Gamma[\cl{E}] \in \cl{Q}_{\rm ns}^{\rm bi}$;
	\item[(ii)] if $\Gamma \in \cl{Q}_{\rm sqc}^{\rm bi}$ and $\cl{E} \in \cl{Q}_{\rm qc}^{\rm bi}$, then $\Gamma[\cl{E}] \in \cl{Q}_{\rm qc}^{\rm bi}$;
	\item[(iii)] if $\Gamma \in \cl{Q}_{\rm sqa}^{\rm bi}$ and $\cl{E} \in \cl{Q}_{\rm qa}^{\rm bi}$, then $\Gamma[\cl{E}] \in \cl{Q}_{\rm qa}^{\rm bi}$;
	\item[(iv)] if $\Gamma \in \cl{Q}_{\rm sq}^{\rm bi}$ and $\cl{E} \in \cl{Q}_{\rm q}^{\rm bi}$, then $\Gamma[\cl{E}] \in \cl{Q}_{\rm q}^{\rm bi}$;
	\item[(v)] if $\Gamma \in \cl{Q}_{\rm sloc}^{\rm bi}$ and $\cl{E} \in \cl{Q}_{\rm loc}^{\rm bi}$, then $\Gamma[\cl{E}] \in \cl{Q}_{\rm loc}^{\rm bi}$.
\end{itemize}
\end{theorem}
\begin{proof}
(i) The claim follows as in Theorem \ref{strat_transp}, once we note that
\begin{eqnarray*}
& &
	\Gamma[\cl{E}]^{*}(x_{2}x_{2}', y_{2}y_{2}'|a_{2}a_{2}', b_{2}b_{2}') \\
& = &
\sum\limits_{x_{1}x_{1}', y_{1}y_{1}'}\sum\limits_{a_{1}a_{1}', b_{1}b_{1}'}\Gamma(x_{1}x_{1}', y_{1}y_{1}', a_{2}a_{2}', b_{2}b_{2}'|x_{2}x_{2}', y_{2}y_{2}', a_{1}a_{1}', b_{1}b_{1}')\cl{E}(a_{1}a_{1}', b_{1}b_{1}'|x_{1}x_{1}', y_{1}y_{1}') \\
& = &
\sum\limits_{a_{1}a_{1}', b_{1}b_{1}'}\sum\limits_{x_{1}x_{1}', y_{1}y_{1}'}\Gamma^{*}(x_{2}x_{2}', y_{2}y_{2}', a_{1}a_{1}', b_{1}b_{1}'|x_{1}x_{1}', y_{1}y_{1}', a_{2}a_{2}', b_{2}b_{2}')\cl{E}^{*}(x_{1}x_{1}', y_{1}y_{1}'|a_{1}a_{1}', b_{1}b_{1}') \\
& = &
\Gamma^{*}[\cl{E}^{*}](x_{2}x_{2}', y_{2}y_{2}'|a_{2}a_{2}', b_{2}b_{2}')
\end{eqnarray*}
\noindent for all $x_{2}, x_{2}' \in X, y_{2}, y_{2}' \in Y, a_{2}, a_{2}' \in A, b_{2}, b_{2}' \in B$. This means $\Gamma^{*}[\cl{E}^{*}] = \Gamma[\cl{E}]^{*}$.

(ii)-(v) follow as in the proof of Theorem \ref{strat_transp} (ii)-(v), once we notice that the stochastic operator matrices defined there are bistochastic. 
\end{proof}


\section{Perfect strategies for various quantum games}\label{s_perfstrats}

For an arbitrary Hilbert space $H$, we write $\overline{H}$ for the Banach space dual of $H$; by the Riesz Representation Theorem, there exists a conjugate linear isometry $\partial: H\rightarrow \overline{H}$, such that $\partial(\xi)(\eta) = \langle \eta, \xi\rangle, \; \xi, \eta \in H$. In what follows, we will write $\overline{\bb{C}}^{X} = \overline{\bb{C}^{X}}$ for any finite set $X$. Set $\overline{\xi} = \partial(\xi)$, for $\xi \in H$. Given a linear operator $A: H\rightarrow K$, (where $H, K$ are Hilbert spaces) let $\overline{A}: \overline{K}\rightarrow \overline{H}$ be its (Banach space) dual operator. Note that $\overline{A}(\overline{\xi}) = \overline{A^{*}\xi}$, for $\xi \in K$. 

\smallskip

If $H, K$ are finite dimensional, and we take vectors $\xi \in H, \eta \in K$ let $\eta\xi^{*}: H\rightarrow K$ be the rank one operator given by $(\eta\xi^{*})(\xi') = \langle \xi', \xi\rangle\eta$. Let $\theta: \overline{H}\otimes K\rightarrow \cl{L}(H, K)$ be the linear isomorphism given by 
\begin{gather*}
	\theta(\overline{\xi}\otimes \eta) = \eta\xi^{*}, \;\;\;\; \xi \in H, \eta \in K.
\end{gather*}
\noindent If $\cl{V} \subseteq \overline{H}\otimes K$ is some subspace, we let $\tilde{\cl{V}} := \theta(\cl{V}) \subseteq \cl{L}(H, K)$ be the corresponding subspace of linear operators. 

\smallskip

Recall (see \cite[Definition 4.1]{gt_two}) that if $X, Y$ are finite sets, a quantum hypergraph over $(X, Y)$ is any subspace $\cl{U} \subseteq \overline{\bb{C}}^{X}\otimes \bb{C}^{Y}$. For a classical hypergraph $E \subseteq X\times Y$, let
\begin{gather*}
	\cl{U}_{E} = {\rm span}\{\overline{e_{x}}\otimes e_{y}: \; (x, y) \in E\}
\end{gather*}
\noindent be viewed as a quantum hypergraph over $(X, Y)$. Furthermore, following the notation established in \cite{gt_two} fix finite sets $X_{i}, Y_{i}, i = 1, 2$. Let $\cl{U}_{1} \subseteq \bb{C}^{X_{1}}\otimes \overline{\bb{C}}^{Y_{1}}, \cl{U}_{2} \subseteq \overline{\bb{C}}^{X_{2}}\otimes \bb{C}^{Y_{2}}$ be quantum hypergraphs and set
\begin{gather*}
	\cl{U}_{1} \Leftrightarrow \cl{U}_{2} := (\cl{U}_{1}\otimes \cl{U}_{2})+(\cl{U}_{1}^{\perp}\otimes \cl{U}_{2}^{\perp}).
\end{gather*}
\noindent If 
\begin{gather*}
	\sigma: \bb{C}^{X_{1}}\otimes \overline{\bb{C}}^{Y_{1}}\otimes \overline{\bb{C}}^{X_{2}}\otimes \bb{C}^{Y_{1}}\rightarrow \overline{\bb{C}}^{X_{2}}\otimes \overline{\bb{C}}^{Y_{1}}\otimes \bb{C}^{X_{1}}\otimes \bb{C}^{Y_{2}}
\end{gather*}
\noindent is the flip on the 1st and 3rd tensor legs, then we set
\begin{gather*}
	\cl{U}_{1} \leftrightarrow \cl{U}_{2} := \sigma(\cl{U}_{1} \Leftrightarrow \cl{U}_{2}).
\end{gather*}
\noindent Given a quantum hypergraph $\cl{U} \subseteq \overline{\bb{C}}^{X}\otimes \bb{C}^{Y}$, we let
\begin{gather*}
	\cl{Q}(\cl{U}) := \bigg\{\Phi: M_{X}\rightarrow M_{Y}: \; {\rm a \; quantum \; channel \; with \;} \tilde{\cl{K}}_{\Phi} \subseteq \tilde{\cl{U}}\bigg\},
\end{gather*}
\noindent where $\tilde{\cl{K}}_{\Phi}$ is the Kraus space corresponding to channel $\Phi$. For the next collection of results, we will need the following definition.
\begin{definition}[\cite{gt_two}]
Let ${\rm t} \in \{\rm loc, q, qa, qc, ns\}$, with $\cl{U}_{1} \subseteq \bb{C}^{X_{1}}\otimes \overline{\bb{C}}^{Y_{1}}$ and $\cl{U}_{2} \subseteq \overline{\bb{C}}^{X_{2}}\otimes \bb{C}^{Y_{2}}$ quantum hypergraphs. We say that $\cl{U}_{1}$ is ${\rm t}$-homomorphic to $\cl{U}_{2}$, written $\cl{U}_{1}\rightarrow_{\rm t} \cl{U}_{2}$ if there exists a quantum channel $\Phi: M_{X_{2}Y_{1}}\rightarrow M_{X_{1}Y_{2}}$ with $\Phi \in \cl{Q}_{\rm t}$ such that $\Phi \in \cl{Q}(\cl{U}_{1}\leftrightarrow \cl{U}_{2})$. 
\end{definition}
\subsection{Perfect strategies for quantum implication games}

Let $X, Y, A, B$ be finite sets. If $P \in M_{XY}, Q \in M_{AB}$ are projections, the quantum implication game (see \cite{tt}) $P \Rightarrow Q$ is the quantum non-local game $\varphi_{P\rightarrow Q}: \cl{P}_{XY}\rightarrow \cl{P}_{AB}$ given by
\begin{gather*}
	\varphi_{P\rightarrow Q}(\tilde{P}) = \begin{cases}
	Q, \;\;\;\; {\rm if\;} 0 \neq \tilde{P} \leq P,
\\	0, \;\;\;\; {\rm if\;} \tilde{P} = 0,
\\	I, \;\;\;\; {\rm otherwise.}
\end{cases}
\end{gather*}
\noindent In other words, a quantum implication game is a quantum non-local game which states that if a player is given an input state supported $P$, their output should be a state supported on $Q$. A QNS correlation $\Phi: M_{XY}\rightarrow M_{AB}$ is called a \textit{perfect strategy} for $\varphi_{P\rightarrow Q}$ if $\langle \Phi(P), Q^{\perp} \rangle = 0$. Equivalently, $\Phi$ is perfect if
\begin{gather*}
	\omega \in M_{XY}^{+} \; \text{and} \; \omega = P\omega P \Rightarrow \Phi(\omega) = Q\Phi(\omega)Q.
\end{gather*}

As both $P, Q$ are finite-rank projections, we may find orthonormal basis $\{\xi_{i}\}_{i=1}^{n} \subseteq \bb{C}^{XY}$ (resp. $\{\gamma_{\ell}\}_{\ell=1}^{m} \subseteq \bb{C}^{AB}$) for ${\rm rng}(P)$ (resp. ${\rm rng}(Q)$). We may then associate to any quantum implication game the subspace
\begin{gather*}
	\cl{U}_{P, Q} := {\rm span}\{\overline{\xi}_{i}\otimes \gamma_{\ell}: \; i \in [n], \; \ell \in [m]\},
\end{gather*}
\noindent considered as a quantum hypergraph over $(XY, AB)$. Note that if $S \in \tilde{\cl{U}}_{P, Q} \subseteq \cl{L}(\bb{C}^{XY}, \bb{C}^{AB})$, then $S = \sum_{j=1}^{t}\lambda_{j}\gamma_{j}\xi_{j}^{*}$ where $\xi_{j} \in {\rm rng}(P), \gamma_{j} \in {\rm rng}(Q)$ and $\lambda_{j} \in \bb{C}$ for $j = 1, \hdots, t$. 

If $\mathfrak{E} \subseteq {\rm QC}(M_{XY}, M_{AB})$ is a convex subset of quantum channels from $M_{XY}$ to $M_{AB}$, we let
\begin{gather}\label{imp_value}
	\omega_{\mathfrak{E}}(P, Q) = \sup\limits_{\Phi \in \mathfrak{E}}{\rm Tr}(\Phi(P)Q)
\end{gather}
\noindent be the $\mathfrak{E}$-value of the quantum implication game $P\Rightarrow Q$. Specifically, if $\mathfrak{E} = \cl{Q}_{\rm t}$ where ${\rm t} \in \{\rm loc, q, qa, qc, ns\}$ we set $\omega_{\rm t}(P, Q) = \omega_{\cl{Q}_{\rm t}}(P, Q)$; one may easily check that $\omega_{\rm t}(P, Q) = 1$ if and only if there exists a perfect ${\rm t}$-strategy $\Phi$ for $P\Rightarrow Q$. 

\begin{lemma}\label{containment_of_kraus}
A QNS correlation $\Phi: M_{XY}\rightarrow M_{AB}$ is a perfect strategy for $\varphi_{P\rightarrow Q}$ if and only if $\tilde{\cl{K}}_{\Phi} \subseteq \tilde{\cl{U}}_{P, Q}$. 
\end{lemma}
\begin{proof}
First, assume that $\Phi$ is a perfect strategy for the implication game $\varphi_{P\rightarrow Q}$. Let $\Phi(S) = \sum_{j=1}^{t}M_{j}SM_{j}^{*}$ be a Kraus decomposition for $\Phi$; as $\Phi$ is perfect, 
\begin{eqnarray*}
	0 
& = &
{\rm Tr}(\Phi(P)Q^{\perp}) \\
& = &
\sum\limits_{j=1}^{t}{\rm Tr}(M_{j}PM_{j}^{*}Q^{\perp}) \\
& = &
\sum\limits_{j=1}^{t}{\rm Tr}(Q^{\perp}M_{j}P^{2}M_{j}^{*}Q^{\perp}) \\
& = &
\sum\limits_{j=1}^{t}{\rm Tr}((Q^{\perp}M_{j}P)(Q^{\perp}M_{j}P)^{*}),
\end{eqnarray*} 
\noindent which implies $Q^{\perp}M_{j}P = 0$ for each $j = 1, \hdots, t$. Taking adjoints, we also have that $PM_{j}^{*}Q^{\perp} = 0$ for each $j = 1, \hdots, t$. If we then fix $1 \leq j \leq t$ and pick $\beta \in {\rm rng}(P), \alpha \in {\rm rng}(Q^{\perp})$, we see
\begin{gather*}
	\langle M_{j}\beta, \alpha\rangle = \langle M_{j}P\beta, Q^{\perp}\alpha\rangle = \langle Q^{\perp}M_{j}P\beta, \alpha\rangle = 0.
\end{gather*}
\noindent As $\alpha \in {\rm rng}(Q^{\perp})$ was arbitrary, this implies $M\beta \in {\rm rng}(Q)$ for $\beta \in {\rm rng}(P)$. A similar argument using the adjoint relation $PM_{j}^{*}Q^{\perp} = 0$ implies that $M_{j}^{*}\alpha \in {\rm rng}(P^{\perp})$ for each $\alpha \in {\rm rng}(Q^{\perp})$. Thus, $M_{j} \in \tilde{\cl{U}}_{P, Q}$. Unfixing our choice of $1 \leq j \leq t$, we see $\tilde{\cl{K}}_{\Phi} \subseteq \tilde{\cl{U}}_{P, Q}$.

Now assume $\tilde{\cl{K}}_{\Phi} \subseteq \tilde{\cl{U}}_{P, Q}$. By the comment before the lemma statement, one may easily see that for any $M \in \tilde{\cl{K}}_{\Phi}$ we have $Q^{\perp}MP = 0$. If $M_{1}, \hdots, M_{t}$ are Kraus operators for $\Phi$, reversing the steps in the previous paragraph we have that
\begin{gather*}
	0 = {\rm Tr}(\Phi(P)Q^{\perp}) = \langle \Phi(P), Q^{\perp}\rangle,
\end{gather*}
\noindent which shows that $\Phi$ is a perfect strategy for the game $\varphi_{P\rightarrow Q}$. 
\end{proof}

\begin{theorem}\label{perfect_strats_imp_games}
Let $X_{i}, Y_{i}, A_{i}, B_{i}, i = 1, 2$ be finite sets, $P_{i} \in M_{X_{i}Y_{i}}, Q_{i} \in M_{A_{i}B_{i}}, i = 1, 2$ projections, and ${\rm t} \in \{\rm loc, q, qa, qc, ns\}$. If $\overline{\cl{U}}_{P_{1}, Q_{1}}\rightarrow_{\rm st} \cl{U}_{P_{2}, Q_{2}}$ via $\Gamma$, and if $\cl{E}: M_{X_{1}Y_{1}}\rightarrow M_{A_{1}B_{1}}$ is a perfect ${\rm t}$-strategy for $\varphi_{P_{1}\rightarrow Q_{1}}$, then $\Gamma[\cl{E}]$ is a perfect ${\rm t}$-strategy for $\varphi_{P_{2}\rightarrow Q_{2}}$. 
\end{theorem}
\begin{proof}
Suppose that $\Gamma(T) = \sum\limits_{p=1}^{t}N_{p}TN_{p}^{*}$ and $\cl{E}(S) = \sum\limits_{r=1}^{s}M_{r}SM_{r}^{*}$. By Lemma \ref{containment_of_kraus}, as $\cl{E}$ is a perfect strategy for $\varphi_{P_{1}\rightarrow Q_{1}}$ then $M_{r} \in \tilde{\cl{U}}_{P_{1}, Q_{1}}$ for each $r = 1, \hdots, s$. By construction (see \cite{gt_two}), the Kraus operators of $\Gamma[\cl{E}]$ are given by operators $N_{p}[M_{r}]: \bb{C}^{X_{2}Y_{2}}\rightarrow \bb{C}^{A_{2}B_{2}}$, dependent on $N_{p}$ and $M_{r}$ for each $p= 1, \hdots, t, r= 1, \hdots, s$. Furthermore, as $\overline{\cl{U}}_{P_{1}, Q_{1}}\rightarrow \cl{U}_{P_{2}, Q_{2}}$ via $\Gamma$ by \cite[Theorem 5.5]{gt_two} we know that $N_{p}[M_{r}] \in \tilde{\cl{U}}_{P_{2}, Q_{2}}$ for each $p = 1, \hdots, t, r = 1, \hdots, s$. This means $\tilde{\cl{K}}_{\Gamma[\cl{E}]} \subseteq \tilde{\cl{U}}_{P_{2}, Q_{2}}$; applying Lemma \ref{containment_of_kraus} once again, we see that $\Gamma[\cl{E}]$ is a perfect strategy for $\varphi_{P_{2}\rightarrow Q_{2}}$. Furthermore, by Theorem \ref{strat_transp} we have that $\Gamma[\cl{E}] \in \cl{Q}_{\rm t}$, whenever $\Gamma \in \cl{Q}_{\rm st}$ with $\cl{E} \in \cl{Q}_{\rm t}$. 
\end{proof}

\subsection{Classical-to-quantum non-local games}

In this subsection, we consider a way to transfer strategies from classical non-local games to quantum ones, and how a perfect strategy remains perfect under this transfer even when both games are not quantum. For a fixed ${\rm t} \in \{\rm loc, q, qa, qc, ns\}$ to make sense of how we can take a classical strategy $\cl{E} \in \cl{C}_{\rm t}$ and ``apply" it to a legitimately quantum game, the most natural way is to first use \cite[Remark 8.1]{tt} to consider $\cl{E} \in \cl{Q}_{\rm t}$, and then use $\Gamma \in \cl{Q}_{\rm st}$ as defined in Definition \ref{sqns_subclass_defn} to construct a quantum strategy $\Gamma[\cl{E}]$ via Theorem \ref{strat_transp}.


\begin{lemma}\label{classical_perfect_strats}
Let $X_{1}, Y_{1}, A_{1}, B_{1}$ be finite sets, and $E \subseteq X_{1}Y_{1}\times A_{1}B_{1}$ be a non-local game. Then $\cl{E}: \cl{D}_{X_{1}Y_{1}}\rightarrow \cl{D}_{A_{1}B_{1}}$ is a perfect strategy for the game $E$ if and only if $\tilde{\cl{K}}_{\cl{E}} \subseteq \tilde{\cl{U}}_{E}$. 
\end{lemma}
\begin{proof}
First, assume $\cl{E}$ is a perfect strategy for classical non-local game $E$. Viewing $\cl{E}$ as a quantum channel, by \cite[Remark 4.3]{gt_two} 
\begin{gather*}
	\tilde{\cl{K}}_{\cl{E}} = {\rm span}\{e_{a_{1}}e_{x_{1}}^{*}\otimes e_{b_{1}}e_{y_{1}}^{*}: \; (x_{1}, y_{1}, a_{1}, b_{1}) \in {\rm supp}(\cl{E})\}.
\end{gather*}
\noindent As ${\rm supp}(\cl{E}) \subseteq E$, clearly $\tilde{\cl{K}}_{\cl{E}} \subseteq \tilde{\cl{U}}_{E}$.

Now assume $\tilde{\cl{K}}_{\cl{E}} \subseteq \tilde{\cl{U}}_{E}$. This means that for any Kraus operator $M$ of $\cl{E}$, $M \in \tilde{\cl{U}}_{E}$. Thus, $M = \sum_{j=1}^{n}\lambda_{j}e_{a_{j}}e_{x_{j}}^{*}\otimes e_{b_{j}}e_{y_{j}}^{*}$, where $(x_{j}, y_{j}, a_{j}, b_{j}) \in E$ for $j = 1, \hdots, n$. If we assume $\cl{E}(S) = \sum\limits_{\ell=1}^{m}M_{\ell}SM_{\ell}^{*}$, as
\begin{eqnarray*}
	\cl{E}(a_{1}, b_{1}|x_{1}, y_{1}) 
& = &
{\rm Tr}(\cl{E}(\epsilon_{x_{1}x_{1}}\otimes \epsilon_{y_{1}y_{1}})(\epsilon_{a_{1}a_{1}}\otimes \epsilon_{b_{1}b_{1}})) \\
& = &
\sum\limits_{\ell=1}^{m}{\rm Tr}(M_{\ell}(\epsilon_{x_{1}x_{1}}\otimes \epsilon_{y_{1}y_{1}})M_{\ell}^{*}(\epsilon_{a_{1}a_{1}}\otimes \epsilon_{b_{1}b_{1}})),
\end{eqnarray*}
\noindent where each $M_{\ell}$ is of the form previously described, it is easy to see that $\cl{E}(a_{1}, b_{1}|x_{1}, y_{1}) \neq 0$ only if $(x_{1}, y_{1}, a_{1}, b_{1}) \in E$. Thus, ${\rm supp}(\cl{E}) \subseteq E$, which means $\cl{E}$ is a perfect strategy for $E$.
\end{proof}

\begin{theorem}
Let $X_{i}, Y_{i}, A_{i}, B_{i}, i = 1, 2$ be finite sets, $E \subseteq X_{1}Y_{1}\times A_{1}B_{1}$ a classical non-local game, $P \in M_{X_{2}Y_{2}}, Q \in M_{A_{2}B_{2}}$ be projections, and let ${\rm t} \in \{\rm loc, q, qa, qc, ns\}$. If $\overline{\cl{U}}_{E} \rightarrow_{\rm st} \cl{U}_{P, Q}$ via $\Gamma$, and if $\cl{E}: \cl{D}_{X_{1}Y_{1}}\rightarrow \cl{D}_{A_{1}B_{1}}$ is a perfect ${\rm t}$-strategy for the game $E$, then $\Gamma[\cl{E}]$ is a perfect ${\rm t}$-strategy for quantum non-local game $\varphi_{P\rightarrow Q}$.
\end{theorem}
\begin{proof}
Assuming the notation and setup as in the proof of Theorem \ref{perfect_strats_imp_games}, by Lemma \ref{classical_perfect_strats} we know that $M_{r} \in \hat{\cl{U}}_{E}^{*}$ for each $r = 1, \hdots, s$. If the Kraus operators for $\Gamma[\cl{E}]$ are given by $N_{p}[M_{r}]$ for $p = 1, \hdots, t$ and $r = 1, \hdots, s$ then by \cite[Theorem 5.5]{gt_two} we know that $N_{p}[M_{r}] \in \tilde{\cl{U}}_{P, Q}$ for each $p = 1, \hdots, t$ and $r = 1, \hdots, s$. By Lemma \ref{containment_of_kraus}, this means $\Gamma[\cl{E}]$ is a perfect strategy for $\varphi_{P\rightarrow Q}$. Furthermore, by Theorem \ref{strat_transp} and Proposition \ref{st_vs_qst} we have that $\Gamma[\cl{E}] \in \cl{Q}_{\rm t}$ whenever $\Gamma \in \cl{Q}_{\rm st}$ with $\cl{E} \in \cl{C}_{\rm t}$. 
\end{proof}

\subsection{Quantum graph games} Quantum graphs, and games played on quantum graphs, have generated considerable interest in the last decade. While there are several connected concepts in the literature of what a quantum graph ``should be" (see \cite{bhinw, daws, dw}), we will use the concept discussed in \cite{synch_bhtt, bhtt, stahlke, tt}.

For a finite set $X$, let $\mathfrak{m}: \bb{C}^{X}\otimes \bb{C}^{X}\rightarrow \bb{C}$ be the map given by
\begin{gather*}
	\mathfrak{m}(\xi) = \bigg\langle \xi, \sum\limits_{x \in X}e_{x}\otimes e_{x}\bigg\rangle. 
\end{gather*}
\noindent Similarly, let $\mathfrak{f}: \bb{C}^{X}\otimes \bb{C}^{X}\rightarrow \bb{C}^{X}\otimes \bb{C}^{X}$ be the flip operation, where $\mathfrak{f}(\xi\otimes \eta) = \eta\otimes \xi$, for $\xi, \eta \in \bb{C}^{X}$.
\begin{definition}\label{quantum_graph}
A quantum graph with vertex set $X$ is a linear subspace $\cl{U} \subseteq \bb{C}^{X}\otimes \bb{C}^{X}$ which is skew- in that $\mathfrak{m}(\cl{U}) = \{0\}$- and symmetric- in that $\mathfrak{f}(\cl{U}) = \cl{U}$.
\end{definition}
\noindent \indent For the remainder of this section, for any subspace $\cl{U} \subseteq \bb{C}^{X}\otimes \bb{C}^{X}$ we denote by $P_{\cl{U}}$ the orthogonal projection onto $\cl{U}$. If $G$ is a classical graph on vertex set $X$, there is a corresponding quantum graph $\cl{U}_{G}$ given by the subspace 
\begin{gather*}
	\cl{U}_{G} := {\rm span}\{e_{x}\otimes e_{y}: \; x\sim y\},
\end{gather*}
\noindent where we write $P_{G} = P_{\cl{U}_{G}}$. If $\cl{U} \subseteq \bb{C}^{X}\otimes \bb{C}^{X}$ and $\cl{V} \subseteq \bb{C}^{A}\otimes \bb{C}^{A}$ are quantum graphs, the \textit{quantum graph homomorphism game} $\cl{U}\rightarrow \cl{V}$ is the quantum implication game $\varphi_{P_{\cl{U}}\rightarrow P_{\cl{V}}}$. As mentioned, QNS correlation $\Gamma: M_{XX}\rightarrow M_{AA}$ is perfect for $\cl{U}\rightarrow \cl{V}$ if 
\begin{gather*}
	\omega \in M_{XX}^{+} \; \text{and} \; P_{\cl{U}}\omega P_{\cl{U}} \Rightarrow \Phi(\omega) = P_{\cl{V}}\Phi(\omega)P_{\cl{V}}.
\end{gather*}
\noindent In the event that $X = A$, then $\Phi$ is a perfect strategy for the \textit{quantum graph isomorphism game} $\cl{U} \cong \cl{V}$ if $\Phi$ is a bicorrelation, with $\Phi$ a perfect strategy for $\cl{U}\rightarrow \cl{V}$ and $\Phi^{*}$ a perfect strategy for $\cl{V}\rightarrow \cl{U}$. 
\begin{theorem}\label{graph_isomorphism_game}
Let $\cl{U}_{i} \subseteq \bb{C}^{X_{i}X_{i}}, \cl{V}_{i} \subseteq \bb{C}^{A_{i}A_{i}}, i = 1, 2$ be quantum graphs, with $P_{i} = P_{\cl{U}_{i}}, Q_{i} = P_{\cl{V}_{i}}, i = 1, 2$ their corresponding projections. 
\begin{itemize}
	\item[(i)] If $\cl{E}: M_{X_{1}X_{1}}\rightarrow M_{A_{1}A_{1}}$ is a perfect strategy for the quantum graph homomorphism game $\cl{U}_{1} \rightarrow_{\rm t} \cl{V}_{1}$ and $\overline{\cl{U}}_{P_{1}, Q_{1}} \rightarrow_{\rm st} \cl{U}_{P_{2}, Q_{2}}$ via $\Gamma$, then $\Gamma[\cl{E}]$ is a perfect strategy for the quantum graph homomorphism game $\cl{U}_{2} \rightarrow_{\rm t} \cl{V}_{2}$;
	\item[(ii)] If $X_{i} = A_{i}, i = 1, 2$ and $\cl{E}: M_{X_{1}X_{1}}\rightarrow M_{X_{1}X_{1}}$ is a perfect strategy for the quantum graph isomorphism game $\cl{U}_{1} \cong_{\rm t} \cl{V}_{1}$ and $\overline{\cl{U}}_{P_{1}, Q_{1}} \cong_{\rm st} \cl{U}_{P_{2}, Q_{2}}$ via $\Gamma$, then $\Gamma[\cl{E}]$ is a perfect strategy for the quantum graph isomorphism game $\cl{U}_{2} \cong_{\rm t} \cl{V}_{2}$. 
\end{itemize}
\end{theorem}
\begin{proof}
(i) This follows as a special consequence of Theorem \ref{perfect_strats_imp_games}.

(ii) Assume the notation and setup as in the proof of Theorem \ref{perfect_strats_imp_games}; by the aforementioned theorem, we already know that $\Gamma[\cl{E}]$ is a perfect strategy for the graph homomorphism game $\cl{U}_{2}\rightarrow \cl{V}_{2}$. To show that $\Gamma[\cl{E}]^{*}$ is a perfect strategy for $\cl{V}_{2}\rightarrow \cl{U}_{2}$, first note that $\Gamma[\cl{E}]^{*} = \Gamma^{*}[\cl{E}^{*}]$ (by construction of simulated channels). Furthermore, by definition of $\cl{U}_{P_{2}, Q_{2}}$ we may easily verify that $\cl{U}_{Q_{2}, P_{2}} = \overline{\cl{U}}_{P_{2}, Q_{2}}$. With the use of \cite[Lemma 2.1]{gt_two}, and as $\theta^{-1}(N_{p}[M_{r}]) \in \cl{U}_{P_{2}, Q_{2}}$, we have that $\theta^{-1}((N_{p}[M_{r}])^{*}) = \theta^{-1}(N_{p}^{*}[M_{r}^{*}]) \in \overline{\cl{U}}_{P_{2}, Q_{2}} = \cl{U}_{Q_{2}, P_{2}}, p = 1, \hdots, t$ and $r = 1, \hdots, s$. Thus, by Lemma \ref{containment_of_kraus} once more we see $\tilde{\cl{K}}_{\Gamma[\cl{E}]^{*}} \subseteq \tilde{\cl{U}}_{Q_{2}, P_{2}}$. This means $\Gamma[\cl{E}]^{*}$ is a perfect strategy for $\cl{V}_{2}\rightarrow \cl{U}_{2}$. That $\Gamma[\cl{E}], \Gamma[\cl{E}]^{*}$ are both in $\cl{Q}_{\rm t}$ follows from Theorem \ref{bi_strat_transp}. Together, these show $\cl{U}_{2} \cong_{\rm t} \cl{V}_{2}$ via $\Gamma[\cl{E}]$, as claimed. 
\end{proof}


\section{Characterization of SQNS correlations and applications to concurrent games}\label{s_char_sqns}

In the final section, we wish to link strongly quantum no-signalling correlations to the (multi-variate) tensor product of operator systems. In order to do so, we first briefly recall some facts about operator systems and related constructions. If $\cl{S}$ and $\cl{T}$ are operator systems, we call $\cl{S}$ and $\cl{T}$ \textit{completely order isomorphic} and write $\cl{S} \cong_{\rm c.o.i} \cl{T}$ if there exists a unital completely positive bijection $\phi: \cl{S}\rightarrow \cl{T}$ with completely positive inverse. We write $\cl{S} \subseteq_{\rm c.o.i} \cl{T}$ if $\cl{S} \subseteq \cl{T}$ and the inclusion map $\iota: \cl{S}\hookrightarrow \cl{T}$ is a complete order isomorphism onto its range. The three main types of operator system tensor products that will be used in the sequel are given as follows:
\begin{itemize}
	\item[(i)] the \textit{minimal} operator system tensor product $\cl{S}\otimes_{\rm min} \cl{T}$ arises from viewing $\cl{S}\otimes \cl{T}$ as a subspace of $\cl{B}(H\otimes K)$, when we concretely realize $\cl{S} \subseteq \cl{B}(H)$ and $\cl{T} \subseteq \cl{B}(K)$ for Hilbert spaces $H, K$;
	\item[(ii)] the \textit{commuting} tensor product $\cl{S}\otimes_{\rm c} \cl{T}$ has the smallest family of matricial cones which make the maps $\phi\cdot \psi$, where $\phi: \cl{S}\rightarrow \cl{B}(H)$ and $\psi: \cl{T}\rightarrow \cl{B}(H)$ are completely positive with commuting ranges, completely positive; note that $(\phi\cdot \psi)(x\otimes y) = \phi(x)\psi(y), x \in \cl{S}, y \in \cl{T}$;
	\item[(iii)] the \textit{maximal} operator system tensor product $\cl{S}\otimes_{\rm max} \cl{T}$ has matricial cones generated by the elementary tensors of the form $S\otimes T$, where $S \in M_{n}(\cl{S})^{+}$ and $T \in M_{m}(\cl{T})^{+}, n, m \in \bb{N}$.
\end{itemize}
\noindent More details about each tensor product may be found in \cite{kptt}; the construction of the multivariate tensor product of each type, with explicit descriptions of their matricial cones, may be found in \cite[Section 7]{gt_one}. 

We recall the notion of a coproduct of operator systems: if $\cl{S}$ and $\cl{T}$ are two operator systems, their coproduct $\cl{S} \oplus_{1} \cl{T}$ is the unique (up to isomorphism) operator system equipped with complete order embeddings $\iota_{\cl{S}}: \cl{S}\rightarrow \cl{S}\oplus_{1} \cl{T}$ and $\iota_{\cl{T}}: \cl{T}\rightarrow \cl{S}\oplus_{1}\cl{T}$ which satisfies the following unversal property: For each ucp map $\phi: \cl{S}\rightarrow \cl{R}$ and $\psi: \cl{T}\rightarrow \cl{R}$, where $\cl{R}$ is an operator system, there exists a unique ucp map $\varphi: \cl{S}\oplus_{1}\cl{T}\rightarrow \cl{R}$ such that $\varphi(\iota_{\cl{S}}(s)) = \phi(s)$ and $\varphi(\iota_{\cl{T}}(t)) = \psi(t)$ for every $s \in \cl{S}, t \in \cl{T}$. For more properties of the coproduct of operator systems, we refer the reader to \cite[Section 8]{kavruk}.

We also will need to recall the universal operator system for stochastic operator matrices, introduced in \cite{tt}. A ${\it ternary \; ring \; of \; operators \; (TRO)}$ is a subspace $\cl{V} \subseteq \cl{B}(H, K)$ for some Hilbert spaces $H$ and $K$, such that $ST^{*}R \in \cl{V}$ whenever $S, T, R \in \cl{V}$ (see \cite[Section 4.3]{bm} or \cite{zettl}). Let $X$ and $A$ be finite sets, and let $\cl{V}_{X, A}$ be the universal TRO generated by the entries $v_{a, x}$ of a block operator isometry $V = (v_{a, x})_{a \in A, x \in X}$. That is, $\cl{V}_{X, A}$ is the universal TRO with generators $v_{a, x}, a \in A, x \in X$, ternary operator $[\cdot, \cdot, \cdot]: \cl{V}_{X, A}\times \cl{V}_{X, A}\times \cl{V}_{X, A}\rightarrow \cl{V}_{X, A}$ given by $[u, v, w] = uv^{*}w$ for $u, v, w \in \cl{V}_{X, A}$, and relations
\begin{gather*}
	\sum\limits_{a \in A}[v_{a'', x''}, v_{a, x}, v_{a, x'}] = \delta_{x, x'}v_{a'', x''}, \;\;\;\; x, x', x'' \in X, a'' \in A.
\end{gather*}
Let $\mathfrak{C}_{X, A}$ be the unital $*$-algebra, generated by the set $\{v_{a, x}^{*}v_{a', x'}: \; x, x' \in X, a, a' \in A\}$, and $\cl{C}_{X, A}$ be the ${\it right \; C^{*}-algebra}$ of $\cl{V}_{X, A}$; up to a $*$-isomorphism, we may view
\begin{gather*}
	\cl{C}_{X, A} = \overline{{\rm span}\{\theta(S)^{*}\theta(T): \; S, T \in \cl{V}_{X, A}\}}
\end{gather*}
\noindent for any faithful ternary representation $\theta: \cl{V}_{X, A}\rightarrow \cl{B}(H, K)$ (where $H, K$ are Hilbert spaces). Set
\begin{gather*}
	e_{x, x', a, a'} = v_{a, x}^{*}v_{a', x'}, \;\;\;\; x, x' \in X, \; a, a' \in A
\end{gather*}
\noindent where the latter is considered either as an element in $\mathfrak{C}_{X, A}$ or $\cl{C}_{X, A}$, depending on the context. The \textit{Brown-Cuntz operator system} (see \cite{tt}) inside $\cl{C}_{X, A}$ is given by
\begin{gather*}
	\cl{T}_{X, A} := {\rm span}\{e_{x, x', a, a'}: \; x, x' \in X, \; a, a' \in A\}.
\end{gather*}
\noindent We also will consider the space
\begin{gather*}
	\cl{S}_{X, A} := {\rm span}\{e_{x, x, a, a}: \; x \in X, a \in A\},
\end{gather*}
\noindent viewed as an operator subsystem inside $\cl{T}_{X, A}$. To help distinguish between operator systems $\cl{T}_{X, A}$ and $\cl{T}_{Y, B}$, we will denote the canonical generators of $\cl{T}_{X, A}$ by $e_{x, x', a, a'}, x, x' \in X, a, a' \in A$ and $\cl{T}_{Y, B}$ by $f_{y, y', b, b'}, y, y' \in Y, b, b' \in B$. 

Similarly, there are canonical operator algebras and operator systems which corresponding to bistochastic operator matrices (first introduced in \cite[Section 3]{bhtt}); these will be needed to show an analogous result to Proposition \ref{dilate_ssop}. Let $\cl{V}_{X}$ be the universal TRO generated by the entries $v_{a, x}$ of a block operator bi-isometry $V = (v_{a, x})_{a \in A, x \in X}$ (see \cite{bhtt}). Let $\cl{C}_{X}$ be the right ${\rm C}^{*}$-algebra of $\cl{V}_{X}$, set $e_{x, x', a, a'} = v_{a, x}^{*}v_{a, x}$ (where the latter is considered as an element of $\cl{C}_{X}$) and 
\begin{gather*}
	\cl{T}_{X} := {\rm span}\{e_{x, x', a, a'}: \; x, x' \in X, a, a' \in A\}
\end{gather*}
\noindent be viewed as an operator system in $\cl{C}_{X}$. Furthermore, we let
\begin{gather*}
	\cl{S}_{X} := {\rm span}\{e_{x, x, a, a}: \; x \in X, a \in A\}
\end{gather*}
\noindent be viewed as an operator subsystem of $\cl{T}_{X}$.

\begin{proposition}\label{dilate_ssop}
If $P = (P_{xx', yy'}^{aa', bb'})_{xx', yy', aa', bb'}$ is a dilatable strongly stochastic operator matrix acting on the Hilbert space $H$, then there exists a unital completely positive map $\gamma: \cl{T}_{X, A}\otimes_{\rm c} \cl{T}_{Y, B}\rightarrow \cl{B}(H)$ such that $\gamma(e_{x, x', a, a'}\otimes f_{y, y', b, b'}) = P_{xx', yy'}^{aa', bb'}$. Conversely, if $\gamma: \cl{T}_{X, A}\otimes_{\rm c}\cl{T}_{Y, B}\rightarrow \cl{B}(H)$ is a unital completely positive map, then $\big(\gamma(e_{x, x', a, a'}\otimes f_{y, y', b, b'})\big)_{xx', yy', aa', bb'}$ is a dilatable strongly stochastic operator matrix. 
\end{proposition}
\begin{proof}
Let $K$ be a Hilbert space, $V: H\rightarrow K$ be an isometry, and
$$ (E_{xx', aa'})_{x, x' \in X, a, a' \in A}, \;\;\;\; (F_{yy', bb'})_{y, y' \in Y, b, b' \in B}$$
be mutually commuting stochastic operator matrices on $K$ satisfying (\ref{dilate_eqn}). Define linear maps $\phi: \cl{T}_{X, A}\rightarrow \cl{B}(K)$ (resp. $\psi: \cl{T}_{Y, B}\rightarrow \cl{B}(K)$) by $\phi(e_{x, x', a, a'}) = E_{xx', aa'}$ (resp. $\psi(f_{y, y', b, b'}) = F_{yy', bb'}$). By \cite[Theorem 5.2]{tt}, both $\phi$ and $\psi$ are unital and completely positive. By the definition of the commuting tensor product of operator spaces, the map $\phi\cdot \psi: \cl{T}_{X, A}\otimes_{\rm c}\cl{T}_{Y, B}\rightarrow \cl{B}(K)$ given by $(\phi\cdot \psi)(u\otimes v) = \phi(u)\psi(v)$ is (unital and) completely positive as well. Set
\begin{gather*}
	\gamma(w) = V^{*}(\phi\cdot \psi)(w)V, \;\;\;\; w \in \cl{T}_{X, A}\otimes_{\rm c}\cl{T}_{Y, B};
\end{gather*}
\noindent it is easy to verify that $\gamma$ is unital and completely positive. Furthermore, we have $\gamma(e_{x, x', a, a'}\otimes f_{y, y', b, b'}) = P_{xx', yy'}^{aa', bb'}, x, x' \in X, y, y' \in Y, a, a' \in A, b, b' \in B$. 

Conversely, suppose that $\gamma: \cl{T}_{X, A}\otimes_{\rm c}\cl{T}_{Y, B}\rightarrow \cl{B}(H)$ is a unital completely positive map. By \cite[Corollary 5.3]{tt} and \cite[Theorem 6.4]{kptt}, $\cl{T}_{X, A}\otimes_{\rm c} \cl{T}_{Y, B} \subseteq_{\rm c.o.i} \cl{C}_{X, A}\otimes_{\rm max} \cl{C}_{Y, B}$. Apply Arveson's Extension Theorem to obtain a completely positive map $\tilde{\gamma}: \cl{C}_{X, A}\otimes_{\rm max}\cl{C}_{Y, B}\rightarrow \cl{B}(H)$ extending $\gamma$. If we then apply Stinespring's Theorem, we may write
\begin{gather*}
	\tilde{\gamma}(w) = V^{*}\pi(w)V, \;\;\;\; w \in \cl{C}_{X, A}\otimes_{\rm max} \cl{C}_{Y, B},
\end{gather*}
\noindent where $\pi: \cl{C}_{X, A}\otimes_{\rm max} \cl{C}_{Y, B}\rightarrow \cl{B}(K)$ is a $*$-representation on some Hilbert space $K$, and $V:H\rightarrow K$ is an isometry. Set $E_{xx', aa'} = \pi(e_{x, x', a, a'}\otimes 1), F_{yy', bb'} = \pi(1\otimes f_{y, y', b, b'})$ for $x, x' \in X, a, a' \in A, y, y' \in Y$ and $b, b' \in B$. Doing so gives us a dilatable representation of the matrix $(\gamma(e_{x, x', a, a'}\otimes f_{y, y', b, b'}))_{xx', yy', aa', bb'}$. 
\end{proof}

\begin{remark}
\rm By Remark \ref{ns_generalization} and \cite[Remark 5.3]{gt_one}, there exist strongly stochastic operator matrices which are not dilatable.
\end{remark}

\begin{proposition}\label{dilate_bissop}
If $P = (P_{xx', yy'}^{aa', bb'})_{xx', yy', aa', bb'}$ is a dilatable strongly bistochastic operator matrix acting on the Hilbert space $H$, then there exists a unital completely positive map $\gamma: \cl{T}_{X}\otimes_{\rm c} \cl{T}_{Y}\rightarrow \cl{B}(H)$ such that $\gamma(e_{x, x', a, a'}\otimes f_{y, y', b, b'}) = P_{xx', yy'}^{aa', bb'}$. Conversely, if $\gamma: \cl{T}_{X}\otimes_{\rm c} \cl{T}_{Y}\rightarrow \cl{B}(H)$ is a unital completely positive map, then $(\gamma(e_{x, x', a, a'}\otimes f_{y, y', b, b'}))_{xx', yy', aa', bb'}$ is a dilatable strongly bistochastic operator matrix.
\end{proposition}
\begin{proof}
Using \cite[Theorem 3.4]{bhtt} and \cite[Theorem 6.4]{kptt}, argue exactly as in Proposition \ref{dilate_ssop} but with bistochastic operator matrices in place of stochastic operator matrices and $\cl{T}_{X}$ (resp. $\cl{T}_{Y}$) in place of $\cl{T}_{X, A}$ (resp. $\cl{T}_{Y, B}$).
\end{proof}


\subsection{Representations of SQNS correlations via operator systems}

Let $\cl{S}$ be an operator system. Recall that the \textit{universal $C^{*}$-cover} of $\cl{S}$ (see \cite{kw}) is the pair $(C_{u}^{*}(\cl{S}), \iota)$ where $C_{u}^{*}(\cl{S})$ is a unital ${\rm C}^{*}$-algebra, $\iota: \cl{S}\rightarrow C_{u}^{*}(\cl{S})$ is a unital complete order embedding such that $\iota(\cl{S})$ generates $C_{u}^{*}(\cl{S})$ as a ${\rm C}^{*}$-algebra, and whenever $H$ is a Hilbert space with $\phi: \cl{S}\rightarrow \cl{B}(H)$ a unital completely positive map, there exists a $*$-representation $\pi_{\phi}: C_{u}^{*}(\cl{S})\rightarrow \cl{B}(H)$ such that $\pi_{\phi} \circ \iota = \phi$. We will briefly introduce a multivariate extension of both the maximal and the commuting tensor product (as discussed in the beginning of this section) in the category of operator systems. If $\cl{S}_{1}, \hdots, \cl{S}_{k}$ are operator systems, as the maximal tensor product of operator systems is associative (via \cite[Theorem 5.5]{kptt}) we may give an unambiguous meaning to the $k$-fold maximal tensor product
\begin{gather*}
	{\rm max}-\otimes_{j=1}^{k}\cl{S}_{j} := \cl{S}_{1}\otimes_{\rm max} \cdots \otimes_{\rm max} \cl{S}_{k}.
\end{gather*}
\noindent The $k$-fold commuting tensor product was initially defined in \cite[Section 7]{gt_one}: if $H$ is a Hilbert space and $\phi_{j}: \cl{S}_{j}\rightarrow \cl{B}(H)$ are completely positive maps for $j = 1, \hdots, k$, we call the family $(\phi_{j})_{j=1}^{k}$ \textit{commuting} if $\phi_{i}$ and $\phi_{j}$ have mutually commuting ranges whenever $i \neq j$, for $1 \leq i, j \leq k$. Given a commuting family $(\phi_{j})_{j=1}^{k}$, we define the linear map $\prod_{j=1}^{k}\phi_{j}: \otimes_{j=1}^{k}\cl{S}_{j}\rightarrow \cl{B}(H)$ via 
\begin{gather*}
	\bigg(\prod_{j=1}^{k}\phi_{j}\bigg)(\otimes_{j=1}^{k}u_{j}) := \prod_{j=1}^{k}\phi_{j}(u_{j}), \;\;\;\; u_{j} \in \cl{S}_{j}, j \in [k].
\end{gather*}
\noindent The positive cones for ${\rm c}-\otimes_{j=1}^{k}\cl{S}_{j}$ are then determined by all elements $u \in M_{n}(\otimes_{j=1}^{k}\cl{S}_{j})$ which are positive in $M_{n}(\cl{B}(H))$ under all mutually commuting families and for all choices of Hilbert space $H$.

The following result shows how we can understand the structure of the multivariate commuting tensor product of operator systems when we view them inside the maximal tensor product of their universal ${\rm C}^{*}$-covers. 

\begin{theorem}\label{op_sys_incl_to_env}
Let $\cl{S}_{1}, \hdots, \cl{S}_{k}, k \in \bb{N}$ be operator systems. The operator system arising from the inclusion of $\otimes_{j=1}^{k}-\cl{S}_{j}$ into ${\rm max}-\otimes_{j=1}^{k}C_{u}^{*}(\cl{S}_{j})$ coincides with ${\rm c}-\otimes_{j=1}^{k}\cl{S}_{j}$. 
\end{theorem}
\begin{proof}

For the sake of brevity, set $\cl{S} = {\rm c}-\otimes_{j=1}^{k}\cl{S}_{j}$. First, suppose $u \in M_{n}(\cl{S})^{+}$ for some $n \in \bb{N}$. We wish to show that $u \in M_{n}({\rm max}-\otimes_{j=1}^{k}C_{u}^{*}(\cl{S}_{j}))^{+}$. By \cite[Lemma 4.1]{kptt}, it is sufficient to prove that $\phi^{(n)}(u) \geq 0$ for each unital completely positive $\phi: {\rm max}-\otimes_{j=1}^{k}C_{u}^{*}(\cl{S}_{j})\rightarrow \cl{B}(H)$. By Stinespring's Theorem, we may also without loss of generality assume $\phi$ is a $*$-homomorphism. By \cite[Proposition 7.4]{gt_one}, associativity, and the universal property of the maximal tensor product of ${\rm C}^{*}$-algebras, each such $\phi$ is equivalent to $\prod_{j=1}^{k}\pi_{j}$, where $\pi_{j}: C_{u}^{*}(\cl{S}_{j})\rightarrow \cl{B}(H)$ is a $*$-homomorphism for $j = 1, \hdots, k$ and all have mutually commuting ranges. As the restriction of $\pi_{j}$ to $\cl{S}_{j}$ remains completely positive for $j = 1, \hdots, k$, the result follows. 

Conversely, let $\tau$ be the operator system structure on $\otimes_{j=1}^{k}\cl{S}_{j}$ arising from the inclusion $\otimes_{j=1}^{k}\cl{S}_{j} \subseteq {\rm max}-\otimes_{j=1}^{k}C_{u}^{*}(\cl{S}_{j})$. Suppose that $u \in M_{n}(\tau-\otimes_{j=1}^{k}\cl{S}_{j})^{+}$, with $\phi_{j}: \cl{S}_{j}\rightarrow \cl{B}(H)$ completely positive maps with mutually commuting ranges for $j = 1, \hdots, k$. By the definition of $C_{u}^{*}(\cl{S}_{j})$, there exists unique $*$-homomorphisms $\pi_{j}: C_{u}^{*}(\cl{S}_{j})\rightarrow \cl{B}(H)$ extending $\phi_{j}$, for $j = 1, \hdots, k$. As $\cl{S}_{j}$ generates $C_{u}^{*}(\cl{S}_{j})$ as a ${\rm C}^{*}$-algebra, we may then conclude the ranges of $\pi_{j}$ are mutually commuting as well for $j = 1, \hdots, k$. Thus, $(\otimes_{j=1}^{k}\phi_{j})^{(n)}(u) \geq 0$, which implies $u \in M_{n}(\cl{S})^{+}$- completing the proof.
\end{proof}

Recall (see \cite{ptt}) that for any Archimedean ordered unit (AOU) space $V$, there exists a unique operator system ${\rm OMIN}(V)$ (respectively, ${\rm OMAX}(V)$) with underlying space $V$, called the \textit{minimal} (respectively, \textit{maximal}) operator system of $V$ which has the universal property that every positive map $\phi: \cl{T}\rightarrow V$ (respectively, $\phi: V\rightarrow \cl{T}$) where $\cl{T}$ is an operator system, is completely positive from $\cl{T} \rightarrow {\rm OMIN}(V)$ (respectively, ${\rm OMAX}(V)\rightarrow \cl{T}$). 
\begin{lemma}\label{omax_state_extension}
Let $V_{1}, \hdots, V_{k}$, $k \in \bb{N}$ be finite dimensional AOU spaces, with units $e_{j}, j = 1, \hdots, k$ respectively. An element $u \in {\rm max}-\otimes_{j=1}^{k}{\rm OMAX}(V_{j})$ is positive if and only if $u = \sum\limits_{i=1}^{m}v_{i}^{(1)}\otimes\hdots \otimes v_{i}^{(k)}$, for some $v_{i}^{(j)} \in V_{j}^{+}, i = 1, \hdots, m$ and $j = 1, \hdots, k$. 
\end{lemma}
\begin{proof}
This proof relies on similar ideas as the proof of \cite[Lemma 6.6]{tt}; we include the details for the convenience of the reader. We only consider the case when $k = 3$; all others will follow similarly. Let $D$ be the set containing all sums of elementary tensors $v_{1}\otimes v_{2}\otimes v_{3}$ with $v_{i} \in V_{i}^{+}, \; i = 1, 2, 3$. We claim that if, for every $\epsilon > 0$, there exists $u_{\epsilon} \in D$ such that $\|u_{\epsilon}\|\rightarrow 0$ as $\epsilon \rightarrow 0$ and $u+u_{\epsilon} \in D$ for each $\epsilon > 0$, then $u \in D$. We may, without loss of generality, further assume that $\|u_{\epsilon}\| \leq 1$ for all $\epsilon > 0$. Set $K = 2{\rm dim}(V_{1}){\rm dim}(V_{2}){\rm dim}(V_{3})+1$ and, using Carathéodory's Theorem, start by writing
\begin{gather*}
	u+u_{\epsilon} = \sum\limits_{j=1}^{K}v_{1, j}^{(\epsilon)}\otimes v_{2, j}^{(\epsilon)}\otimes v_{3, j}^{(\epsilon)},
\end{gather*}
\noindent where $v_{i, j}^{(\epsilon)} \in V_{i}^{+}$ and $\|v_{1, j}^{(\epsilon)}\| = \|v_{2, j}^{(\epsilon)}\| = \|v_{3, j}^{(\epsilon)}\|$ for $i = 1, 2, 3$, $j = 1, \hdots, K$ and all $\epsilon > 0$. As $v_{1, j}^{(\epsilon)}\otimes v_{2, j}^{(\epsilon)} \otimes v_{3, j}^{(\epsilon)} \leq u+u_{\epsilon}$ and $\|u+u_{\epsilon}\| \leq \|u\|+1$ for all $\epsilon > 0$, we get that $\|v_{i, j}^{(\epsilon)}\| \leq \sqrt[3]{\|u\|+1}$ for $i = 1, 2, 3$ and $j = 1, \hdots, K$. Using the fact that all of our AOU spaces are finite-dimensional, by compactness we can also assume $v_{i, j}^{(\epsilon)} \rightarrow v_{i, j}$ as $\epsilon \rightarrow 0$ for $i = 1, 2, 3$ and all $j = 1, \hdots, K$. We may then conclude that $u = \sum_{j=1}^{K}v_{1, j}\otimes v_{2, j}\otimes v_{3, j} \in D$.

Let
\begin{gather}\label{elementary_positive_tensors}
	S_{0} = \sum\limits_{p=1}^{\ell}a_{p}\otimes v_{p}^{(1)}, \;\;\;\; T_{0} = \sum\limits_{q = 1}^{s}b_{q}\otimes v_{q}^{(2)}, \;\;\;\; U_{0} = \sum\limits_{r = 1}^{t}c_{r}\otimes v_{r}^{(3)},
\end{gather}
for some $a_{p} \in M_{n}^{+}, v_{p}^{(1)} \in V_{1}^{+}$, $p = 1, \hdots, \ell$, $b_{q} \in M_{m}^{+}, v_{q}^{(2)} \in V_{2}^{+}, q = 1, \hdots, s$ and $c_{r} \in M_{d}^{+}, v_{r}^{(3)} \in V_{3}^{+}, r = 1, \hdots, t$. If $\alpha \in M_{1, nmd}$, then
\begin{gather*}
	\alpha(S_{0}\otimes T_{0}\otimes U_{0})\alpha^{*} = \sum\limits_{p=1}^{\ell}\sum\limits_{q=1}^{s}\sum\limits_{r=1}^{t}(\alpha(a_{p}\otimes b_{q}\otimes c_{r})\alpha^{*})v_{p}^{(1)}\otimes v_{q}^{(2)}\otimes v_{r}^{(3)} \in D.
\end{gather*}
\noindent Now suppose first that $S \in M_{n}({\rm OMAX}(V_{1}))^{+}$ and $\alpha \in M_{1, nmd}$. By the description of the multivariate maximal tensor product as provided in \cite[Proposition 7.2]{gt_one}, if $\epsilon > 0$ then $S+\epsilon 1_{n}$ has the form of $S_{0}$ as in (\ref{elementary_positive_tensors}). Therefore, 
\begin{gather*}
	\alpha(S\otimes T_{0}\otimes U_{0})\alpha^{*}+\epsilon \alpha(1_{n}\otimes T_{0}\otimes U_{0})\alpha^{*} = \alpha((S+\epsilon 1_{n})\otimes T_{0}\otimes U_{0})\alpha^{*} \in D.
\end{gather*}
\noindent Since $\alpha(1_{n}\otimes T_{0}\otimes U_{0})\alpha^{*} \in D$ (using \cite{ptt}), using our initial arguments we conclude that
\begin{gather*}
	\alpha(S\otimes T_{0}\otimes U_{0})\alpha^{*} \in D.
\end{gather*}
\noindent If we now pick $T \in M_{m}({\rm OMAX}(V_{2}))^{+}$ and $U \in M_{d}({\rm OMAX}(V_{3}))^{+}$, using similar arguments as before we may conclude that $\alpha(S\otimes T\otimes U)\alpha^{*} \in D$. 

Let $u \in {\rm max}-\otimes_{j=1}^{3}{\rm OMAX}(V_{j})$ be positive; by the definition of the multivariate maximal tensor product \cite{gt_one}, for every $\epsilon > 0$ there exists $n, m, d \in \bb{N}, S \in M_{n}({\rm OMAX}(V_{1}))^{+}, T \in M_{m}({\rm OMAX}(V_{2}))^{+}, U \in M_{d}({\rm OMAX}(V_{3}))^{+}$ and $\alpha \in M_{1, nmd}$ such that $u+\epsilon 1 = \alpha(S\otimes T\otimes U)\alpha^{*}$. By the previous and first paragraphs, this implies $u \in D$. 
\end{proof}

For a linear functional
\begin{gather*}
	s: \cl{T}_{X_{2}, X_{1}}\otimes \cl{T}_{Y_{2}, Y_{1}}\otimes \cl{T}_{A_{1}, A_{2}}\otimes \cl{T}_{B_{1}, B_{2}}\rightarrow \bb{C}, 
\end{gather*}
\noindent let $\Gamma_{s}: M_{X_{2}Y_{2}A_{1}B_{1}}\rightarrow M_{X_{1}Y_{1}A_{2}B_{2}}$ be the linear map whose Choi matrix coincides with
\begin{gather}\label{choi_matrix_formula}
	\big(s(e_{x_{2}, x_{2}', x_{1}, x_{1}'}\otimes e_{y_{2}, y_{2}', y_{1}, y_{1}'}\otimes e_{a_{1}, a_{1}', a_{2}, a_{2}'}\otimes e_{b_{1}, b_{1}', b_{2}, b_{2}'})\big)_{x_{2}x_{2}', y_{2}y_{2}', a_{1}a_{1}', b_{1}b_{1}'}^{x_{1}x_{1}', y_{1}y_{1}', a_{2}a_{2}', b_{2}b_{2}'}.
\end{gather}
\begin{theorem}\label{state_correspondence}
The map $s \mapsto \Gamma_{s}$ is an affine isomorphism from 
\begin{itemize}
	\item[(i)] the state space of $\cl{T}_{X_{2}, X_{1}}\otimes_{\rm max} \cl{T}_{Y_{2}, Y_{1}}\otimes_{\rm max} \cl{T}_{A_{1}, A_{2}}\otimes_{\rm max} \cl{T}_{B_{1}, B_{2}}$ onto $\cl{Q}_{\rm sns}$;
	\item[(ii)] the state space of $\cl{T}_{X_{2}, X_{1}}\otimes_{\rm c} \cl{T}_{Y_{2}, Y_{1}}\otimes_{\rm c} \cl{T}_{A_{1}, A_{2}}\otimes_{\rm c} \cl{T}_{B_{1}, B_{2}}$ onto $\cl{Q}_{\rm sqc}$;
	\item[(iii)] the state space of $\cl{T}_{X_{2}, X_{1}}\otimes_{\rm min} \cl{T}_{Y_{2}, Y_{1}}\otimes_{\rm min} \cl{T}_{A_{1}, A_{2}}\otimes_{\rm min} \cl{T}_{B_{1}, B_{2}}$ onto $\cl{Q}_{\rm sqa}$.
	\item[(iv)] the state space of ${\rm OMIN}(\cl{T}_{X_{2}, X_{1}})\otimes_{\rm min} {\rm OMIN}(\cl{T}_{Y_{2}, Y_{1}})\otimes_{\rm min} {\rm OMIN}(\cl{T}_{A_{1}, A_{2}})\otimes_{\rm min} {\rm OMIN}(\cl{T}_{B_{1}, B_{2}})$ onto $\cl{Q}_{\rm sloc}$. 
\end{itemize}
\end{theorem}
\begin{proof}
(i) Let $\Gamma \in \cl{Q}_{\rm sns}$. If 
\begin{gather*}
	C = \big(C_{x_{2}x_{2}', y_{2}y_{2}', a_{1}a_{1}', b_{1}b_{1}'}^{x_{1}x_{1}', y_{1}y_{1}', a_{2}a_{2}', b_{2}b_{2}'}\big)
\end{gather*}
\noindent is the Choi matrix of $\Gamma$ (where the indices above range over the corresponding sets $X_{i}, Y_{i}, A_{i}, B_{i}, i=1, 2$), then $C \in M_{X_{2}Y_{2}A_{1}B_{1}, X_{1}Y_{1}A_{2}B_{2}}^{+}$. We also note that condition (\ref{pt_one}) implies that for $y_{i}, y_{i}' \in Y_{i}, a_{i}, a_{i}' \in A_{i}$ and $b_{i}, b_{i}' \in B_{i}, i = 1, 2$ there exists a constant $C_{y_{2}y_{2}', a_{1}a_{1}', b_{1}b_{1}'}^{y_{1}y_{1}', a_{2}a_{2}', b_{2}b_{2}'} \in \bb{C}$ such that 
\begin{gather*}
	\sum\limits_{x_{1}}C_{x_{2}x_{2}', y_{2}y_{2}', a_{1}a_{1}', b_{1}b_{1}'}^{x_{1}x_{1}, y_{1}y_{1}', a_{2}a_{2}', b_{2}b_{2}'} = \delta_{x_{2}, x_{2}'}C_{y_{2}y_{2}', a_{1}a_{1}', b_{1}b_{1}'}^{y_{1}y_{1}', a_{2}a_{2}', b_{2}b_{2}'},
\end{gather*}
\noindent for all $x_{2}, x_{2}' \in X_{2}$. This equality implies
\begin{gather*}
	L_{\rho}(C) \in \cl{L}_{X_{2}, X_{1}} \text{ for all } \rho \in M_{Y_{2}Y_{1}A_{1}A_{2}B_{1}B_{2}}.
\end{gather*}
Similarly, using conditions (\ref{pt_two})-(\ref{pt_four}) we can make a similar argument showing
\begin{gather*}
	L_{\rho}(C) \in \cl{L}_{Y_{2}, Y_{1}} \text{ for all } \rho \in M_{X_{2}X_{1}A_{1}A_{2}B_{1}B_{2}}, 
\\	L_{\rho}(C) \in \cl{L}_{A_{1}, A_{2}} \text{ for all } \rho \in M_{X_{2}X_{1}Y_{2}Y_{1}B_{1}B_{2}},
\end{gather*}
\noindent and
\begin{gather*}
	L_{\rho}(C) \in \cl{L}_{B_{1}, B_{2}} \text{ for all } \rho \in M_{X_{2}X_{1}Y_{2}Y_{1}A_{1}A_{2}}.
\end{gather*}
\noindent Thus, 
\begin{gather*}
	C \in (\cl{L}_{X_{2}, X_{1}}\otimes \cl{L}_{Y_{2}, Y_{1}}\otimes \cl{L}_{A_{1}, A_{2}}\otimes \cl{L}_{B_{1}, B_{2}}) \cap M_{X_{2}Y_{2}A_{1}B_{1}, X_{1}Y_{1}A_{2}B_{2}}^{+};
\end{gather*}
\noindent by the injectivity of the minimal operator system tensor product, $C \in (\cl{L}_{X_{2}, X_{1}}\otimes_{\rm min} \cl{L}_{Y_{2}, Y_{1}}\otimes_{\rm min} \cl{L}_{A_{1}, A_{2}}\otimes_{\rm min} \cl{L}_{B_{1}, B_{2}})^{+}$. By \cite[Proposition 1.9]{fp}, and \cite[Proposition 5.5]{tt}, 
\begin{gather*}
	(\cl{T}_{X_{2}, X_{1}}\otimes_{\rm max} \cl{T}_{Y_{2}, Y_{1}}\otimes_{\rm max} \cl{T}_{A_{1}, A_{2}}\otimes_{\rm max} \cl{T}_{B_{1}, B_{2}})^{\rm d} \cong_{\rm c.o.i.} 
\\	\;\;\;\;\;\;\;\; \cl{L}_{X_{2}, X_{1}}\otimes_{\rm min} \cl{L}_{Y_{2}, Y_{1}}\otimes_{\rm min} \cl{L}_{A_{1}, A_{2}}\otimes_{\rm min} \cl{L}_{B_{1}, B_{2}}.
\end{gather*}
Arguing now as in \cite[Theorem 6.2]{tt}, this establishes the claim.

(ii) First suppose that
\begin{gather*}
	s: \cl{T}_{X_{2}, X_{1}}\otimes_{\rm c} \cl{T}_{Y_{2}, Y_{1}}\otimes_{\rm c} \cl{T}_{A_{1}, A_{2}}\otimes_{\rm c} \cl{T}_{B_{1}, B_{2}}\rightarrow \bb{C}
\end{gather*}
\noindent is a state. By Theorem \ref{op_sys_incl_to_env}, we may consider the state $s$ as the restriction of a state
\begin{gather*}
	\tilde{s}: \cl{C}_{X_{2}, X_{1}}\otimes_{\rm max} \cl{C}_{Y_{2}, Y_{1}}\otimes_{\rm max} \cl{C}_{A_{1}, A_{2}}\otimes_{\rm max} \cl{C}_{B_{1}, B_{2}}\rightarrow \bb{C}.
\end{gather*}
\noindent Applying the GNS construction to $\tilde{s}$ and using \cite[Theorem 5.2]{tt}, we obtain a Hilbert space $H$, a unit vector $\xi \in H$ and mutually commuting stochastic operator matrices
\begin{gather*}
	(E_{x_{2}, x_{2}', x_{1}, x_{1}'})_{x_{2}, x_{2}'}^{x_{1}, x_{1}'} \in M_{X_{2}X_{1}}\otimes \cl{B}(H), \;\;\;\; (E_{y_{2}, y_{2}', y_{1}, y_{1}'})_{y_{2}, y_{2}'}^{y_{1}, y_{1}'} \in M_{Y_{2}Y_{1}}\otimes \cl{B}(H), 
\\	(F_{a_{1}, a_{1}', a_{2}, a_{2}'})_{a_{1}, a_{1}'}^{a_{2}, a_{2}'} \in M_{A_{1}A_{2}}\otimes \cl{B}(H), \;\;\;\; (F_{b_{1}, b_{1}', b_{2}, b_{2}'})_{b_{1}, b_{1}'}^{b_{2}, b_{2}'} \in M_{B_{1}B_{2}}\otimes \cl{B}(H)
\end{gather*}
\noindent satisfying
\begin{gather*}
	s(e_{x_{2}, x_{2}', x_{1}, x_{1}}\otimes e_{y_{2}, y_{2}', y_{1}, y_{1}'}\otimes e_{a_{1}, a_{1}', a_{2}, a_{2}'}\otimes e_{b_{1}, b_{1}', b_{2}, b_{2}'}) = 
\\ \;\;\;\;\;\;\;\; \langle E_{x_{2}, x_{2}', x_{1}, x_{1}'}E_{y_{2}, y_{2}', y_{1}, y_{1}'}F_{a_{1}, a_{1}', a_{2}, a_{2}'}F_{b_{1}, b_{1}', b_{2}, b_{2}'}, \xi\xi^{*}\rangle
\end{gather*}
\noindent for all $x_{i}, x_{i}', y_{i}, y_{i}', a_{i}, a_{i}', b_{i}, b_{i}', i = 1, 2$. Setting $E_{x_{2}x_{2}', y_{2}y_{2}'}^{x_{1}x_{1}', y_{1}y_{1}'} = E_{x_{2}, x_{2}', x_{1}, x_{1}'}E_{y_{2}, y_{2}', y_{1}, y_{1}'}$ and $F_{a_{1}a_{1}', b_{1}b_{1}'}^{a_{2}a_{2}', b_{2}b_{2}'} = F_{a_{1}, a_{1}', a_{2}, a_{2}'}F_{b_{1}, b_{1}', b_{2}, b_{2}'}$, we have that the stochastic operator matrices 
\begin{gather*}
	E = \big(E_{x_{2}x_{2}', y_{2}y_{2}'}^{x_{1}x_{1}', y_{1}y_{1}'}\big)_{x_{2}x_{2}', y_{2}y_{2}'}^{x_{1}x_{1}', y_{1}y_{1}'}, \;\;\;\; F = \big(F_{a_{1}a_{1}', b_{1}b_{1}'}^{a_{2}a_{2}', b_{2}b_{2}'}\big)_{a_{1}a_{1}', b_{1}b_{1}'}^{a_{2}a_{2}', b_{2}b_{2}'}
\end{gather*}
\noindent are dilatable and have mutually commuting entries. Thus, $\Gamma_{s} = \Gamma_{E, F, \xi}$ with $\Gamma_{s} \in \cl{Q}_{\rm sqc}$. 

Conversely, suppose that $\Gamma \in \cl{Q}_{\rm sqc}$; by definition, there exists a Hilbert space $K$, a unit vector $\eta \in K$ and mutually commuting strongly stochastic operator matrices
\begin{gather*}
	E_{X} = (E_{x_{2}x_{2}', x_{1}x_{1}'})_{x_{2}x_{2}', x_{1}x_{1}'}, \;\;\;\; E_{Y} = (E_{y_{2}y_{2}', y_{1}y_{1}'})_{y_{2}y_{2}', y_{1}y_{1}'}, 
\\	F_{A} = (F_{a_{1}a_{1}', a_{2}a_{2}'})_{a_{1}a_{1}', a_{2}a_{2}'}, \;\;\;\; F_{B} = (F_{b_{1}b_{1}', b_{2}b_{2}'})_{b_{1}b_{1}', b_{2}b_{2}'}
\end{gather*}
\noindent acting on $K$ so that $\Gamma = \Gamma_{E, F, \eta}$, for $E = E_{X}\cdot E_{Y}$ and $F = F_{A} \cdot F_{B}$. Let $\pi_{X}, \pi_{Y}, \pi_{A}$ and $\pi_{B}$ be the (unital) $*$-representations of $\cl{C}_{X_{2}, X_{1}}, \cl{C}_{Y_{2}, Y_{1}}, \cl{C}_{A_{1}, A_{2}}$ and $\cl{C}_{B_{1}, B_{2}}$ on $\cl{B}(K)$ arising from $E_{X}, E_{Y}, F_{A}$, and $F_{B}$ respectively (see \cite[Theorem 5.2]{tt}). Then $\pi := \pi_{X}\otimes \pi_{Y}\otimes \pi_{A}\otimes \pi_{B}$ is a unital $*$-representation of $\cl{C}_{X_{2}, X_{1}}\otimes_{\rm max} \cl{C}_{Y_{2}, Y_{1}}\otimes_{\rm max} \cl{C}_{A_{1}, A_{2}}\otimes_{\rm max} \cl{C}_{B_{1}, B_{2}}$ on $\cl{B}(K)$. Using Theorem \ref{op_sys_incl_to_env} once more, if we let $s$ be the restriction to $\cl{T}_{X_{2}, X_{1}}\otimes_{\rm c}\cl{T}_{Y_{2}, Y_{1}}\otimes_{\rm c} \cl{T}_{A_{1}, A_{2}}\otimes_{\rm c} \cl{T}_{B_{1}, B_{2}}$ of the state
\begin{gather*}
	\tilde{s}: \cl{C}_{X_{2}, X_{1}}\otimes_{\rm max} \cl{C}_{Y_{2}, Y_{1}}\otimes_{\rm max} \cl{C}_{A_{1}, A_{2}}\otimes_{\rm max} \cl{C}_{B_{1}, B_{2}} \rightarrow \bb{C}, 
\\ 	w \mapsto \langle \pi(w)\eta, \eta\rangle,
\end{gather*}
\noindent it is clear that $\Gamma = \Gamma_{s}$. 

(iii) First, let $\Gamma$ be a quantum SQNS correlation. By definition, there exists stochastic operator matrices
\begin{gather*}
	M_{X} = (M_{x_{2}x_{2}', x_{1}x_{1}'})_{x_{2}x_{2}', x_{1}x_{1}'}, \;\;\;\; M_{A} = (M_{a_{1}a_{1}', a_{2}a_{2}'})_{a_{1}a_{1}', a_{2}a_{2}'}, 
\\ 	N_{Y} = (N_{y_{2}y_{2}', y_{1}y_{1}'})_{y_{2}y_{2}', y_{1}y_{1}'}, \;\;\;\; N_{B} = (N_{b_{1}b_{1}', b_{2}b_{2}'})_{b_{1}b_{1}', b_{2}b_{2}'}
\end{gather*}
\noindent acting on finite dimensional Hilbert spaces $H_{X}, H_{A}, H_{Y}$ and $H_{B}$, respectively, along with unit vector $\eta \in H_{X}\otimes H_{A}\otimes H_{Y}\otimes H_{B}$ such that
\begin{gather*}
	\langle \Gamma(\epsilon_{x_{2}x_{2}'}\otimes \epsilon_{y_{2}y_{2}'}\otimes \epsilon_{a_{1}a_{1}'}\otimes \epsilon_{b_{1}b_{1}'}), \epsilon_{x_{1}x_{1}'}\otimes \epsilon_{y_{1}y_{1}'}\otimes \epsilon_{a_{2}a_{2}'}\otimes \epsilon_{b_{2}b_{2}'}\rangle =
\\ \bigg\langle \big(M_{x_{2}x_{2}', x_{1}x_{1}'}\otimes M_{a_{1}a_{1}', a_{2}a_{2}'}\otimes N_{y_{2}y_{2}', y_{1}y_{1}'}\otimes N_{b_{1}b_{1}', b_{2}b_{2}'}\big), \eta\eta^{*}\bigg\rangle
\end{gather*}
\noindent for $x_{i}, x_{i}' \in X_{i}, y_{i}, y_{i}' \in Y_{i}, a_{i}, a_{i}' \in A_{i}, b_{i}, b_{i}' \in B_{i}, i = 1, 2$. Let $\pi_{X}: \cl{C}_{X_{2}, X_{1}}\rightarrow \cl{B}(H_{X}), \pi_{Y}: \cl{C}_{Y_{2}, Y_{1}}\rightarrow \cl{B}(H_{Y}), \pi_{A}: \cl{C}_{A_{1}, A_{2}}\rightarrow \cl{B}(H_{A})$ and $\pi_{B}: \cl{C}_{B_{1}, B_{2}}\rightarrow \cl{B}(H_{B})$ be the unital $*$-representations arising from $M_{X}, M_{A}, N_{Y}$, and $N_{B}$ respectively. Let $\pi := \pi_{X}\otimes \pi_{Y}\otimes \pi_{A}\otimes \pi_{B}$, $\tilde{\eta} \in H_{X}\otimes H_{Y}\otimes H_{A}\otimes H_{B}$ be the unit vector obtained from applying the canonical shuffle to $\eta$, $\tilde{s}$ be the state given by
\begin{gather*}
	\tilde{s}: \cl{C}_{X_{2}, X_{1}}\otimes_{\rm min} \cl{C}_{Y_{2}, Y_{1}}\otimes_{\rm min} \cl{C}_{A_{1}, A_{2}}\otimes_{\rm min} \cl{C}_{B_{1}, B_{2}}\rightarrow \bb{C}, 
\\	w \mapsto \langle \pi(w)\tilde{\eta}, \tilde{\eta}\rangle,
\end{gather*}
\noindent and $s$ be the restriction of $\tilde{s}$ to $\cl{T}_{X_{2}, X_{1}}\otimes_{\rm min} \cl{T}_{Y_{2}, Y_{1}}\otimes_{\rm min} \cl{T}_{A_{1}, A_{2}}\otimes_{\rm min}\cl{T}_{B_{1}, B_{2}}$; it is clear that $\Gamma = \Gamma_{s}$ in this case. 

If $\Gamma \in \cl{Q}_{\rm sqa}$, let $(\Gamma_{n})_{n \in \bb{N}}$ be a sequence of quantum SQNS correlations with $\Gamma_{n} \rightarrow \Gamma$ as $n\rightarrow \infty$. By the previous paragraph, for each $n \in \bb{N}$ we may pick a state $s_{n}: \cl{T}_{X_{2}, X_{1}}\otimes_{\rm min} \cl{T}_{Y_{2}, Y_{1}}\otimes_{\rm min} \cl{T}_{A_{1}, A_{2}}\otimes_{\rm min}\cl{T}_{B_{1}, B_{2}}\rightarrow \bb{C}$ such that $\Gamma_{n} = \Gamma_{s_{n}}$. If $s$ is a cluster point (in the ${\rm weak}^{*}$ topology) of the sequence of states $(s_{n})_{n \in \bb{N}}$, we have $\Gamma = \Gamma_{s}$. 

Now let $s: \cl{T}_{X_{2}, X_{1}}\otimes_{\rm min} \cl{T}_{Y_{2}, Y_{1}}\otimes_{\rm min} \cl{T}_{A_{1}, A_{2}}\otimes_{\rm min} \cl{T}_{B_{1}, B_{2}}\rightarrow \bb{C}$ be a state, and using the fact that the minimal operator system tensor product is injective (see \cite{gt_one}) let $\tilde{s}: \cl{C}_{X_{2}, X_{1}}\otimes_{\rm min} \cl{C}_{Y_{2}, Y_{1}}\otimes_{\rm min} \cl{C}_{A_{1}, A_{2}}\otimes_{\rm min} \cl{C}_{B_{1}, B_{2}}\rightarrow \bb{C}$ be an extension of $s$. By \cite[Corollary 4.3.10]{kr},  $\tilde{s}$ can be approximated in the ${\rm weak}^{*}$ topology by elements of the convex hull of vector states on $\pi_{X}(\cl{C}_{X_{2}, X_{1}})\otimes_{\rm min} \pi_{Y}(\cl{C}_{Y_{2}, Y_{1}})\otimes_{\rm min} \pi_{A}(\cl{C}_{A_{1}, A_{2}})\otimes_{\rm min} \pi_{B}(\cl{C}_{B_{1}, B_{2}})$ (where $\pi_{X}, \pi_{Y}, \pi_{A},$ and $\pi_{B}$ are unital $*$-representations of $\cl{C}_{X_{2}, X_{1}}, \cl{C}_{Y_{2}, Y_{1}}, \cl{C}_{A_{1}, A_{2}}$ and $\cl{C}_{B_{1}, B_{2}}$, respectively). Using an argument similar to the proof of \cite[Theorem 5.6]{bhtt} or \cite[Theorem 6.5]{tt}, we can show that $\Gamma$ is a limit of quantum SQNS correlations. 

(iv) This proof is along the lines of the proof for \cite[Theorem 6.7]{tt}; we include the details for the convenience of the reader. We first let $s$ be a state on
\begin{gather*}
	{\rm OMIN}(\cl{T}_{X_{2}, X_{1}})\otimes_{\rm min} {\rm OMIN}(\cl{T}_{Y_{2}, Y_{1}})\otimes_{\rm min} {\rm OMIN}(\cl{T}_{A_{1}, A_{2}})\otimes_{\rm min} {\rm OMIN}(\cl{T}_{B_{1}, B_{2}}).
\end{gather*}
\noindent By \cite[Theorem 9.9]{kavruk_thesis} and \cite[Proposition 1.9]{fp}, we may consider $s$ as an element of
\begin{gather*}
	\hspace{-0.65cm}\bigg({\rm OMAX}(\cl{T}_{X_{2}, X_{1}})\otimes_{\rm max} {\rm OMAX}(\cl{T}_{Y_{2}, Y_{1}})\otimes_{\rm max} {\rm OMAX}(\cl{T}_{A_{1}, A_{2}})\otimes_{\rm max} {\rm OMAX}(\cl{T}_{B_{1}, B_{2}})\bigg)^{+}.
\end{gather*}
\noindent By Lemma \ref{omax_state_extension}, there exist states $\phi_{X}^{(j)} \in \big(\cl{T}_{X_{2}, X_{1}}\big)^{+}, \phi_{Y}^{(j)} \in \big(\cl{T}_{Y_{2}, Y_{1}}\big)^{+}, \phi_{A}^{(j)} \in \big(\cl{T}_{A_{1}, A_{2}}\big)^{+}$ and $\phi_{B}^{(j)} \in \big(\cl{T}_{B_{1}, B_{2}}\big)^{+}$, and non-negative scalars $\lambda_{j}, j= 1, \hdots, k$ such that $s = \sum\limits_{j=1}^{k}\lambda_{j}\phi_{X}^{(j)}\otimes \phi_{Y}^{(j)}\otimes\phi_{A}^{(j)}\otimes\phi_{B}^{(j)}$. Set
\begin{gather*}
	E_{X}^{(j)} = (\phi_{X}^{(j)}(e_{x_{2}, x_{2}', x_{1}, x_{1}'}))_{x_{2}x_{2}', x_{1}x_{1}'}, \;\;\;\; E_{Y}^{(j)} = (\phi_{Y}^{(j)}(e_{y_{2}, y_{2}', y_{1}, y_{1}'}))_{y_{2}y_{2}', y_{1}y_{1}'}, 
\\	E_{A}^{(j)} = (\phi_{A}^{(j)}(e_{a_{1}, a_{1}', a_{2}, a_{2}'}))_{a_{1}a_{1}', a_{2}a_{2}'}, \;\;\;\; E_{B}^{(j)} = (\phi_{B}^{(j)}(e_{b_{1}, b_{1}', b_{2}, b_{2}'}))_{b_{1}b_{1}', b_{2}b_{2}'},
\end{gather*}
\noindent and let $\Phi_{X}^{(j)}: M_{X_{2}}\rightarrow M_{X_{1}}, \Phi_{Y}^{(j)}: M_{Y_{2}}\rightarrow M_{Y_{1}}, \Phi_{A}^{(j)}: M_{A_{1}}\rightarrow M_{A_{2}}$, and $\Phi_{B}^{(j)}: M_{B_{1}}\rightarrow M_{B_{2}}$ be the quantum channels with Choi matrices $E_{X}^{(j)}, E_{Y}^{(j)}, E_{A}^{(j)}$, and $E_{B}^{(j)}$ (respectively) for $j = 1, \hdots, k$. Clearly,
\begin{gather}\label{convex_sum_loc_channels}
	\Gamma_{s} = \sum\limits_{j=1}^{k}\lambda_{j}\Phi_{X}^{(j)}\otimes \Phi_{Y}^{(j)}\otimes \Phi_{A}^{(j)}\otimes \Phi_{B}^{(j)}.
\end{gather}
\noindent By Remark \ref{sloc_convex_remark}, this shows $\Gamma_{s} \in \cl{Q}_{\rm sloc}$.

Now suppose $\Gamma$ is of the form (\ref{convex_sum_loc_channels}), with $s$ a functional on $\cl{T}_{X_{2}, X_{1}}\otimes \cl{T}_{Y_{2}, Y_{1}}\otimes \cl{T}_{A_{1}, A_{2}}\otimes \cl{T}_{B_{1}, B_{2}}$ such that $\Gamma = \Gamma_{s}$. Let $E_{X}^{(j)} \in (M_{X_{2}}\otimes M_{X_{1}})^{+}$ (resp. $E_{Y}^{(j)} \in (M_{Y_{2}}\otimes M_{Y_{1}})^{+}, E_{A}^{(j)} \in (M_{A_{1}}\otimes M_{A_{2}})^{+}$, and $E_{B}^{(j)} \in (M_{B_{1}}\otimes M_{B_{2}})^{+}$) be the Choi matrix of $\Phi_{X}^{(j)}$ (resp. $\Phi_{Y}^{(j)}, \Phi_{A}^{(j)}$ and $\Phi_{B}^{(j)}$); then $E_{X}^{(j)}, E_{Y}^{(j)}, E_{A}^{(j)}$, and $E_{B}^{(j)}$ are stochastic operator matrices acting on $\bb{C}$. By \cite[Theorem 5.2]{tt}, there exist positive functionals $\phi_{X}^{(j)}: \cl{T}_{X_{2}, X_{1}}\rightarrow \bb{C}, \phi_{Y}^{(j)}: \cl{T}_{Y_{2}, Y_{1}}\rightarrow \bb{C}, \phi_{A}^{(j)}: \cl{T}_{A_{1}, A_{2}}\rightarrow \bb{C}$, and $\phi_{B}^{(j)}: \cl{T}_{B_{1}, B_{2}}\rightarrow \bb{C}$ such that
\begin{gather*}
	E_{X}^{(j)} = (\phi_{X}^{(j)}(e_{x_{2}, x_{2}', x_{1}, x_{1}'}))_{x_{2}x_{2}', x_{1}x_{1}'}, \;\;\;\; E_{Y}^{(j)} = (\phi_{Y}^{(j)}(e_{y_{2}, y_{2}', y_{1}, y_{1}'}))_{y_{2}y_{2}', y_{1}y_{1}'}, 
\\	E_{A}^{(j)} = (\phi_{A}^{(j)}(e_{a_{1}, a_{1}', a_{2}, a_{2}'}))_{a_{1}a_{1}', a_{2}a_{2}'}, \;\;\;\; E_{B}^{(j)} = (\phi_{B}^{(j)}(e_{b_{1}, b_{1}', b_{2}, b_{2}'}))_{b_{1}b_{1}', b_{2}b_{2}'},
\end{gather*}
\noindent for $j = 1, \hdots, k$. It is thus straightforward to see that $s$ is the functional corresponding to 
\begin{gather*}
	\sum\limits_{j=1}^{k}\lambda_{j}\phi_{X}^{(j)}\otimes \phi_{Y}^{(j)}\otimes \phi_{A}^{(j)}\otimes \phi_{B}^{(j)},
\end{gather*}
\noindent and thus, by Lemma \ref{omax_state_extension} is a state on 
\begin{gather*}
	{\rm OMIN}(\cl{T}_{X_{2}, X_{1}})\otimes_{\rm min} {\rm OMIN}(\cl{T}_{Y_{2}, Y_{1}})\otimes_{\rm min} {\rm OMIN}(\cl{T}_{A_{1}, A_{2}})\otimes_{\rm min} {\rm OMIN}(\cl{T}_{B_{1}, B_{2}}),
\end{gather*}
\noindent as claimed.
\end{proof}
\begin{remark}
\rm We note that, using an almost identical argument as in (iv), we have an affine isomorphism between the state space of ${\rm OMIN}(\cl{S}_{X_{2}, X_{1}})\otimes_{\rm min} {\rm OMIN}(\cl{S}_{Y_{2}, Y_{1}})\otimes_{\rm min} {\rm OMIN}(\cl{S}_{A_{1}, A_{2}})\otimes_{\rm min} {\rm OMIN}(\cl{S}_{B_{1}, B_{2}})$ and $\cl{C}_{\rm sloc}$; the classical analogue for cases (i)-(iii) are addressed in \cite[Theorem 7.11]{gt_one}.
\end{remark}
\begin{corollary}\label{qsqc_closed}
The set $\cl{Q}_{\rm sqc}$ is closed and convex.
\end{corollary}
\begin{proof}
By Theorem \ref{state_correspondence}, it is straightforward to show that the affine mapping $s \mapsto \Gamma_{s}$ is also homeomorphism when the state space of $\cl{T}_{X_{2}X_{1}}\otimes_{\rm c} \cl{T}_{Y_{2}Y_{1}}\otimes_{\rm c} \cl{T}_{A_{1}A_{2}}\otimes_{\rm c}\cl{T}_{B_{1}B_{2}}$ is equipped with the ${\rm weak}^{*}$-topology. As the state space will be ${\rm weak}^{*}$-compact, the range of this homeomorphism must be (convex and) closed. 
\end{proof}
\begin{remark}
\rm By \cite[Theorem 7.11]{gt_one}, we may use an identical argument as in Corollary \ref{qsqc_closed} to show that $\cl{C}_{\rm sqc}$ is closed and convex. 
\end{remark} 

 \begin{remark}\label{reduction_remark}
\rm Suppose that for finite sets $X_{i}, Y_{i}, A_{i}, B_{i}, i = 1, 2$ we make the additional restriction that $Y_{i} = B_{i} = [1], i = 1, 2$. One may easily verify that for ${\rm t} \in \{\rm loc, q, qa, qc, ns\}$ we have the reduction(s)
\begin{gather*}
	\cl{Q}_{\rm st}(X_{2}[1], A_{1}[1], X_{1}[1], A_{2}[1]) = \cl{Q}_{\rm t}(X_{2}, A_{1}, X_{1}, A_{2}),
\\	\cl{C}_{\rm st}(X_{2}[1], A_{1}[1], X_{1}[1], A_{2}[1]) = \cl{C}_{\rm t}(X_{2}, A_{1}, X_{1}, A_{2}).
\end{gather*}
\end{remark}
\begin{theorem}
For all finite sets $X_{i}, Y_{i}, A_{i}, B_{i}, i = 1, 2$ of sufficiently large cardinality, the following hold true:
\begin{itemize}
	\item[(i)] $\cl{Q}_{\rm sqa}(X_{2}Y_{2}, A_{1}B_{1}, X_{1}Y_{1}, A_{2}B_{2}) \neq \cl{Q}_{\rm sqc}(X_{2}Y_{2}, A_{1}B_{1}, X_{1}Y_{1}, A_{2}B_{2})$.
	\item[(ii)] $\cl{T}_{X_{2}, X_{1}}\otimes_{\rm c} \cl{T}_{Y_{2}, Y_{1}}\otimes_{\rm c} \cl{T}_{A_{1}, A_{2}}\otimes_{\rm c} \cl{T}_{B_{1}, B_{2}} \neq \cl{T}_{X_{2}, X_{1}}\otimes_{\rm min}\cl{T}_{Y_{2}, Y_{1}}\otimes_{\rm min} \cl{T}_{A_{1}, A_{2}}\otimes_{\rm min}\cl{T}_{B_{1}, B_{2}}$. 
	\item[(iii)] $\cl{C}_{\rm sqa}(X_{2}Y_{2}, A_{1}B_{1}, X_{1}Y_{1}, A_{2}B_{2}) \neq \cl{C}_{\rm sqc}(X_{2}Y_{2}, A_{1}B_{1}, X_{1}Y_{1}, A_{2}B_{2})$.
	\item[(iv)] $\cl{S}_{X_{2}, X_{1}}\otimes_{\rm c} \cl{S}_{Y_{2}, Y_{1}}\otimes_{\rm c} \cl{S}_{A_{1}, A_{2}}\otimes_{\rm c} \cl{S}_{B_{1}, B_{2}} \neq \cl{S}_{X_{2}, X_{1}}\otimes_{\rm min}\cl{S}_{Y_{2}, Y_{1}}\otimes_{\rm min} \cl{S}_{A_{1}, A_{2}}\otimes_{\rm min}\cl{S}_{B_{1}, B_{2}}$
\end{itemize}
\end{theorem}
\begin{proof}
Statement (i) follows from Remark \ref{reduction_remark}, and an application of \cite[Theorem 8.3]{tt}. Statement (ii) follows from (i), and Theorem \ref{state_correspondence}. Statements (iii) and (iv) follow similarly, when paired with \cite[Remark 8.1]{tt}.
\end{proof}

As in the setup before Theorem \ref{state_correspondence}, for a linear functional
\begin{gather*}
	s: \cl{T}_{X}\otimes \cl{T}_{Y}\otimes \cl{T}_{A}\otimes \cl{T}_{B}\rightarrow \bb{C},
\end{gather*}
\noindent let $\Gamma_{s}: M_{XYAB}\rightarrow M_{XYAB}$ be the linear map whose Choi matrix coincides with (\ref{choi_matrix_formula}). 
\begin{theorem}
The map $s \mapsto \Gamma_{s}$ is an affine isomorphism from
\begin{itemize}
	\item[(i)] the state space of $\cl{T}_{X}\otimes_{\rm max}\cl{T}_{Y}\otimes_{\rm max}\cl{T}_{A}\otimes_{\rm max}\cl{T}_{B}$ onto $\cl{Q}_{\rm sns}^{\rm bi}$;
	\item[(ii)] the state space of $\cl{T}_{X}\otimes_{\rm c}\cl{T}_{Y}\otimes_{\rm c} \cl{T}_{A}\otimes_{\rm c}\cl{T}_{B}$ onto $\cl{Q}_{\rm sqc}^{\rm bi}$;
	\item[(iii)] the state space of $\cl{T}_{X}\otimes_{\rm min}\cl{T}_{Y}\otimes_{\rm min}\cl{T}_{A}\otimes_{\rm min}\cl{T}_{B}$ onto $\cl{Q}_{\rm sqa}^{\rm bi}$;
	\item[(iv)] the state space of ${\rm OMIN}(\cl{T}_{X})\otimes_{\rm min} {\rm OMIN}(\cl{T}_{Y})\otimes_{\rm min} {\rm OMIN}(\cl{T}_{A})\otimes_{\rm min} {\rm OMIN}(\cl{T}_{B})$ onto $\cl{Q}_{\rm sloc}^{\rm bi}$.
\end{itemize}
\end{theorem}
\begin{proof}
(i) Let $\Gamma \in \cl{Q}_{\rm sns}^{\rm bi}$. If
\begin{gather*}
	C = \big(C_{x_{2}x_{2}', y_{2}y_{2}', a_{1}a_{1}', b_{1}b_{1}'}^{x_{1}x_{1}', y_{1}y_{1}', a_{2}a_{2}', b_{2}b_{2}'}\big)
\end{gather*}
\noindent is the Choi matrix of $\Gamma$ and
\begin{gather*}
	\tilde{C} = \big(\tilde{C}_{x_{1}x_{1}', y_{1}y_{1}', a_{2}a_{2}', b_{2}b_{2}'}^{x_{2}x_{2}', y_{2}y_{2}', a_{1}a_{1}', b_{1}b_{1}'}\big)
\end{gather*}
\noindent is the Choi matrix of $\Gamma^{*}$ (where the indices above range over the corresponding sets $X, Y, A$, and $B$) then $C, \tilde{C} \in M_{XYABXYAB}^{+}$. As $\Gamma$ and $\Gamma^{*}$ are both SNS, for $y_{i}, y_{i}' \in Y, a_{i}, a_{i}' \in A, b_{i}, b_{i}' \in B, i = 1, 2$ there exist constants $C_{y_{2}y_{2}', a_{1}a_{1}', b_{1}b_{1}'}^{y_{1}y_{1}', a_{2}a_{2}', b_{2}b_{2}'}, \tilde{C}_{y_{1}y_{1}', a_{2}a_{2}', b_{2}b_{2}'}^{y_{2}y_{2}', a_{1}a_{1}', b_{1}b_{1}'} \in \bb{C}$ such that
\begin{gather*}
	\sum\limits_{x_{1}}C_{x_{2}x_{2}', y_{2}y_{2}', a_{1}a_{1}', b_{1}b_{1}'}^{x_{1}x_{1}', y_{1}y_{1}', a_{2}a_{2}', b_{2}b_{2}'} = \delta_{x_{2}, x_{2}'}C_{y_{2}y_{2}', a_{1}a_{1}', b_{1}b_{1}'}^{y_{1}y_{1}', a_{2}a_{2}', b_{2}b_{2}'}, 
\\ 	\sum\limits_{x_{2}}\tilde{C}_{x_{1}x_{1}', y_{1}y_{1}', a_{2}a_{2}', b_{2}b_{2}'}^{x_{2}x_{2}', y_{2}y_{2}', a_{1}a_{1}', b_{1}b_{1}'} = \delta_{x_{1}, x_{1}'}\tilde{C}_{y_{1}y_{1}', a_{2}a_{2}', b_{2}b_{2}'}^{y_{2}y_{2}', a_{1}a_{1}', b_{1}b_{1}'},
\end{gather*}
\noindent for all $x_{i}, x_{i}' \in X, i = 1, 2$. Note that
\begin{gather*}
	\tilde{C}_{x_{1}x_{1}', y_{1}y_{1}', a_{2}a_{2}', b_{2}b_{2}'}^{x_{2}x_{2}', y_{2}y_{2}', a_{1}a_{1}', b_{1}b_{1}'} = C_{x_{2}x_{2}', y_{2}y_{2}', a_{1}a_{1}', b_{1}b_{1}'}^{x_{1}x_{1}', y_{1}y_{1}', a_{2}a_{2}', b_{2}b_{2}'},
\end{gather*}
\noindent for $x_{i}, x_{i}' \in X, y_{i}, y_{i}' \in Y, a_{i}, a_{i}' \in A, b_{i}, b_{i}' \in B, i = 1, 2$. These together imply
\begin{gather*}
	L_{\rho}(C) \in \cl{L}_{X} \text{ for all } \rho \in M_{YAB}.
\end{gather*}
\noindent We then may use the other SNS conditions for $\Gamma$ and $\Gamma^{*}$ to show that
\begin{gather*}
	L_{\rho}(C) \in \cl{L}_{Y} \text{ for all } \rho \in M_{XAB},
\\	L_{\rho}(C) \in \cl{L}_{A} \text{ for all } \rho \in M_{XYB},
\end{gather*}
\noindent and 
\begin{gather*}
	L_{\rho}(C) \in \cl{L}_{B} \text{ for all } \rho \in M_{XYA}.
\end{gather*}
\noindent Together, these imply $C \in (\cl{L}_{X}\otimes \cl{L}_{Y}\otimes \cl{L}_{A}\otimes \cl{L}_{B})^{+}$. Use \cite[Proposition 3.6]{bhtt} and argue as in Theorem \ref{state_correspondence} (i) to finish.

(ii)-(iv) The proofs for (ii)-(iv) follow exactly the same as in the proofs of Theorem \ref{state_correspondence} (ii)-(iv), just by replacing the use of \cite[Theorem 5.2]{tt} with \cite[Theorem 3.4]{bhtt}, the use of Remark \ref{sloc_convex_remark} with Remark \ref{bisloc_convex_remark}, and by replacing $\cl{T}_{X_{2}, X_{1}}$ and $\cl{C}_{X_{2}, X_{1}}$ (resp. $\cl{T}_{Y_{2}, Y_{1}}$ and $\cl{C}_{Y_{2}, Y_{1}}$, $\cl{T}_{A_{1}, A_{2}}$ and $\cl{C}_{A_{1}, A_{2}}$, $\cl{T}_{B_{1}, B_{2}}$ and $\cl{C}_{B_{1}, B_{2}}$) with $\cl{T}_{X}$ and $\cl{C}_{X}$ (resp. $\cl{T}_{Y}$ and $\cl{C}_{Y}$, $\cl{T}_{A}$ and $\cl{C}_{A}$, $\cl{T}_{B}$ and $\cl{C}_{B}$). 
\end{proof}

We end this subsection with the following (obvious) result.

\begin{corollary}
The set $\cl{Q}_{\rm sqc}^{\rm bi}$ is closed and convex.
\end{corollary}

\subsection{Concurrent strategies}
The goal of this subsection is to apply the results connecting subclasses of strongly QNS correlations with state spaces on tensor products of canonical operator systems previously developed, to the study of a particular class of quantum input-output game which are called \textit{concurrent games}. A quantum non-local game $\varphi: \cl{P}_{XX}\rightarrow \cl{P}_{AA}$ is \textit{concurrent} if $\varphi(\tilde{J}_{X}) = \tilde{J}_{A}$. A QNS correlation $\Phi \in \cl{Q}_{\rm ns}$ is called concurrent if $\Phi(\tilde{J}_{X}) = \tilde{J}_{A}$; equivalently, if $\Phi$ is a perfect strategy for the (trivial) implication game $\varphi_{\tilde{J}_{X}\rightarrow \tilde{J}_{A}}$. 

The study of concurrent games and their perfect strategies has grown in recent years (see \cite{synch_bhtt, bhtt, tt}), and is of particular interest for its connections to the study of quantum automorphism groups and compact quantum groups, along with its ramifications for the study of synchronous games; indeed--- concurrent games were introduced as a ``quantization" of synchronous games. The reasoning behind this claim comes from \cite[Remark 2.1]{synch_bhtt}: for classical non-local game $\lambda: X\times X\times A\times A\rightarrow \{0, 1\}$, then $\lambda$ is synchronous if and only if $\varphi_{\lambda}(J_{X}^{\rm cl}) \leq J_{A}^{\rm cl}$. Concurrency allows us to see how a correlation affects the ``non-classical" portion of the (unnormalized) maximally entangled state- ideally, sending it to a state supported on $\tilde{J}_{A}$. 

Let $\tau: \cl{C}_{X, A}\rightarrow \bb{C}$ be a tracial state; the linear map $\Gamma_{\tau}: M_{XX}\rightarrow M_{AA}$ given by
\begin{gather*}
	\Gamma_{\tau}(\epsilon_{xx'}\otimes \epsilon_{yy'}) = \sum\limits_{a, a', b, b' \in A}\tau(e_{x, x', a, a'}e_{y', y, b', b})\epsilon_{aa'}\otimes \epsilon_{bb'},
\end{gather*}
\noindent is a QNS correlation. This particular type of a QNS correlation is called \textit{tracial}, and was first introduced in \cite{tt}. Subclasses of \textit{quantum tracial} and \textit{locally tracial} QNS correlations are defined by requiring that $\tau$ factors through a finite-dimensional and abelian $*$-representation, respectively. We focus on such QNS correlations for their connection to concurrent games: it was first established in \cite[Theorem 4.1]{synch_bhtt} that any perfect strategy $\Gamma$ for a concurrent game $\varphi$ must be tracial, in the sense defined above. In the rest of this section, we will develop a notion of when SQNS correlations are tracial, and describe their structure. 

Let $\Gamma: M_{X_{2}X_{2}\times A_{1}A_{1}}\rightarrow M_{X_{1}X_{1}\times A_{2}A_{2}}$ be an SQNS correlation, and let $\sigma \in M_{A_{1}A_{1}}$ be an arbitrary state. Define the mapping $\Gamma_{X_{2}X_{2}\rightarrow X_{1}X_{1}}^{\sigma}: M_{X_{2}X_{2}}\rightarrow M_{X_{1}X_{1}}$ via
\begin{gather*}
	\Gamma_{X_{2}X_{2}\rightarrow X_{1}X_{1}}^{\sigma}(\rho) := {\rm Tr}_{A_{2}A_{2}}\Gamma(\rho\otimes \sigma), \;\;\;\; \rho \in M_{X_{2}X_{2}}.
\end{gather*} 
\noindent One easily shows that $\Gamma_{X_{2}X_{2}\rightarrow X_{1}X_{1}}^{\sigma}$ is a well-defined QNS correlation; furthermore, $\Gamma_{X_{2}X_{2}\rightarrow X_{1}X_{1}}^{\sigma} = \Gamma_{X_{2}X_{2}\rightarrow X_{1}X_{1}}^{\sigma'}$ for any states $\sigma, \sigma' \in M_{A_{1}A_{1}}$. Thus, we let $\Gamma_{X_{2}X_{2}\rightarrow X_{1}X_{1}}$ denote the right marginal channel of $\Gamma$, for an arbitrary fixed state $\sigma \in M_{A_{1}A_{1}}$. An analogous argument shows that $\Gamma^{A_{1}A_{1}\rightarrow A_{2}A_{2}}$ is a well-defined left marginal QNS correlation for $\Gamma$, for an arbitrary fixed state $\rho \in M_{X_{2}X_{2}}$. 

\begin{remark}
\rm Note that, in the event $\Gamma \in \cl{C}_{\rm st}$ for ${\rm t} \in \{\rm loc, q, qa, qc, ns\}$, then the marginal channels previously discussed (when considering $\Gamma$ inside $\cl{Q}_{\rm st}$) coincide with the marginal channels discussed in \cite[Section 5]{gt_one}.
\end{remark}

\begin{definition}
An SQNS correlation $\Gamma$ over $(X_{2}X_{2}, A_{1}A_{1}, X_{1}X_{1}, A_{2}A_{2})$ is called {\rm jointly tracial} if $\Gamma_{X_{2}X_{2}\rightarrow X_{1}X_{1}}$ and $\Gamma^{A_{1}A_{1}\rightarrow A_{2}A_{2}}$ are tracial QNS correlations. 
\end{definition}

\begin{remark}\label{hom_from_max_tensor_prod}
\rm There exists a surjective $*$-homomorphism $\pi: \cl{C}_{X_{2}, A_{2}}\rightarrow \cl{C}_{X_{2}, X_{1}}\otimes_{\rm max} \cl{C}_{A_{1}, A_{2}}\otimes_{\rm max} \cl{C}_{X_{1}, A_{1}}$. Indeed: for $x_{2}, x_{2}' \in X_{2}, a_{2}, a_{2}' \in A_{2}$ let
\begin{gather*}
	\tilde{e}_{x_{2}x_{2}', a_{2}a_{2}'} := \sum\limits_{x_{1}, x_{1}'}\sum\limits_{a_{1}, a_{1}'}e_{x_{2}x_{2}', x_{1}x_{1}'}\otimes e_{a_{1}a_{1}', a_{2}a_{2}'}\otimes f_{x_{1}x_{1}', a_{1}a_{1}'}.
\end{gather*}
\noindent One easily checks that $\tilde{e}_{x_{2}x_{2}, a_{2}a_{2}} \geq 0$ and $\sum_{a_{2}}\tilde{e}_{x_{2}x_{2}', a_{2}a_{2}} = \delta_{x_{2}x_{2}'}1$ for all $x_{2}, x_{2}' \in X_{2}, a_{2} \in A_{2}$. Therefore, by universality there exists a surjective $*$-homomorphism $\pi: \cl{C}_{X_{2}, A_{2}} \rightarrow \cl{C}_{X_{2}, X_{1}}\otimes_{\rm max} \cl{C}_{A_{1}, A_{2}}\otimes_{\rm max} \cl{C}_{X_{1}, A_{1}}$ which sends $\pi(e_{x_{2}x_{2}', a_{2}a_{2}'}) = \tilde{e}_{x_{2}x_{2}', a_{2}a_{2}'}, x_{2}, x_{2}' \in X_{2}, a_{2}, a_{2}' \in A_{2}$. 
\end{remark}
For a linear functional $\mathsf{T}$ on $\cl{C}_{X_{2}, X_{1}}\otimes \cl{C}_{A_{1}, A_{2}}$, set
\begin{gather*}
	\Gamma_{\mathsf{T}}(x_{1}x_{1}', y_{1}y_{1}', a_{2}a_{2}', b_{2}b_{2}'|x_{2}x_{2}', y_{2}y_{2}', a_{1}a_{1}', b_{1}b_{1}') = \\\mathsf{T}(e_{x_{2}x_{2}', x_{1}x_{1}'}e_{y_{2}'y_{2}, y_{1}'y_{1}}\otimes e_{a_{1}a_{1}', a_{2}a_{2}'}e_{b_{1}'b_{1}, b_{2}'b_{2}}),
\end{gather*}
\noindent where $x_{i}, x_{i}', y_{i}, y_{i}' \in X_{i}, a_{i}, a_{i}', b_{i}, b_{i}' \in A_{i}, i = 1, 2$. 
\begin{theorem}\label{joint_tracial_strategies_thm}
The following hold:
\begin{itemize}
	\item[(i)] If $\Gamma$ is a jointly tracial and quantum commuting SQNS correlation, then there exists a trace $\mathsf{T}$ on $\cl{C}_{X_{2}, X_{1}}\otimes_{\rm max} \cl{C}_{A_{1}, A_{2}}$ such that $\Gamma = \Gamma_{\mathsf{T}}$. 
	\item[(ii)] If $\Gamma$ is a jointly tracial and approximately quantum SQNS correlation, then there exists an amenable trace $\mathsf{T}$ on $\cl{C}_{X_{2}, X_{1}}\otimes_{\rm min} \cl{C}_{A_{1}, A_{2}}$ such that $\Gamma = \Gamma_{\mathsf{T}}$.
	\item[(iii)] If $\Gamma$ is a jointly tracial and quantum SQNS correlation, then there exists a trace $\mathsf{T}$ factoring through a finite-dimensional $*$-representation of $\cl{C}_{X_{2}, X_{1}}\otimes_{\rm min} \cl{C}_{A_{1}, A_{2}}$ such that $\Gamma = \Gamma_{\mathsf{T}}$.
	\item[(iv)] If $\Gamma$ is a jointly tracial and local SQNS correlation, then there exists a trace $\mathsf{T}$ factoring through an abelian $*$-representation of $\cl{C}_{X_{2}, X_{1}}\otimes_{\rm max} \cl{C}_{A_{1}, A_{2}}$ such that $\Gamma = \Gamma_{\mathsf{T}}$.
\end{itemize}
\end{theorem}
\begin{proof}
(i) First, assume $\Gamma \in \cl{Q}_{\rm sqc}$ is jointly tracial. For notational simplification, write
\begin{gather*}
	\mathfrak{B} = \cl{C}_{X_{2}, X_{1}}\otimes_{\rm max} \cl{C}_{X_{2}, X_{1}}\otimes_{\rm max} \cl{C}_{A_{1}, A_{2}}\otimes_{\rm max} \cl{C}_{A_{1}, A_{2}}
\end{gather*}
\noindent and
\begin{gather*}
	\mathfrak{U} = \cl{C}_{X_{2}, X_{1}}\otimes_{\rm max} \cl{C}_{A_{1}, A_{2}}.
\end{gather*}
\noindent We note that, up to a flip of tensor legs, $\mathfrak{B} = \mathfrak{U}\otimes_{\rm max} \mathfrak{U}$; in the sequel, we will use this identification without explicitly mentioning it. 

By (the proof of) Theorem \ref{state_correspondence}, there exists a state $\tilde{s}: \mathfrak{B}\rightarrow \bb{C}$ such that $\Gamma = \Gamma_{\tilde{s}}$. For $v, w \in \mathfrak{U}\otimes_{\rm max} \mathfrak{U}$, write $v \sim w$ if $\tilde{s}(v-w) = 0$. Let $V = (v_{x_{1}, x_{2}})_{x_{1}, x_{2}}$ be the isometry such that $e_{x_{2}x_{2}', x_{1}x_{1}'} = v_{x_{1}, x_{2}}^{*}v_{x_{1}', x_{2}'}$. Then
\begin{gather*}
	VV^{*} = \bigg(\sum\limits_{x_{2}}v_{x_{1}, x_{2}}v_{x_{1}', x_{2}}^{*}\bigg)_{x_{1}, x_{1}'}
\end{gather*}
\noindent is a projection, with $\sum\limits_{x_{2}}v_{x_{1}, x_{2}}v_{x_{1}, x_{2}}^{*} \leq 1$ for all $x_{1} \in X_{1}$. Using this, we see that
\begin{gather*}
	\sum\limits_{x_{2} \in X_{2}}e_{x_{2}'x_{2}, x_{1}'x_{1}}e_{x_{2}x_{2}', x_{1}x_{1}'} = \sum\limits_{x_{2} \in X_{2}}v_{x_{1}', x_{2}'}^{*}v_{x_{1}, x_{2}}v_{x_{1}, x_{2}}^{*}v_{x_{1}', x_{2}'} \\
\;\;\;\;\;\;\;\;\;\;\;\;\;\;\;\;\;\;\;\;\;\;\;\;\;\;\;\;\;\;\;\;\; \leq v_{x_{1}', x_{2}'}^{*}v_{x_{1}', x_{2}'} = e_{x_{2}'x_{2}', x_{1}'x_{1}'}. 
\end{gather*}
\noindent From this, we compute
\begin{gather}\label{inequality_of_generators_one}
	\sum\limits_{x_{2}, x_{2}' \in X_{2}}\sum\limits_{x_{1}, x_{1}' \in X_{1}}e_{x_{2}'x_{2}, x_{1}'x_{1}}e_{x_{2}x_{2}', x_{1}x_{1}'} \leq |X_{2}||X_{1}|1.
\end{gather} 
\noindent Let $\tau_{X}$ be the tracial state on $\cl{C}_{X_{2}, X_{1}}$ corresponding to marginal channel $\Gamma_{X_{2}X_{2}\rightarrow X_{1}X_{1}}$. Without loss of generality, fix $a_{1}, b_{1} \in A_{1}$. We see
\begin{eqnarray*}
& &
	\tilde{s}(e_{x_{2}x_{2}', x_{1}x_{1}'}\otimes 1 \otimes 1 \otimes 1) \\
& = &
	\sum\limits_{y_{1} \in X_{1}}\sum\limits_{a_{2}, b_{2} \in A_{2}}\tilde{s}(e_{x_{2}x_{2}', x_{1}x_{1}'}\otimes e_{x_{2}x_{2}, y_{1}y_{1}}\otimes e_{a_{1}a_{1}, a_{2}a_{2}}\otimes e_{b_{1}b_{1}, b_{2}b_{2}}) \\
& = &
	\sum\limits_{a_{2}, b_{2} \in A_{2}}\sum\limits_{y_{1} \in X_{1}}\tau_{X}(e_{x_{2}x_{2}', x_{1}x_{1}'}e_{x_{2}x_{2}, y_{1}y_{1}}) \\
& = &
	\tau_{X}(e_{x_{2}x_{2}', x_{1}x_{1}'}).
\end{eqnarray*}
\noindent Similarly, 
\begin{eqnarray*}
& &
	\tilde{s}(1\otimes e_{x_{2}'x_{2}, x_{1}'x_{1}}\otimes 1\otimes 1) \\
& = &
	\sum\limits_{y_{1} \in X_{1}}\sum\limits_{a_{2}, b_{2} \in A_{2}}\tilde{s}(e_{x_{2}'x_{2}', y_{1}y_{1}}\otimes e_{x_{2}'x_{2}, x_{1}'x_{1}}\otimes e_{a_{1}a_{1}, a_{2}a_{2}}\otimes e_{b_{1}b_{1}, b_{2}b_{2}}) \\
& = &
	\sum\limits_{a_{2}, b_{2} \in A_{2}}\sum\limits_{y_{1} \in X_{1}}\tau_{X}(e_{x_{2}'x_{2}', y_{1}y_{1}}e_{x_{2}x_{2}', x_{1}x_{1}'}) \\
& = &
	\tau_{X}(e_{x_{2}x_{2}', x_{1}x_{1}'}).
\end{eqnarray*}
\noindent Therefore, 
\begin{gather*}
	e_{x_{2}x_{2}', x_{1}x_{1}'}\otimes 1 \otimes 1 \otimes 1 \sim 1\otimes e_{x_{2}'x_{2}, x_{1}'x_{1}}\otimes 1 \otimes 1, \;\;\;\; x_{i}, x_{i}' \in X_{i}, i = 1, 2.
\end{gather*}
\noindent Write 
\begin{gather*}
	h_{x_{2}x_{2}', x_{1}x_{1}'} = e_{x_{2}x_{2}', x_{1}x_{1}'}\otimes 1\otimes 1\otimes 1 - 1\otimes e_{x_{2}'x_{2}, x_{1}'x_{1}}\otimes 1\otimes 1,
\end{gather*}
\noindent for $x_{i}, x_{i}' \in X_{i}, i = 1, 2$; we note that
\begin{eqnarray*}
& &
	h_{x_{2}x_{2}', x_{1}x_{1}'}^{*}h_{x_{2}x_{2}', x_{1}x_{1}'} \\
& = &
	(e_{x_{2}'x_{2}, x_{1}'x_{1}}e_{x_{2}x_{2}', x_{1}x_{1}'}\otimes 1 \otimes 1 \otimes 1+1\otimes e_{x_{2}x_{2}', x_{1}x_{1}'}e_{x_{2}'x_{2}, x_{1}'x_{1}}\otimes 1 \otimes 1) \\
& - &
	 (e_{x_{2}'x_{2}, x_{1}'x_{1}}\otimes e_{x_{2}'x_{2}, x_{1}'x_{1}}\otimes 1 \otimes 1+e_{x_{2}x_{2}', x_{1}x_{1}'}\otimes e_{x_{2}x_{2}', x_{1}x_{1}'}\otimes 1\otimes 1).
\end{eqnarray*}
\noindent By (\ref{inequality_of_generators_one}), we have
\begin{eqnarray*}
& &
	\sum\limits_{x_{2}, x_{2}', x_{1}, x_{1}'}\tilde{s}(h_{x_{2}x_{2}', x_{1}x_{1}'}^{*}h_{x_{2}x_{2}', x_{1}x_{1}'}) \\
& = &
	\sum\limits_{x_{2}, x_{2}', x_{1}, x_{1}'}\tilde{s}(e_{x_{2}'x_{2}, x_{1}'x_{1}}e_{x_{2}x_{2}', x_{1}x_{1}'}\otimes 1 \otimes 1 \otimes 1+1\otimes e_{x_{2}x_{2}', x_{1}x_{1}'}e_{x_{2}'x_{2}, x_{1}'x_{1}}\otimes 1\otimes 1) \\
& - &
	\tilde{s}(e_{x_{2}'x_{2}, x_{1}'x_{1}}\otimes e_{x_{2}'x_{2}, x_{1}'x_{1}}\otimes 1 \otimes 1+e_{x_{2}x_{2}', x_{1}x_{1}'}\otimes e_{x_{2}x_{2}', x_{1}x_{1}'}\otimes 1\otimes 1) \\
& \leq &
	2|X_{2}||X_{1}|1-2|X_{2}||X_{1}|1 = 0.
\end{eqnarray*}
\noindent From this, we have
\begin{gather}\label{h_star_eqn}
	\tilde{s}(h_{x_{2}x_{2}', x_{1}x_{1}}^{*}h_{x_{2}x_{2}', x_{1}x_{1}'}) = 0, \;\;\;\; x_{i}, x_{i}' \in X_{i}, i = 1, 2.
\end{gather}
By applying Cauchy-Schwarz inequality in conjunction with (\ref{h_star_eqn}), we have $wh_{x_{2}x_{2}', x_{1}x_{1}'} \sim 0$ and $h_{x_{2}x_{2}', x_{1}x_{1}'}w\sim 0$ for all $w \in \mathfrak{B}$. In particular, for all $x_{i}, x_{i}' \in X_{i}, i = 1, 2$ we have
\begin{gather}\label{equiv_one}
	ze_{x_{2}x_{2}', x_{1}x_{1}'}\otimes 1 \otimes v \sim z\otimes e_{x_{2}'x_{2}, x_{1}'x_{1}}\otimes v \sim e_{x_{2}x_{2}', x_{1}x_{1}'}z\otimes 1 \otimes v,
\end{gather}
\noindent for $z \in \cl{C}_{X_{2}, X_{1}}, v \in \mathfrak{U}$. If instead we use $\tau_{A}$ for the tracial state corresponding to marginal channel $\Gamma^{A_{1}A_{1}\rightarrow A_{2}A_{2}}$, we can make an almost identical argument to show that
\begin{gather}\label{equiv_two}
	u \otimes ze_{a_{1}a_{1}', a_{2}a_{2}'}\otimes 1 \sim u \otimes z \otimes e_{a_{1}'a_{1}, a_{2}'a_{2}} \sim u\otimes e_{a_{1}a_{1}', a_{2}a_{2}'},
\end{gather}
\noindent for $a_{i}, a_{i}' \in A_{i}, z \in \cl{C}_{A_{1}, A_{2}}$ and $u \in \mathfrak{U}$, $i = 1, 2$. Equations (\ref{equiv_one}) and (\ref{equiv_two}) imply
\begin{gather}\label{equiv_three}
	ze_{x_{2}x_{2}', x_{1}x_{1}'}\otimes 1 \otimes z'e_{a_{1}a_{1}', a_{2}a_{2}'}\otimes 1 \sim z \otimes e_{x_{2}'x_{2}, x_{1}'x_{1}}\otimes z' \otimes e_{a_{1}'a_{1}, a_{2}'a_{2}} 
\\\sim e_{x_{2}x_{2}', x_{1}x_{1}'}z\otimes 1\otimes e_{a_{1}a_{1}', a_{2}a_{2}'}z'\otimes 1,
\end{gather}
\noindent for all $z \in \cl{C}_{X_{2}, X_{1}}, z' \in \cl{C}_{A_{1}, A_{2}}$, and $x_{i}, x_{i}' \in X_{i}, a_{i}, a_{i}' \in A_{i}, i = 1, 2$. An induction argument on the lengths of the words $w$ on $\{e_{x_{2}x_{2}', x_{1}x_{1}'}: \; x_{i}, x_{i}' \in X_{i}, i = 1, 2\}$ and $w'$ on $\{e_{a_{1}a_{1}', a_{2}a_{2}'}: \; a_{i}, a_{i}' \in A_{i}, i = 1, 2\}$ whose base step is provided by (\ref{equiv_three}) shows that
\begin{gather*}
	zw\otimes 1 \otimes z'w'\otimes 1 \sim wz\otimes 1 \otimes w'z'\otimes 1, \;\;\;\; z, w \in \cl{C}_{X_{2}, X_{1}}, z', w' \in \cl{C}_{A_{1}, A_{2}}
\end{gather*}
\noindent (see, for instance, the proofs of \cite[Theorem 3.2, Theorem 4.1]{synch_bhtt}). We may then conclude that the functional $\mathsf{T}$ on $\cl{C}_{X_{2}, X_{1}}\otimes_{\rm max} \cl{C}_{A_{1}, A_{2}}$, given by
\begin{gather*}
	\mathsf{T}(u\otimes v) = \tilde{s}(u\otimes 1 \otimes v\otimes 1), \;\;\;\; u \in \cl{C}_{X_{2}, X_{1}}, v \in \cl{C}_{A_{1}, A_{2}},
\end{gather*}
\noindent is a tracial state. Furthermore, (\ref{equiv_three}) implies $\Gamma = \Gamma_{\mathsf{T}}$. 

(ii) If $\Gamma$ is an approximately quantum jointly tracial correlation, again by the proof of Theorem \ref{state_correspondence} there exists a state $s :\cl{C}_{X_{2}, X_{1}}\otimes_{\rm min} \cl{C}_{X_{2}, X_{1}}\otimes_{\rm min} \cl{C}_{A_{1}, A_{2}}\otimes_{\rm min} \cl{C}_{A_{1}, A_{2}}\rightarrow \bb{C}$ such that $\Gamma = \Gamma_{s}$. As each approximately quantum SQNS correlation is quantum commuting, by (i) there exists some trace $\mathsf{T} : \cl{C}_{X_{2}, X_{1}}\otimes_{\rm max} \cl{C}_{A_{1}, A_{2}}\rightarrow \bb{C}$ such that $\Gamma = \Gamma_{\mathsf{T}}$ as well. By \cite[Lemma 9.2]{tt}, for any finite sets $X, Y$ there exists a $*$-isomorphism $\partial: \cl{C}_{X, Y}\rightarrow \cl{C}_{X, Y}^{\rm op}$ such that
\begin{gather*}
	\partial(e_{xx', yy'}) = e_{x'x, y'y}^{\rm op}, \;\;\;\; x, x' \in X, y, y' \in Y.
\end{gather*}
\noindent Let $\partial_{X}: \cl{C}_{X_{2}, X_{1}}\rightarrow \cl{C}_{X_{2}, X_{1}}^{\rm op}$ and $\partial_{A}: \cl{C}_{A_{1}, A_{2}}\rightarrow \cl{C}_{A_{1}, A_{2}}^{\rm op}$ be the $*$-isomorphisms corresponding to each ${\rm C}^{*}$-algebra. Let
\begin{gather*}
	\cl{F}: \cl{C}_{X_{2}, X_{1}}\otimes_{\rm min} \cl{C}_{X_{2}, X_{1}}\otimes_{\rm min} \cl{C}_{A_{1}, A_{2}}\otimes_{\rm min} \cl{C}_{A_{1}, A_{2}} \\\rightarrow \cl{C}_{X_{2}, X_{1}}\otimes_{\rm min} \cl{C}_{A_{1}, A_{2}}\otimes_{\rm min} \cl{C}_{X_{2}, X_{1}}\otimes_{\rm min} \cl{C}_{A_{1}, A_{2}}
\end{gather*}
\noindent be the flip operation, and 
\begin{gather*}
	q: \cl{C}_{X_{2}, X_{1}}\otimes_{\rm max} \cl{C}_{A_{1}, A_{2}}\rightarrow \cl{C}_{X_{2}, X_{1}}\otimes_{\rm min} \cl{C}_{A_{1}, A_{2}}
\end{gather*}
\noindent be the quotient map. Define linear functional
\begin{gather*}
	\mu: (\cl{C}_{X_{2}, X_{1}}\otimes_{\rm max} \cl{C}_{A_{1}, A_{2}})\otimes _{\rm min} (\cl{C}_{X_{2}, X_{1}}\otimes_{\rm max} \cl{C}_{A_{1}, A_{2}})^{\rm op} \rightarrow \bb{C}
\end{gather*}
\noindent via
\begin{gather*}
	\mu = s\circ \cl{F}\circ (q\otimes (q\circ (\partial_{X}^{-1}\otimes \partial_{A}^{-1}))).
\end{gather*}
\noindent One may easily check that $\mu(u \otimes v^{\rm op}) = \mathsf{T}(uv)$. Indeed: on the canonical generators we have
\begin{eqnarray*}
& &
	\mu(e_{x_{2}x_{2}', x_{1}x_{1}'}\otimes e_{a_{1}a_{1}', a_{2}a_{2}'}\otimes (e_{y_{2}'y_{2}, y_{1}'y_{1}}\otimes e_{b_{1}'b_{1}, b_{2}'b_{2}})^{\rm op}) \\
& = &
	s\circ \cl{F}(e_{x_{2}x_{2}', x_{1}x_{1}'}\otimes e_{a_{1}a_{1}', a_{2}a_{2}'}\otimes e_{y_{2}y_{2}', y_{1}y_{1}'}\otimes e_{b_{1}b_{1}', b_{2}b_{2}'}) \\
& = &
	s(e_{x_{2}x_{2}', x_{1}x_{1}'}\otimes e_{y_{2}y_{2}', y_{1}y_{1}'}\otimes e_{a_{1}a_{1}', a_{2}a_{2}'}\otimes e_{b_{1}b_{1}', b_{2}b_{2}'}) \\
& = &
	\mathsf{T}(e_{x_{2}x_{2}', x_{1}x_{1}'}e_{y_{2}'y_{2}, y_{1}'y_{1}}\otimes e_{a_{1}a_{1}', a_{2}a_{2}'}e_{b_{1}'b_{1}, b_{2}'b_{2}}).
\end{eqnarray*}
\noindent By \cite[Theorem 6.2.7]{bo}, trace $\mathsf{T}$ is amenable.  

(iii) Let $\Gamma$ be a perfect concurrent strategy in $\cl{Q}_{\rm sq}$. By the proof of Theorem \ref{state_correspondence} (iii), there exist finite-dimensional Hilbert spaces $H_{X}, H_{X'}, H_{A}, H_{A'}$, $*$-representations $\pi_{X}: \cl{C}_{X_{2}, X_{1}}\rightarrow \cl{B}(H_{X}), \pi_{X'}: \cl{C}_{X_{2}, X_{1}}\rightarrow \cl{B}(H_{X'}), \pi_{A}: \cl{C}_{A_{1}, A_{2}}\rightarrow \cl{B}(H_{A})$ and $\pi_{A'}: \cl{C}_{A_{1}, A_{2}}\rightarrow \cl{B}(H_{A'})$ along with a unit vector $\xi \in H_{X}\otimes H_{X'}\otimes H_{A}\otimes H_{A'}$ such that $\Gamma = \Gamma_{\tilde{s}}$, where 
\begin{gather*}
	\tilde{s}: \cl{C}_{X_{2}, X_{1}}\otimes_{\rm min} \cl{C}_{X_{2}, X_{1}}\otimes_{\rm min} \cl{C}_{A_{1}, A_{2}}\otimes_{\rm min} \cl{C}_{A_{1}, A_{2}}\rightarrow \bb{C}
\end{gather*}
\noindent is a state given by
\begin{gather}\label{q_concurrent_state_eqn}
	\tilde{s}(u) = \langle (\pi_{X}\otimes \pi_{X'}\otimes \pi_{A}\otimes \pi_{A'})(u))\xi, \xi\rangle,
\end{gather}
\noindent for $u \in \cl{C}_{X_{2}, X_{1}}\otimes_{\rm min}\cl{C}_{X_{2}, X_{1}}\otimes_{\rm min} \cl{C}_{A_{1}, A_{2}}\otimes_{\rm min} \cl{C}_{A_{1}, A_{2}}$. By the proof of (i), the state $\mathsf{T}: \cl{C}_{X_{2}, X_{1}}\otimes_{\rm min} \cl{C}_{A_{1}, A_{2}} \rightarrow \bb{C}$ constructed from $\tilde{s}$ is tracial, and in this case factors through a finite-dimensional $*$-representation. 

(iv) The proof of this statement is similar to the proof of Proposition \ref{qloc_prop}; we include the details for the benefit of the reader. Assume $\Gamma \in \cl{Q}_{\rm sloc}$ is a perfect concurrent strategy; by Remark \ref{sloc_convex_remark}, $\Gamma = \sum_{j=1}^{k}\lambda_{j}\Phi_{X}^{(j)}\otimes \Phi_{X'}^{(j)}\otimes\Phi_{A}^{(j)}\otimes \Phi_{A'}^{(j)}$ as a convex combination of quantum channels $\Phi_{X}^{(j)}, \Phi_{X'}^{(j)}: M_{X_{2}}\rightarrow M_{X_{1}}, \Phi_{A}^{(j)}, \Phi_{A'}^{(j)}: M_{A_{1}}\rightarrow M_{A_{2}}, j = 1, \hdots, k$. Let 
\begin{gather*}
	(\lambda_{x_{2}x_{2}', x_{1}x_{1}'}^{(j)})_{x_{1}, x_{1}'} = \Phi_{X}^{(j)}(\epsilon_{x_{2}x_{2}'}), \;\;\;\; (\lambda_{y_{2}y_{2}', y_{1}y_{1}'}^{(j)})_{y_{1}, y_{1}'} = \Phi_{X'}^{(j)}(\epsilon_{y_{2}y_{2}'}),  \\
	(\mu_{a_{1}a_{1}', a_{2}a_{2}'}^{(j)})_{a_{2}, a_{2}'} = \Phi_{A}^{(j)}(\epsilon_{a_{1}a_{1}'}), \;\;\;\; (\mu_{b_{1}b_{1}', b_{2}b_{2}'}^{(j)})_{b_{2}, b_{2}'} = \Phi_{A'}^{(j)}(\epsilon_{b_{1}b_{1}'}),
\end{gather*}
\noindent for $x_{2}, x_{2}', y_{2}, y_{2}' \in X_{2}, a_{1}, a_{1}', b_{1}, b_{1}' \in A_{1}$, and $\pi_{X}^{(j)}, \pi_{X'}^{(j)}: \cl{C}_{X_{2}, X_{1}}\rightarrow \bb{C}, \pi_{A}^{(j)}, \pi_{A'}^{(j)}: \cl{C}_{A_{1}, A_{2}}\rightarrow \bb{C}$ be the $*$-representations given by
\begin{gather*}
	\pi_{X}^{(j)}(e_{x_{2}x_{2}', x_{1}x_{1}'}) = \lambda_{x_{2}x_{2}', x_{1}x_{1}'}^{(j)}, \;\;\;\; \pi_{X'}^{(j)}(e_{y_{2}y_{2}', y_{1}y_{1}'}) = \lambda_{y_{2}y_{2}', y_{1}y_{1}'}^{(j)}, 
\\ 	\pi_{A}^{(j)}(e_{a_{1}a_{1}', a_{2}a_{2}'}) = \mu_{a_{1}a_{1}', a_{2}a_{2}'}^{(j)}, \;\;\;\; \pi_{A'}^{(j)}(e_{b_{1}b_{1}', b_{2}b_{2}'}) = \mu_{b_{1}b_{1}', b_{2}b_{2}'}^{(j)}
\end{gather*}
\noindent for $j = 1, \hdots, k$. Furthermore, let $\pi_{X}, \pi_{X'}: \cl{C}_{X_{2}, X_{1}}\rightarrow \cl{B}(\bb{C}^{k}), \pi_{A}, \pi_{A'}: \cl{C}_{A_{1}, A_{2}}\rightarrow \cl{B}(\bb{C}^{k})$ be the $*$-representations given by
\begin{gather*}
	\pi_{X}(u) = \sum\limits_{j=1}^{k}\pi_{X}^{(j)}(u)\epsilon_{jj}, \;\;\;\; \pi_{X'}(u) = \sum\limits_{j=1}^{k}\pi_{X'}^{(j)}(u)\epsilon_{jj}, 
\\ 	\pi_{A}(v) = \sum\limits_{j=1}^{k}\pi_{A}^{(j)}(v)\epsilon_{jj}, \;\;\;\; \pi_{A'}(v) = \sum\limits_{j=1}^{k}\pi_{A'}^{(j)}(v)\epsilon_{jj},
\end{gather*}
\noindent for $u \in \cl{C}_{X_{2}, X_{1}}, v \in \cl{C}_{A_{1}, A_{2}}$. Clearly, the images of $\pi_{X}, \pi_{X'}, \pi_{A}, \pi_{A'}$ are all abelian. If we then set $\xi = \sum_{j=1}^{k}\sqrt{\lambda_{j}}e_{j}\otimes e_{j}\otimes e_{j}\otimes e_{j} \in \bb{C}^{k}\otimes \bb{C}^{k} \otimes \bb{C}^{k}\otimes \bb{C}^{k}$, we have
\begin{gather*}
	\Gamma(\epsilon_{x_{2}x_{2}'}\otimes \epsilon_{y_{2}y_{2}'}\otimes \epsilon_{a_{1}a_{1}'}\otimes \epsilon_{b_{1}b_{1}'})  
	\\ = \bigg(\langle(\pi_{X}(e_{x_{2}x_{2}',x_{1}x_{1}'})\otimes \pi_{X'}(e_{y_{2}y_{2}', y_{1}y_{1}'})\otimes \pi_{A}(e_{a_{1}a_{1}', a_{2}a_{2}'})\otimes \pi_{A'}(e_{b_{1}b_{1}', b_{2}b_{2}'}))\xi, \xi\rangle\bigg)_{x_{1}x_{1}', y_{1}y_{1}'}^{a_{2}a_{2}', b_{2}b_{2}'},
\end{gather*}
\noindent with corresponding state $\tilde{s}$ given by 
\begin{gather*}
	\tilde{s}(u) = \langle (\pi_{X}\otimes \pi_{X'}\otimes \pi_{A}\otimes \pi_{A'})(u)\xi, \xi\rangle, 
\end{gather*}
\noindent for  $u \in \cl{C}_{X_{2}, X_{1}}\otimes_{\rm max} \cl{C}_{X_{2}, X_{1}}\otimes_{\rm max} \cl{C}_{A_{1}, A_{2}}\otimes_{\rm max} \cl{C}_{A_{1}, A_{2}}$. Argue now as in (iii) to conclude that tracial state $\mathsf{T}$ factors through an abelian $*$-representation of $\cl{C}_{X_{2}, X_{1}}\otimes_{\rm max} \cl{C}_{A_{1}, A_{2}}$.
\end{proof}

\begin{theorem}\label{simulated_tracial}
Let ${\rm t} \in \{\rm loc, q, qa, qc\}$. If $\Gamma \in \cl{Q}_{\rm st}$ is jointly tracial, and $\cl{E} \in \cl{Q}_{\rm t}$ is tracial, then $\Gamma[\cl{E}]$ is tracial. 
\end{theorem}
\begin{proof}
We handle the case when ${\rm t} = {\rm qc}$. By Theorem \ref{joint_tracial_strategies_thm} and \cite[Theorem 4.1]{synch_bhtt}, there exists tracial states $\mathsf{T}: \cl{C}_{X_{2}, X_{1}}\otimes_{\rm max} \cl{C}_{A_{1}, A_{2}}\rightarrow \bb{C}$ and $\tau: \cl{C}_{X_{1}, A_{1}}\rightarrow \bb{C}$ such that $\Gamma = \Gamma_{\mathsf{T}}$ and $\cl{E} = \cl{E}_{\tau}$. Let 
\begin{gather*}
	\mathsf{T}\odot \tau : \cl{C}_{X_{2}, X_{1}}\otimes_{\rm max} \cl{C}_{A_{1}, A_{2}}\otimes_{\rm max} \cl{C}_{X_{1}, A_{1}}\rightarrow \bb{C}
\end{gather*}
\noindent be the linear map given by
\begin{gather*}
	(\mathsf{T}\odot \tau)(u\otimes v\otimes w) = \mathsf{T}(u\otimes v)\tau(w), \;\;\;\; u \in \cl{C}_{X_{2}, X_{1}}, v\in \cl{C}_{A_{1}, A_{2}}, w \in \cl{C}_{X_{1}, A_{1}}.
\end{gather*}
That this is a state follows from \cite[Proposition 4.23]{takesaki}; that it is also tracial follows from a standard argument using the (automatic) continuity of the mapping on the maximal tensor product, in conjunction with calculations on the dense subset of finite linear combinations of simple tensors. Define tracial state $\tilde{\tau}$ on $\cl{C}_{X_{2}, A_{2}}$ in the following way: by Remark \ref{hom_from_max_tensor_prod}, there exists a surjective $*$-representation $\pi: \cl{C}_{X_{2}, A_{2}}\rightarrow \cl{C}_{X_{2}, X_{1}}\otimes_{\rm max} \cl{C}_{A_{1}, A_{2}}\otimes_{\rm max} \cl{C}_{X_{1}, A_{1}}$ sending $\pi(e_{x_{2}x_{2}', a_{2}a_{2}'}) = \tilde{e}_{x_{2}x_{2}', a_{2}a_{2}'}$ for $x_{2}, x_{2}' \in X_{2}, a_{2}, a_{2}' \in A_{2}$. Using $\pi$, we let
\begin{gather}\label{constructed_tracial_state}
	\tilde{\tau}(u) := ((\mathsf{T}\odot \tau) \circ \pi)(u), \;\;\;\; u \in \cl{C}_{X_{2}, A_{2}}.
\end{gather}
\noindent We note
\begin{eqnarray*}
& &
\Gamma[\cl{E}](a_{2}a_{2}', b_{2}b_{2}'|x_{2}x_{2}', y_{2}y_{2}') \\
& = &
\sum\limits_{\substack{x_{1}, x_{1}' \\ y_{1}, y_{1}'}}\sum\limits_{\substack{a_{1}, a_{1}' \\ b_{1}, b_{1}'}}\Gamma(x_{1}x_{1}', y_{1}y_{1}', a_{2}a_{2}', b_{2}b_{2}'|x_{2}x_{2}', y_{2}y_{2}', a_{1}a_{1}', b_{1}b_{1}')\cl{E}(a_{1}a_{1}', b_{1}b_{1}'|x_{1}x_{1}', y_{1}y_{1}') \\
& = &
\sum\limits_{\substack{x_{1}, x_{1}' \\ y_{1}, y_{1}'}}\sum\limits_{\substack{a_{1}, a_{1}' \\ b_{1}, b_{1}'}}(\mathsf{T}\odot \tau)(e_{x_{2}x_{2}', x_{1}x_{1}'}e_{y_{2}'y_{2}, y_{1}'y_{1}}\otimes e_{a_{1}a_{1}', a_{2}a_{2}'}e_{b_{1}'b_{1}, b_{2}'b_{2}}\otimes f_{x_{1}x_{1}', a_{1}a_{1}'}f_{y_{1}'y_{1}, b_{1}'b_{1}}) \\
& = &
(\mathsf{T}\odot \tau)(\pi(e_{x_{2}x_{2}', a_{2}a_{2}}e_{y_{2}'y_{2}, b_{2}'b_{2}})) \\
& = &
\tilde{\tau}(e_{x_{2}x_{2}', a_{2}a_{2}'}e_{y_{2}'y_{2}, b_{2}'b_{2}}),
\end{eqnarray*}
\noindent for $x_{2}, x_{2}', y_{2}, y_{2}' \in X_{2}, a_{2}, a_{2}', b_{2}, b_{2}' \in A_{2}$. This shows $\Gamma[\cl{E}] = \Gamma_{\tilde{\tau}}$, and so by \cite[Theorem 4.10]{synch_bhtt} once more we conclude that $\Gamma[\cl{E}]$ is tracial. For ${\rm t} \in \{\rm loc, q, qa\}$, the argument is similar once we use Theorem \ref{joint_tracial_strategies_thm} (ii)-(iv). 
\end{proof}


For finite sets $X, A$, let
\begin{gather*}
	E_{XA} = (e_{x, x', a, a'})_{x, x', a, a'}, \;\;\;\; E_{XA}^{\rm op} = (e_{x', x, a', a}^{\rm op})_{x, x', a, a'},
\end{gather*}
\noindent be considered as elements in $M_{XA}\otimes \mathfrak{C}_{X, A}$ and $M_{XA}\otimes \mathfrak{C}_{X, A}^{\rm op}$, respectively. If $P \in \cl{P}_{XX}, Q \in \cl{P}_{AA}$, define a linear map
\begin{gather}\label{alg_map}
	\gamma_{P, Q}: M_{XX}\otimes M_{AA}\otimes \mathfrak{C}_{X, A}\otimes \mathfrak{C}_{X, A}^{\rm op} \rightarrow \mathfrak{C}_{X, A}
\end{gather}
\noindent by setting
\begin{gather*}
	\gamma_{P, Q}(\omega\otimes u\otimes v^{\rm op}) = {\rm Tr}(\omega(P\otimes Q))uv, \;\;\;\; \omega \in M_{XX}\otimes M_{AA}, u, v \in \mathfrak{C}_{X, A}. 
\end{gather*}
\noindent If $\varphi: \cl{P}_{XX}\rightarrow \cl{P}_{AA}$ is any quantum non-local game, let
\begin{gather*}
	\mathfrak{J}(\varphi) = \bigg\langle \gamma_{P, \varphi(P)^{\perp}}(E_{XA}\otimes E_{XA}^{\rm op}): \; P \in \cl{P}_{XX}\bigg\rangle
\end{gather*}
\noindent be the generated $*$-ideal in $\mathfrak{C}_{X, A}$, and $J(\varphi)$ be the corresponding ideal considered in $\cl{C}_{X, A}$. Finally, write $\mathfrak{C}(\varphi) = \mathfrak{C}_{X, A}/\mathfrak{J}(\varphi)$ (resp. $\cl{C}(\varphi) = \cl{C}_{X, A}/J(\varphi)$) for the quotient $*$-algebra (resp. quotient ${\rm C}^{*}$-algebra). Perfect strategies for a concurrent quantum game $\varphi$ were shown to correspond to tracial states acting on $*$-representations of $\cl{C}_{X, A}$ which annihilate $\mathfrak{J}(\varphi)$ or $J(\varphi)$ in \cite{synch_bhtt, bhtt}. In light of previous results, we may conclude by giving an algebraic characterization of our transfer of perfect strategies between quantum games.
\begin{theorem}
Let $X_{i}, A_{i}, i = 1, 2$ be finite sets, $P_{i} \in M_{X_{i}X_{i}}, Q_{i} \in M_{A_{i}A_{i}}, i = 1, 2$ be projections, and ${\rm t} \in \{\rm loc, q, qa, qc\}$. If $\overline{\cl{U}}_{P_{1}, Q_{1}}\rightarrow_{\rm st} \cl{U}_{P_{2}, Q_{2}}$ via jointly tracial $\Gamma_{\mathsf{T}}$, then for any tracial state $\tau$ on $\cl{C}(\varphi_{P_{1}\rightarrow Q_{1}})$, the tracial state $\tilde{\tau}$ given in (\ref{constructed_tracial_state}) restricts to a tracial state on $\cl{C}(\varphi_{P_{2}\rightarrow Q_{2}})$.  
\end{theorem}
\begin{proof}
As corresponding QNS correlation $\cl{E}_{\tau}$ is both tracial and a perfect strategy for $\varphi_{P_{1}\rightarrow Q_{1}}$ (via \cite[Corollary 4.4]{synch_bhtt}), by Theorem \ref{simulated_tracial} we know $\Gamma_{\mathsf{T}}[\cl{E}_{\tau}]$ is tracial, and corresponds to tracial state $\tilde{\tau}$ given in (\ref{constructed_tracial_state}). Furthermore, by Theorem \ref{perfect_strats_imp_games} $\Gamma_{\mathsf{T}}[\cl{E}_{\tau}]$ is a perfect strategy for $\varphi_{P_{2}\rightarrow Q_{2}}$; thus, $\tilde{\tau}$ annihilates $J(\varphi_{P_{2}\rightarrow Q_{2}})$ as claimed. 
\end{proof}

\begin{remark}
\rm We wish to comment on the specific case when both games are the quantum graph isomorphism game. Let $X_{i} = A_{i}, i = 1, 2$, let $\cl{U}_{i} \subseteq \bb{C}^{X_{i}}\otimes \bb{C}^{X_{i}}, i = 1, 2$ be quantum graphs, and set $P_{i} := P_{\cl{U}_{i}}, i = 1, 2$. 

Recall, that, for finite set $X$, the ${\rm C}^{*}$-algebra of the free unitary quantum group $C(\cl{U}_{X}^{+})$ is the universal unital ${\rm C}^{*}$-algebra generated by the entries $u_{x, a}$ of an $X\times X$ matrix $U = (u_{x, a})_{x, a \in X}$ under the condition that $U$ and $U^{\rm t}$ are unitary (known otherwise as a bi-unitary matrix). Subalgebra $C(\bb{P}\cl{U}_{X}^{+})$ is generated by length two words of the form $u_{x, a}^{*}u_{x', a'}, x, x', a, a' \in X$- and was shown in \cite[Corollary 4.1]{banica} to be the proper quantization of the automorphism group of $M_{X}$. In \cite[Theorem 6.7]{bhtt}, concurrent bicorrelations of types ${\rm t}\in \{\rm loc, q, qc\}$ (and therefore, strategies for the quantum graph isomorphism game) were shown to be in correspondence with different tracial states on $C(\bb{P}\cl{U}_{X}^{+})$.

Abusing notation, for $S, T \in M_{XX}$ let 
\begin{gather*}
	\gamma_{S, T} : M_{XX}\otimes C(\bb{P}\cl{U}_{X}^{+}) \otimes M_{XX}\otimes C(\bb{P}\cl{U}_{X}^{+})^{\rm op} \rightarrow C(\bb{P}\cl{U}_{X}^{+})
\end{gather*}
\noindent be defined as in (\ref{alg_map}). We let $\tilde{U} = (u_{x, x', a, a'})_{x, x', a, a'} \in M_{XX}(C(\bb{P}\cl{U}_{X}^{+}))$, and for $P \in {\rm Proj}(M_{XX})$ set
\begin{gather*}
	\cl{I}_{P, P} := \bigg\langle \gamma_{P, P^{\perp}}(\tilde{U}\otimes \tilde{U}^{\rm op}), \gamma_{P^{\perp}, P}(\tilde{U}\otimes \tilde{U}^{\rm op})\bigg\rangle.
\end{gather*}
\noindent Finally, set $\cl{A}_{P, P} := C(\bb{P}\cl{U}_{X}^{+})/\cl{I}_{P, P}$. In \cite[Remark 7.12]{bhtt}, it was shown that ${\rm C}^{*}$-algebra $\cl{A}_{P, P}$ can be endowed with a natural co-associative comultiplication $\Delta_{P}: \cl{A}_{P, P}\rightarrow \cl{A}_{P, P}\otimes \cl{A}_{P, P}$; this means $\cl{A}_{P, P}$ can be viewed as a compact quantum group. By \cite[Theorem 7.10]{bhtt} and Theorem \ref{graph_isomorphism_game}, we see that the existence of a quantum hypergraph isomorphism $\Gamma$ between $\overline{\cl{U}}_{P_{1}, P_{1}}$ and $\cl{U}_{P_{2}, P_{2}}$ ensures a way to construct $*$-representations of compact quantum group $\cl{A}_{P_{2}, P_{2}}$ from $*$-reps of $\cl{A}_{P_{1}, P_{1}}$; that is, a way to transfer quantum automorphisms between quantum graphs. 
\end{remark}



\begin{thebibliography}{CS79}

\bibitem{banica}
{\scshape Banica, T.}, Symmetries of a generic coaction,
{\it Math. Ann.} {\bf 314} (1999), no. 4, 763-780. \mrev{1709109} (2001g:46146), \zbl{0928.46038}, \doi{10.1007/s002080050315}.

\bibitem{bell}
{\scshape Bell, J. S.}, 
On the Einstein  Podolsky Rosen paradox, 
{\it Phys. Phys. Fiz.} {\bf 1} (1964), no. 3, 195-200. \mrev{MR3790629}, \zbl{0152.23605}, \doi{10.1103/PhysicsPhysiqueFizika.1.195}.

\bibitem{bm}
{\scshape Blecher, D.; Le-Merdy, C.},
Operator algebras and their modules--- an operator space approach. ,
{\it The Clarendon Press, Oxford University Press, Oxford}, 2004. x+387pp. ISBN: 0-19-852659-8 \mrev{2111973} (2006a:46070), \zbl{1061.47002}, \doi{10.1093/acprof:oso/9780198526599.001.0001}. 

\bibitem{bks}
{\scshape Bochniak, A.; Kasprzak, P.; Sołtan, P. M.},
Quantum correlations on quantum spaces,
{\it Int. Math. Res. Not. IMRN}  (2023), no.14, 12400-12440. \mrev{4615234}, \zbl{1520.58002}, \doi{10.1093/imrn/rnac139}.

\bibitem{bhinw}
{\scshape Brannan, M.; Hamidi, M.; Ismert, L.; Nelson, B.; Wasilewski, M.},
Quantum edge correspondences and quantum Cuntz-Krieger algebras,
{\it J. London Math. Soc. (2)} {\bf 107} (2023), no. 3, 886-913. \mrev{4555986}, \zbl{1531.46031}, \doi{10.1112/jlms.12702}.

\bibitem{synch_bhtt}
{\scshape Brannan, M.; Harris, S.J.; Todorov, I.G.; Turowska, L.},
Synchronicity for quantum non-local games,
{\it J. Funct. Anal.} {\bf 284} (2023), no. 2, Paper No. 109738, 54 pp. \mrev{4507619}, \zbl{1515.81117}, \doi{10.1016/j.jfa.2022.109738}.

\bibitem{bhtt}
{\scshape Brannan, M.; Harris, S.J.; Todorov, I.G.; Turowska, L.},
Quantum no-signalling bicorrelations,
{\it Adv. Math.} {\bf 449} (2024), Paper No. 109732, 81 pp. \mrev{4752739}, \zbl{1553.46073}, \doi{10.1016/j.aim.2024.109732}.

\bibitem{bct}
{\scshape Brassard, G.; Cleve, R.; Tapp, A.},
Cost of exactly simulating quantum entanglement with classical communication,
{\it Phys. Rev. Lett.} {\bf 83} (1999). \doi{10.1103/PhysRevLett.83.1874}.

\bibitem{bo}
{\scshape Brown, N. P.; Ozawa, N.},
C*-algebras and finite-dimensional approximations. Graduate Studies in Mathematics, 88. 
{\it American Mathematical Society, Providence, RI}, 2008. xvi+509 pp. ISBN: 978-0-8218-4381-9; 0-8218-4381-8 \mrev{2391387} (2009h:46101), \zbl{1160.46001}, \doi{10.1090/gsm/088}.

\bibitem{choi}
{\scshape Choi, M. D.},
Completely positive linear maps on complex matrices,
{\it Linear Algebra Appl.} {\bf 10} (1975), 285-290. \mrev{376726} (51 \#12901), \zbl{0327.15018}, \doi{10.1016/0024-3795(75)90075-0}.

\bibitem{cjpp}
{\scshape Cooney, T.; Junge, M.; Palazuelos, C.; Pérez-García, D.},
Rank-one quantum games,
{\it Comput. Complexity} {\bf 24} (2015), no. 1, 133-196. \mrev{3320304}, \zbl{1328.81063}, \doi{10.1007/s00037-014-0096-x}.

\bibitem{clmw}
{\scshape Cubitt, T. S.; Leung, D.; Matthews, W.; Winter, A.}, 
Zero-error capacity and simulation assisted by non-local correlations,
{\it IEEE Trans. Inform. Theory} {\bf 57} (2011), no. 8, 5509-5523. \mrev{2849372} (2012e:81041), \zbl{1365.81024}, \doi{10.1109/TIT.2011.2159047}.

\bibitem{daws}
{\scshape Daws, M.},
Quantum graphs: different perspectives, homomorphisms, and quantum automorphisms,
{\it Commun. Am. Math. Soc.} {\bf 4} (2024), 117-181. \mrev{4706978}, \zbl{1554.46049}, \doi{10.1090/cams/30}.

\bibitem{douglas}
{\scshape Douglas, R. G.},
On majorization, factorization, and range inclusion of operators on Hilbert space,
{\it Proc. Am. Math. Soc.} {\bf 17} (1966), 413-415. \mrev{203464} (34 \#3315), \zbl{0146.12503}, \doi{10.2307/2035178}.

\bibitem{dw}
{\scshape Duan, R.; Winter, A.}, 
No-signalling assisted zero-error capacity of quantum channels and an information theoretic interpretation of the 
Lov\'asz number,
{\it IEEE Trans. Inform. Theory} {\bf 62} (2016), no. 2, 891-914. \mrev{3455905}, \zbl{1359.81062}, \doi{10.1109/TIT.2015.2507979}.

\bibitem{dpp}
{\scshape Dykema, K.; Paulsen, V. I.; Prakash, J.},
Non-closure of the set of quantum correlations via graphs,
{\it Comm. Math. Phys.} {\bf 365} (2019), no. 3, 1125-1142. \mrev{3916992}, \zbl{1408.81008}, \doi{10.1007/s00220-019-03301-1}.

\bibitem{fp}
{\scshape Farenick, D.; Paulsen, V. I.},
Operator system quotients of matrix algebras and their tensor products,
{\it Math. Scand.} {\bf 111} (2012), no. 2, 210-243. \mrev{3023524}, \zbl{1273.46038}, \doi{10.7146/math.scand.a-15225}.

\bibitem{fritz}
{\scshape Fritz, T.},
Tsirelson's problem and Kirchberg's conjecture,
{\it Rev. Math. Phys.} {\bf 24} (2012), no. 5, 1250012, 67 pp. \mrev{2928100}, \zbl{1250.81023}, \doi{10.1142/S0129055X12500122}.

\bibitem{hmps}
{\scshape Helton, J. W.; Meyer, K. P.; Paulsen, V. I.; Satriano, M.},
Algebras, synchronous games, and chromatic numbers of graphs,
{\it New York J. Math.} {\bf 25} (2019), 328-361. \mrev{3933765}, \zbl{1410.05069}.

\bibitem{gt_one}
{\scshape Hoefer, G.; Todorov, I. G.}, 
Quantum hypergraph homomorphisms and non-local games,
{\it Dissertaciones Math.} {\bf 588} (2023), 1-54. \mrev{4698861}, \zbl{1539.81030}, \doi{10.4064/dm230309-14-11}.

\bibitem{gt_two}
{\scshape Hoefer, G.; Todorov, I. G.},
Homomorphisms of quantum hypergraphs,
{\it J. Math. Anal. Appl.} {\bf 543} (2025), no. 2 (pt. 1), Paper No. 128907, 20. \mrev{4803784}, \zbl{1564.46057}, \doi{10.1016/j.jmaa.2024.128907}.

\bibitem{mipre}
{\scshape Ji, Z.; Natarajan, A.; Vidick, T.; Wright, J.; Yuen, H.},
MIP*=RE,
{\it preprint (2020), arXiv:2001.04383}. \zbl{1503.68075}.

\bibitem{jnppsw}
{\scshape Junge, M.; Navascues, M.; Palazuelos, C.; Perez-Garcia, D.; Scholz, V. B.; Werner, R. F.},
Connes' embedding problem and Tsirelson's problem,
{\it J. Math. Phys.} {\bf 52} (2011), no. 1, 012102, 12 pp. \mrev{2790067} (2011m:46099), \zbl{1314.81038}, \doi{10.1063/1.3514538}.

\bibitem{kr}
{\scshape Kadison, R. V.; Ringrose, J. R.},
Fundamentals of the theory of operator algebras. {V}ol. {I}. {E}lementary theory. Graduate Studies in Mathematics, 16.
{\it American Mathematical Society, Providence, RI}, 1997. pp. i-xxii and 399-1074. ISBN: 0-8218-0820-6 \mrev{1468230} (98f:46001a), \zbl{0831.46060}, \doi{10.1090/gsm/015}.

\bibitem{kavruk}
{\scshape Kavruk, A. S.},
Nuclearity related properties in operator systems, 
{\it J. Operator Theory} {\bf 71} (2014), no. 1, 95-156. \mrev{3173055}, \zbl{1349.46060}, \doi{10.7900/jot.2011nov16.1977}.

\bibitem{kavruk_thesis}
{\scshape Kavruk, A. S.},
Tensor products of operator systems and applications. Thesis (Ph.D.)-University of Houston.
{\it ProQuest LLC, Ann Arbor, MI}, 2011. 87 pp. ISBN: 978-1267-11028-2. \mrev{2992647}.
 

\bibitem{kptt}
{\scshape Kavruk, A. S.; Paulsen, V. I.; Todorov, I. G.; Tomforde, M.},
Tensor products of operator systems,
{\it J. Funct. Anal.} {\bf 261} (2011), no. 2, 267-299. \mrev{2793115} (2012h:46094), \zbl{1235.46051}, \doi{10.1016/j.jfa.2011.03.014}.

\bibitem{kw}
{\scshape Kirchberg, E.; Wassermann, S.},
$C^{*}$-algebras generated by operator systems,
{\it J. Funct. Anal.} {\bf 155} (1998), no. 2, 324-351. \mrev{1624549} (99f:46085), \zbl{0940.46038}, \doi{10.1006/jfan.1997.3226}.

\bibitem{lmprsstw}
{\scshape Lupini, M.; Man\v{c}inska, L.; Paulsen, V. I.; Roberson, D. E.; Scarpa, G.; Severini, S.; Todorov, I. G.; Winter, A.},
Perfect strategies for non-local games,
{\it Math. Phys. Anal. Geom.} {\bf 23} (2020), no. 1, Paper No. 7, 31 pp. \mrev{4069856}, \zbl{1436.91021}, \doi{10.1007/s11040-020-9331-7}.

\bibitem{mr}
{\scshape Man\v{c}inska, L.; Roberson, D. E.},
Graph homomorphisms for quantum players,
{\it 9th Conference on the Theory of Quantum Computation, Communication, and Cryptography}, 212-216, LIPIcs. Leibniz Int. Proc. Inform. 27 {\it Schloss Dagstuhl. Leibniz-Zent. Inform., Wadern,} 2014. \mrev{3354695}, \zbl{1360.91048}.

\bibitem{ozawa}
{\scshape Ozawa, N.},
About the Connes' embedding conjecture: algebraic approaches,
{\it Jpn. J. Math.} 8 (2013), no. 1, 147-183. \mrev{3067294}, \zbl{1278.46053}, \doi{10.1007/s11537-013-1280-5}.


\bibitem{pa} 
{\scshape Paulsen, V. I.},
Completely bounded maps and operator algebras. Cambridge Studies in Advanced Mathematics, 78. 
{\it Cambridge University Press, Cambridge}, 2002. xii+300 pp. ISBN: 0-521-81669-6 \mrev{1976867} (2004c:46118), \zbl{1029.47003}, \doi{10.1017/CBO9780511546631}.

\bibitem{ptt}
{\scshape Paulsen, V. I.; Todorov, I. G.; Tomforde, M.},
Operator system structures on ordered spaces, 
{\it Proc. Lond. Math. Soc. (3)} {\bf 102} (2011), no. 1, 25-49. \mrev{2747723} (2012h:47151), \zbl{1213.46050}, \doi{10.1112/plms/pdq011}.

\bibitem{rv}
{\scshape Regev, O.; Vidick, T.},
Quantum XOR games,
{\it ACM Trans. Comput. Theory} {\bf 7} (2015), no. 4, Art. 15, 43. \mrev{3414675}, \zbl{1348.81177}, \doi{10.1145/2799560}.

\bibitem{slofstra_one}
{\scshape Slofstra, W.},
The set of quantum correlations is not closed,
{\it Forum Math. Pi} {\bf 7} (2019), e1, 41. \mrev{3898717}, \zbl{1405.81021}, \doi{10.1017/fmp.2018.3}.

\bibitem{slofstra_two}
{\scshape Slofstra, W.},
Tsirelson's problem and an embedding theorem for groups arising from non-local games,
{\it J. Amer. Math. Soc.} {\bf 33} (2020), no. 1, 1-56. \mrev{4066471}, \zbl{1480.20083}, \doi{10.1090/jams/929}.

\bibitem{stahlke}
{\scshape Stahlke, D.},
Quantum zero-error source-channel coding and non-commutative graph theory,
{\it IEEE Trans. Inform. Theory} {\bf 62} (2016), no. 1, 554-577. \mrev{3447998}, \zbl{1359.94454}, \doi{10.1109/TIT.2015.2496377}.

\bibitem{takesaki}
{\scshape Takesaki, M.},
Theory of operator algebras. I.,
{\it Springer-Verlag, New York-Heidelberg}, 1979. vii+415pp. ISBN: 0-387-90391-7  \mrev{1873025} (2002m:46083), \zbl{0990.46034}, \doi{10.1007/978-1-4612-6188-9}.

\bibitem{tt}
{\scshape Todorov, I. G.; Turowska, L.}, Quantum no-signalling correlations and non-local games,
{\it Comm. Math. Phys.} {\bf 405} (2024), no. 6, Paper No. 141, 65 pp. \mrev{4751723}, \zbl{07868800}, \doi{10.1007/s00220-024-05001-x}.


\bibitem{watrous}
{\scshape Watrous, J.},
The theory of quantum information,
{\it Cambridge University Press, Cambridge}, (2018), viii+590 pp. ISBN: 978-1-107-18056-7. \zbl{1393.81004}, \doi{10.1017/9781316848142}.

\bibitem{zettl}
{\scshape Zettl, H.},
A characterization of ternary rings of operators,
{\it Adv. in Math.} {\bf 48} (1983), no. 2, 117-143. \mrev{700979} (84h:46093), \zbl{0517.46049}, \doi{10.1016/0001-8708(83)90083-X}.


\end{thebibliography}
\end{document}